\documentclass[12pt]{article}
\usepackage{amsmath}
\usepackage{graphicx}
\usepackage{enumerate}
\usepackage{natbib}
\usepackage{url} 

\usepackage{usebib}
\usepackage{blkarray}
\usepackage{array}
\usepackage{amsmath}
\usepackage{amssymb}
\usepackage{amsthm}
\usepackage{natbib}
\usepackage{amsfonts}
\RequirePackage[colorlinks,citecolor=blue,urlcolor=blue]{hyperref}
\usepackage{graphicx}
\usepackage{framed}
\usepackage[normalem]{ulem}
\usepackage{enumerate}
\usepackage{enumitem}
\usepackage[utf8]{inputenc}
\usepackage{bbm}
\usepackage{mathtools}
\usepackage{latexsym}
\usepackage{epsfig}
\usepackage{color}
\usepackage{xparse}
\usepackage{float}
\usepackage{lscape}
\usepackage{subfigure}
\usepackage{MnSymbol}
\usepackage{fancyhdr}
\usepackage{lastpage}
\usepackage{url}
\usepackage{mathrsfs}
\usepackage{cleveref}
\usepackage{multirow}
\usepackage{appendix}
\usepackage[font = scriptsize]{caption}

\DeclareSymbolFont{AMSb}{U}{msb}{m}{n}  
\DeclareMathSymbol{\Sph}{\mathbin}{AMSb}{"53} \DeclareMathSymbol{\R}{\mathbin}{AMSb}{"52}
\DeclareMathSymbol{\T}{\mathbin}{AMSb}{"54} \DeclareMathSymbol{\Z}{\mathbin}{AMSb}{"5A}
\DeclareMathSymbol{\K}{\mathbin}{AMSb}{"4B} \DeclareMathSymbol{\Pb}{\mathbin}{AMSb}{"50}
\DeclareMathSymbol{\Q}{\mathbin}{AMSb}{"51}
\DeclareMathSymbol{\N}{\mathbin}{AMSb}{"49}
\DeclareMathOperator{\vect}{vec}

\newcommand{\abs}[1]{\ensuremath{\left| #1 \right|}}

\newcommand{\lp}{\ensuremath{\left(}}
\newcommand{\rp}{\ensuremath{\right)}}
\newcommand{\lb}{\ensuremath{\left[}}
\newcommand{\rb}{\ensuremath{\right]}}
\newcommand{\vc}[1]{\ensuremath{\vect\lp #1 \rp}}

\DeclareMathOperator{\rank}{rank}

\newcommand{\ind}[1]{\ensuremath{\mathbbm{1}\lp#1\rp}}
\newcommand{\deter}[1]{\ensuremath{\det\lp#1\rp}}

\NewDocumentCommand\Prob{mg}{\ensuremath{\Pb \IfNoValueTF{#2}{}{_{#2}} \lp #1 \rp}}

\NewDocumentCommand\Exp{mg}{\ensuremath{\mathbb{E} \IfNoValueTF{#2}{}{_{#2}} \lb #1 \rb}}
\newcommand{\indy}{\ensuremath{\perp\!\!\!\perp}}

\NewDocumentCommand\func{mg}{\ensuremath{f \IfNoValueTF{#2}{}{^{#2}} \lp #1 \rp}}
\NewDocumentCommand\charfunc{mg}{\ensuremath{\varphi \IfNoValueTF{#2}{}{_{#2}} \lp #1 \rp}}

\crefname{assumption}{assumption}{assumptions}
\crefname{proof}{proof}{proofs}

\usepackage{amsbsy}

\crefname{enumi}{condition}{conditions}
\Crefname{enumi}{Condition}{Conditions}

\newtheorem{theorem}{Theorem}[section]
\newtheorem{lemma}[theorem]{Lemma}

\newtheorem{assumption}[theorem]{Assumption}
\newtheorem{corollary}[theorem]{Corollary}
\newtheorem{definition}{Definition}[section]

\newcommand{\blind}{1}

\begin{document}
\bibliographystyle{agsm}

\def\spacingset#1{\renewcommand{\baselinestretch}%
{#1}\small\normalsize} \spacingset{1}


\if1\blind
{
  \title{\bf Vaccine efficacy for binary post-infection outcomes under misclassification without monotonicity}
  \author{Rob Trangucci\thanks{
    The authors gratefully acknowledge helpful conversations with Gongjun Xu, Fan Li, and Peng Ding about this work.}\\
    Department of Statistics, Oregon State University, Corvallis\\
    \and 
    Yang Chen \\
    Department of Statistics, University of Michigan, Ann Arbor\\
    \and 
    Jon Zelner \\
    Department of Epidemiology\\ Center for Social Epidemiology and Population Health\\ University of Michigan School of Public Health, Ann Arbor
    }
  \maketitle
} \fi

\if0\blind
{
  \bigskip
  \bigskip
  \bigskip
  \begin{center}
    {\LARGE\bf Identified vaccine efficacy for binary post-infection outcomes under misclassification without monotonicity}
\end{center}
  \medskip
} \fi

\bigskip
\begin{abstract}
In order to meet regulatory approval, a new vaccine must show that it reduces the total risk of a post-infection outcome like transmission, symptomatic disease, severe illness, or death in a randomized clinical trial.
Because infection is necessary for a post-infection outcome, one can use principal stratification to partition the total causal effect of vaccination into two causal effects: vaccine efficacy against infection, and the principal effect of vaccine efficacy on post-infection outcomes in patients who would be infected under both placebo and vaccination.
Despite the importance of such principal effects to policymakers, these estimands are generally unidentifiable, even under strong assumptions that are rarely satisfied in real-world trials.
We develop a novel method to nonparametrically point identify these principal effects while eliminating the monotonicity assumption and allowing for measurement error.
Furthermore, our results allow for multiple treatments, and are general enough to be applicable outside of vaccine efficacy.
Our method relies on the fact that many vaccine trials are run at geographically disparate health centers, and measure biologically-relevant categorical pretreatment covariates.
We show that our method can be applied to a variety of clinical trial settings where vaccine efficacy against infection and a post-infection outcome can be jointly inferred.
This methodology can yield new insights from existing vaccine efficacy trial data and will aid researchers in designing new multi-arm clinical trials.
\end{abstract}

\noindent%
{\it Keywords:}  3 to 6 keywords, that do not appear in the title
\vfill

\newpage
\spacingset{1.9} 

\section{Introduction}

Phase 3 randomized, placebo-controlled clinical trials are the gold-standard by which vaccine candidates are assessed for efficacy and safety.
Such trials are an important source of data about whether vaccines prevent infection, and post-infection outcomes like secondary transmission, severe illness, or death.
For example, the COVID-19 vaccination trials detailed in \cite{polackSafetyEfficacyBNT162b22020} and \cite{baden_efficacy_2021} measured vaccine efficacy against symptomatic disease, as well as severe illness and death.
Principal stratification, developed in \cite{frangakisPrincipalStratificationCausal2002}, may be used to partition the intention-to-treat effect of vaccination on an outcome like hospitalization into vaccine efficacy against infection and vaccine efficacy against hospitalization given infection in the always-infected stratum; these separate effects help policy makers optimize vaccination programs, communicate with the public, allocate scarce resources, and guide future pharmaceutical therapeutic development \citep{lipsitchInterpretingVaccineEfficacy2021}.
Methods to infer principal effects for vaccine efficacy were first developed for continuous post-infection outcomes in \cite{gilbertSensitivityAnalysisAssessment2003,jemiaiSemiparametricEstimationTreatment2007,shepherdSensitivityAnalysesComparing2006,shepherdSensitivityAnalysesComparing2007}, and further developed for binary post-infection outcomes in \cite{hudgensCausalVaccineEffects2006}.

Unfortunately, vaccine efficacy against post-infection outcomes, binary or otherwise, is not generally identifiable, even under the assumption that vaccine efficacy against infection is non-negative almost-surely (monotonicity).
Moreover, the method requires that both infection and post-infection outcomes are perfectly measured.
Neither monotonicity nor error-free measurements can be assumed to hold in vaccine trials.
Monotonicity can be violated if a vaccine increases the per-exposure probability of infection for a participant \citep{gilbertSensitivityAnalysisAssessment2003}, which is possible in influenza vaccine trials where the vaccine targets a different antigen than the circulating strain.
Another way monotonicity can be violated is if vaccination increases exposure for certain participants.
This can occur in a double-blinded placebo-controlled study where the vaccine is reactogenic and leads to some participants in the vaccine group becoming unblinded.

Measurement error is common in vaccine trials due to the imperfect nature of diagnostic tests for infection \citep{kisslerViralDynamicsAcute2021,wangDetectionSARSCoV2Different2020}.
Post-infection outcomes like symptoms may also be observed with error.
For example, in an influenza vaccine trial, many different viruses circulate during influenza season that produce similar symptom profiles.

We develop novel methodology to point identify vaccine efficacy against binary post-infection outcomes without assuming monotonicity while allowing infection and post-infection outcomes to be misclassified.
Our framework immediately generalizes to multiple treatments as we will show.
We capitalize on the fact that many randomized trials for vaccines are run as multi-center trials (i.e. geographically well-separated study sites) \citep{francis_prevention_1982,jr_estimation_2000,noauthor_quadrivalent_2007,halloran_design_2010,baden_efficacy_2021,polackSafetyEfficacyBNT162b22020}, and typically measure pretreatment covariates that can plausibly be noisy measurements of infection principal stratum.
We show that when combined with appropriate and reasonable assumptions, these extra sources of information yield point identification of the estimand of interest, VE against post-infection outcomes.

Our method builds on literature for identifying principal stratum effects with covariates
\citep{rubinCausalInferencePotential2006,dingIdentifiabilityEstimationCausal2011,jiangPrincipalCausalEffect2016},
on inferring principal stratum effects in multisite randomized trials \citep{wang2017causal,yuan_identifying_2019,luo2023causal},
on using covariates to hone large-sample nonparametric bounds
\citep{zhangEstimationCausalEffects2003,grilliNonparametricBoundsCausal2008,longSharpeningBoundsPrincipal2013},
and on identifying causal estimands under unmeasured confounding \citep{miaoIdentifyingCausalEffects2018,shiMultiplyRobustCausal2020}.
Our method also relates to recent literature on inferring causal estimands under measurement error \citep{jiangMeasurementErrorsBinary2020} and on identification of latent variable models \citep{ouyangIdentifiabilityLatentClass2022}. 

We show that our method can be used to design randomized trials for comparison of multiple vaccines, which will be a necessity for public health agencies in future pandemics.
Due to recent updates to regulatory guidance from the European Medicines Agency, the authority that authorizes pharmaceuticals in the European Union, principal effects are acceptable target estimands in randomized clinical trials and principal stratification is an acceptable analysis method for these trial data \citep{bornkamp, lipkovich_using_2022}.
This means that our methodology can be used by regulatory agencies to design new clinical trials for vaccines that target post-infection outcome VE.
As noted by several authors, vaccine efficacy against post-infection outcomes is mathematically analogous to the widely-studied survivor average treatment effects \citep{dingIdentifiabilityEstimationCausal2011,tchetgentchetgenIdentificationEstimationSurvivor2014,dingPrincipalStratificationAnalysis2017}, so our methodology can be readily used outside the domain of vaccine efficacy. 

\section{Vaccine efficacy in multi-arm multi-center trials}\label{sec:multi-arm}
A common VE study design is a multi-arm multi-center trial in which participants are enrolled at many different well-geographically-separated health centers and subsequently randomized to receive one of several treatments, one of which is usually a placebo.
In order to make the statistical model more concrete, we will focus the development of the method in this section on a hypothetical influenza vaccine trial.
Despite this focus, our method is general enough to be applicable to any vaccine-preventable disease and any post-infection outcome.
Examples of these outcomes include secondary transmission, medically-attended disease, severe symptoms, viral load.
All of these infection/post-infection-outcome pairs satisfies the basic structure we will develop in this section.

Consider the example of an influenza vaccine trial, where researchers are interested in understanding vaccine efficacy against influenza infection and vaccine efficacy against severe illness caused by influenza infection.
Crucially, it is not possible to perfectly observe influenza infection or severe illness. 
Instead, researchers are limited to using imperfect tests for infection, like polymerase chain reaction (PCR) tests, or serology to detect a participant's infection status. 
These methods measure infection with error, with varying levels of sensitivity and specificity. 
For example, PCRs for COVID-19 have very high specificity, but tend to have sensitivities in the range of $0.6$ to $0.8$ due to variation among patients in how the virus populates the nasal cavity, variation in swab quality, and viral RNA dynamics \citep{kisslerViralDynamicsAcute2021,wangDetectionSARSCoV2Different2020}.
Depending on the severity of the post-infection outcome, these outcomes may also be mismeasured.
For instance, a high proportion of participants report symptoms in vaccine efficacy studies, despite many of these participants testing negative for the target disease.
In the presence of high-sensitivity tests, this necessarily means that specificity of symptoms following infection is below $1$.
This is because it is possible for participants to develop influenza-like severe illness from non-influenza viruses during a clinical trial.
Thus our framework assumes that observed infection and severe illness are noisy proxies for true unobservable infection and severe illness states.
The next section outlines the data structure for each participant.

\subsection{Notation and assumptions}

Suppose there are $n$ participants in the trial, and we observe the following sextuplet for each participant $i$: $(\tilde{S}_i, \tilde{Y}_i, Z_i, R_i, A_i, X_i)$, 
where $\tilde{S}_i$ is binary influenza test result, $\tilde{Y}_i$ is observed binary severe illness status.
$Z_i$ is a categorical variable with $N_z \geq 2$ categories representing treatment assignment.
Let $Z_i \in \{z_1, \dots, z_{N_z}\}$.
$R_i$ is a categorical variable indicating the health center with which each participant is associated, $A_i$ is a discrete pre-treatment covariate related to infection under treatment and control, and $X_i$ is a univariate discrete pre-treatment covariate that may represent the intersection of several distinct covariates like age, sex, occupation, and pre-existing conditions. 
Let $R_i$ take values from $1$ to $N_r$, $A_i$ take values from $1$ to $N_a$, and $X_i$ take values from $1$ to $N_x$.

Let $S_i$ be the latent influenza infection state, and $Y_i$ be the latent influenza-caused severe illness state for each participant.
We use the Neyman-Rubin causal model to define $S_i$, and $Y_i$ as partially-observed realizations of counterfactual outcomes \citep{neyman1923application,rubin_estimating_1974,rubin_bayesian_1978, holland_statistics_1986}.
For an extensive review of statistical approaches to causal inference through the lens of missing data see \cite{ding2018causal}.
Let any potential treatment plan for all $n$ individuals in the trial be the length-$n$ categorical vector $\mathbf{z}$, where the $i^\mathrm{th}$ element is the potential treatment of the $i^\mathrm{th}$ participant.
Accordingly, each individual is associated with a length-$N_z$ counterfactual infection outcome, $S_i(\mathbf{z})$, and a counterfactual severe illness outcome, $Y_i(\mathbf{z}, S_i(\mathbf{z}))$, under treatment status $\mathbf{z}$.
Let the observed treatment status for all $n$ individuals in the trial be the length-$n$ categorical vector $\mathbf{Z}$, where the $i^\mathrm{th}$ element is the assigned treatment of the $i^\mathrm{th}$ participant.

Our causal model enforces the constraint that an individual not infected by influenza cannot have severe illness caused by influenza infection.
In other words, post-infection outcomes are defined such that they are \emph{caused} by infection from a pathogen of interest \citep{gilbertSensitivityAnalysisAssessment2003,hudgensCausalVaccineEffects2006}.
Then $Y_i(\mathbf{z},0)$ is undefined for all $\mathbf{z}$, and is denoted as $Y_i(\mathbf{z}, 0) = \star$.
$Y_i(\mathbf{z},S_i(\mathbf{z}))$ is defined as a binary variable only when $S_i(\mathbf{z}) = 1$, or, equivalently, $Y_i(\mathbf{z},1)$.
For the remainder of the paper we assume that $S_i(\mathbf{z}) = S_i(\mathbf{z}^\prime)$ and $Y_i(\mathbf{z},S_i(\mathbf{z})) = Y_i(\mathbf{z}^\prime,S_i(\mathbf{z}^\prime))$ if $\mathbf{z}_i = \mathbf{z}^\prime_i$.
Therefore, we assume the Stable Unit Treatment Value Assumption (SUTVA) holds:
\begin{assumption}[SUTVA]
\label{cond:SUTVA}
There is only one version of each treatment, and counterfactual outcomes are a function of only a unit's respective treatment status, $z$.
\end{assumption}

SUTVA can be satisfied for vaccine efficacy trials by restrictions on
participants and recruitment \citep{gilbertSensitivityAnalysisAssessment2003}. 
Furthermore, recruited participants are a small fraction of the total population at risk of infection \citep{zhangLikelihoodBasedAnalysisCausal2009}.
Thus the length $2\,N_z$ vector $S_i(z_1), Y_i(z_1,S(z_1)), \dots, S_i(z_{N_z}), Y_i(z_{N_z},S(z_{N_z}))$ is the complete definition of counterfactual outcomes under each treatment $z_j$ for each individual in the trial.

We assume that the study is a randomized experiment. 
This means that all trial participants have positive probabilities of being assigned to any treatment, and that treatment assignment is unconfounded \citep{imbens2015causal}.

\begin{assumption}[Random treatment assignment]
  \label{cond:pc-obs-unconfound-multi-z}
  The probability of being assigned to treatment for each individual lies strictly between $0$ and $1$: $$0 < P(Z_i = z_j \mid S_i(z_1), Y_i(z_1,S(z_1)), \dots, S_i(z_{N_z}), Y_i(z_{N_z},S(z_{N_z}))) < 1$$ for all
  $z_j \in \{z_1, \dots, z_{N_z}\}$.\\
  Treatment assignment is independent of all potential outcomes, or $$S_i(z_1), Y_i(z_1,S(z_1)), \dots, S_i(z_{N_z}), Y_i(z_{N_z},S(z_{N_z})) \indy Z_i.$$ 
\end{assumption}

Given \Cref{cond:SUTVA} and \Cref{cond:pc-obs-unconfound-multi-z}, the latent realized values of the counterfactual variables are as follows:
\begin{align}
  \label{eq:observational-model}
  S_i  = \textstyle\sum_{j=1}^{N_z} S_i(z_j) \ind{Z_i = z_j},  \quad
  Y_i  = \textstyle\sum_{j=1}^{N_z}Y_i(z_j,S_i(z_j)) \ind{Z_i=z_j}.
\end{align}

Let principal stratum, $S_i^{P_0}$, be defined as the ordered $N_z$-vector of counterfactual infection outcomes for unit $i$, or 
\begin{definition}[Principal stratum]
  $$S^{P_0}_i = (S_i(z_1),S_i(z_2), \dots, S_i(z_{N_z})),\, S_i(z_j) \in \{0,1\},\, 1 \leq j \leq N_z.$$
\end{definition}
Let the set of all principal strata be denoted as $\mathcal{S}$.
When the set of principal strata is not restricted $\mathcal{S} \equiv \{0,1\}^{N_z}$ and $\abs{\mathcal{S}} = 2^{N_z}$.
Let $u$ be an element of $\mathcal{S}$, and let $u_j \equiv P(S_i(z_j)=1 \mid S^{P_0}_i = u)$.

A third condition that will be assumed to hold is that of non-differential measurement errors for the influenza test result and the severe illness observation.
\begin{assumption}[Non-differential Misclassification Errors]\label{cond:nondiff}
  Misclassification errors for $\tilde{S}_i, \tilde{Y}_i$ are conditionally independent of all else given the true values $S_i, Y_i$ or 
  $$\tilde{S}_i \indy Z_i, S^{P_0}_i, R_i, A_i, Y_i, X_i \mid S_i,\quad\tilde{Y}_i \indy Z_i, S^{P_0}_i, R_i, A_i, S_i, X_i \mid Y_i.$$
\end{assumption}
Under \Cref{cond:nondiff}, we may completely characterize the distributions of the noisy outcomes $\tilde{S}_i, \tilde{Y}_i$ via the following four unknown parameters $\mathrm{sn}_S = P(\tilde{S}_i = 1 \mid S_i = 1), \mathrm{sp}_S = P(\tilde{S}_i = 0 \mid S_i = 0)$ and $\mathrm{sn}_Y = P(\tilde{Y}_i = 1 \mid Y_i = 1), \mathrm{sp}_Y = P(\tilde{Y}_i = 0 \mid Y_i = 0)$, or the respective sensitivities and specificities for infection and the post-infection outcome.

\Cref{cond:nondiff} may not be a reasonable assumption for some combinations of infection test and post-infection outcome.
For example, a PCR test may have higher sensitivity if a patient is infected and also exhibiting respiratory symptoms; this would invalidate the conditional independence of $\tilde{S}_i$ and $Y_i$ given $S_i$.
An example where this assumption would reasonably hold is in a secondary transmission study of dyads where the post-infection outcome is the infection status of an unvaccinated partner in a dyad.
We discuss how this assumption can be loosened for infection misclassification in \Cref{sec:discuss-ve}.

We can thus define several causal estimands of interest related to the latent infection and severe illness states, $S_i(z_j), Y_i(z_j)$, comparing treatment $z_j$ to treatment $z_k$.
\begin{definition}[Vaccine efficacy against infection]\label{def:ve-infection}
  \[\mathrm{VE}_{S,jk} =\Exp{S_i(z_k) - S_i(z_j)}/\Exp{S_i(z_k)},\,\mathrm{and}\] 
\end{definition}
\begin{definition}[Intention-to-treat vaccine efficacy against severe illness]\label{def:itt-post}
  \[\mathrm{VE}_{ITT,jk} = \Exp{Y_i(z_k,S_i(z_k))S_i(z_k) - Y_i(z_j,S_i(z_j))S_i(z_j)}/\Exp{Y_i(z_k,S_i(z_k))S_i(z_k)}.\]
\end{definition}
where we let $Y_i(z_j,S_i(z_j))S_i(z_j) = 0$ when $S_i(z_j) = 0$ and similarly for $z_k$.
By convention, $z_k$ is set to be the placebo in the above expressions.
When $N_z=2$, the subscript $jk$ may be dropped from $\mathrm{VE}_{S,jk}$ because there are only two treatment arms being compared:
$$\mathrm{VE}_S = \Exp{S_i(0) - S_i(1)}/\Exp{S_i(0)}.$$
$\mathrm{VE}_S$ can be thought of as the percent change in the risk of infection conferred by vaccination.

Next we examine how to use principal stratification to formulate causal estimands that condition on infection \citep{frangakisPrincipalStratificationCausal2002}.

\subsection{Conditional-on-infection effects and principal stratification}

Despite the conceptual simplicity of comparing rates of severe illness in between treatment arms among infected participants, we will show these comparisons are not valid causal quantities.
In a two arm trial, consider the estimand represented as $\mathrm{VE}_{I}^{\mathrm{net}} = \frac{\Exp{Y_i(0) \mid S_i(0) = 1} - \Exp{Y_i(1) \mid S_i(1) = 1}}{\Exp{Y_i(0) \mid S_i(0) = 1}}$ \citep{hudgensCausalVaccineEffects2006}.
As shown in \cite{frangakisPrincipalStratificationCausal2002,hudgensCausalVaccineEffects2006}, the set of participants $\{i \mid S_i = 1, Z_i = 1\}$ is different from the set of participants $\{i \mid S_i = 1, Z_i = 0\}$, which invalidates the contrast as a causal quantity. 
For a comparison to have a causal interpretation, the \emph{only} systematic difference between the quantities being compared may be the difference in treatment, as in the numerator for $\mathrm{VE}_S$: $\Exp{S_i(0) - S_i(1)}$.
We show in the Appendix that the numerator of $\mathrm{VE}_I^{\mathrm{net}}$ contains a true causal estimand $\Exp{Y_i(0,1) - Y_i(1,1) \mid S_i^{P_0} = (1,1)}$, but is contaminated by differences in baseline risk of severe illness under treatment and vaccination among principal strata.

The $(1, 1)$ stratum, otherwise known as the always-infected stratum, is the only stratum in which a causal comparison can be made.
This is because it is the only stratum in which participants have a well-defined post-infection outcome under vaccination and under placebo.
Thus, as in \cite{hudgensCausalVaccineEffects2006}, let the vaccine efficacy against severe illness when $N_z=2$ be defined as:
\begin{definition}[Vaccine efficacy against post-infection outcome]\label{defn:ve-i}
  \[\mathrm{VE}_{I} = \Exp{Y_i(0) - Y_i(1) \mid S^{P_0}_i = (1,1)}/\Exp{Y_i(0) \mid S^{P_0}_i = (1,1)}.\]
\end{definition}
$\mathrm{VE}_{I}$ is a principal effect as defined in \cite{frangakisPrincipalStratificationCausal2002} because it is conditional on a principal stratum.

In the two-arm trial, only one stratum admits a well-defined post-infection causal estimand, but when $N_z > 2$ principal effects can be constructed in several strata.
Let the set of principal strata that admit well-defined principal effects be:
$$
\mathcal{U} \equiv \{u \in \mathcal{S} \mid \textstyle{\sum}_{\ell=1}^{2^{N_z}} u_\ell > 1\}.
$$
Then for any $u \in \mathcal{U}$ we may define the following principal effect for any two treatments $z_j$ and $z_k$ such that $u_j u_k = 1$:
\begin{definition}[Vaccine efficacy against post-infection outcome $Y$ for multiarm studies]\label{defn:ve-i-multi}
  \[\mathrm{VE}^u_{I,jk} = \Exp{Y_i(z_k) - Y_i(z_j) \mid S^{P_0}_i = u}/\Exp{Y_i(z_k) \mid S^{P_0}_i = u}.\]
\end{definition}
Because there are now several principal strata in which a comparison between  $z_j$ and $z_k$ can be made, we add a superscript $u$ to denote the stratum in which the comparison is being made.

To give a concrete example about how one might use the expanded definition of vaccine efficacy against severe illness, we will use \cite{montoComparativeEfficacyInactivated2009} as an example.
\cite{montoComparativeEfficacyInactivated2009} treated the four-arm trial as a three-arm trial by combining the two separate placebo arms into one unified arm.
Given that both placebo arms received inert treatments, albeit via different routes of administration, this is a reasonable assumption.
The aim of the trial was to measure the absolute and relative efficacies against symptomatic influenza; thus, it may be of interest to infer the relative efficacy against severe illness given influenza infection for the two competing vaccines.
The following causal estimand captures this effect, where $z_1$ is the placebo, $z_2$ is the live-attenuated vaccine, and $z_3$ is the inactivated vaccine:
\begin{equation}\label{eq:ve-ests}
  \frac{\Exp{Y_i(z_2)\mid S^{P_0}_i = (1,1,1)} - \Exp{Y_i(z_3)\mid S^{P_0}_i = (1,1,1)}}{\Exp{Y_i(z_1)\mid S^{P_0}_i = (1,1,1)}} = \mathrm{VE}^{(1,1,1)}_{I,31} - \mathrm{VE}^{(1,1,1)}_{I,21}.
\end{equation}

The fundamental problem of causal inference \citep{holland_statistics_1986}, namely that we observe only one counterfactual outcome for each individual, prevents the development of a simple ratio estimator.
Indeed, \cite{hudgensCausalVaccineEffects2006} show that $\mathrm{VE}_I$ is not identifiable.
In the next section, we will show how taking into account the extra information available in multisite randomized VE trials may yield identifiability.

\subsection{Incorporating study-site and covariate information} \label{subsec:nonmono-multiple-groups}


Enrollment criteria for vaccine efficacy trials typically stipulate that patients with a history of infection or a recent infection with the pathogen of interest are excluded from the trial. 
Furthermore, multisite studies typically rely on study sites themselves to recruit patients for the trial, which mean that site-specific full patient accrual may occur at different times \citep{weinberger_multisite_2001}.
Because enrollment is conditional on the patient-specific lack of infection and intersite patient enrollment may systematically differ, patients at different study sites can have different times at risk for infection prior to the study.
To the extent there is patient-level frailty in time to infection, this can lead to patients who are enrolled later in the study as being less frail than those enrolled earlier.
It is reasonable to assume that patient frailty is related to principal strata, and thus it is reasonable to assume that study site is not independent of principal strata given covariates.

Another reason to expect that the distribution of principal strata is not independent of study site is variation in exposure to the pathogen of interest between study sites.
Variation in disease prevalence during the trial can lead to variation in rates of infection, and, subsequently, principal strata at the study site.
Multi-scale spatial variation of disease prevalence is a hallmark of infectious disease \citep{bauerStratifiedSpaceTime2018}.
Thus, if study sites are sufficiently separated geographically, it is reasonable to expect that the study site variable is predictive of exposure during the duration of the trial.
This variation in exposure should lead to variation in principal strata due to differences in exposure.
\begin{assumption}[Study site relevance]\label{cond:r-relevant}
  $S^{P_0}_i$ is conditionally dependent on study site affiliation and baseline covariates, or $R_i \not \indy S^{P_0}_i \mid X_i$. 
\end{assumption}
This is a common assumption in multi-site principal stratification modeling, but it is one that is especially reasonable in vaccine efficacy trials.
See \cite{yuan_identifying_2019} and references therein.

Another key feature of vaccine efficacy trials is that patients' susceptibility to the pathogen of interest is assessed via comprehensive baseline measurements of correlates of protection.
For example, in influenza trials it is common to measure the pre-season, pre-vaccination (i.e. baseline) influenza antibody concentrations via hemagglutination inhibition (HAI) assays or neuraminidase inhibition (NAI) assays against different strains of influenza \citep{montoComparativeEfficacyInactivated2009}.
These assays are categorical measurements generated from serial two-fold dilutions of patient serum samples; high titer values are associated with lower infection risk, and can be considered surrogates for past infections and/or past influenza vaccinations \citep{zelner_effects_2019}.
Given the fact that the participants will be inoculated against influenza, it is reasonable that the measurements of these values are not independent of principal strata.
We will call these measurements $A_i$, and assume that $A_i \not \indy S^{P_0}_i \mid X_i$.
The structure of multisite randomized trials is such that we make two further assumptions about the joint distribution of covariate values $A_i$ and potential post-infection outcomes $(Y_i(z_1), \dots, Y_i(z_{N_z}))$.
Formally,
\begin{assumption}[Covariate homogeneity]\label{cond:a-indy}
  $A_i$ is conditionally independent of the study site and treatment assignment given the principal stratum and baseline covariates, or $A_i \indy R_i, Z_i \mid S^{P_0}_i, X_i$, {\rm and}
\end{assumption}
\begin{assumption}[Causal Homogeneity]\label{cond:homogeneity}
  Conditional on principal stratum $S^{P_0}_i, A_i$, and $X_i$ the potential outcomes $(Y_i(z_1),\dots, Y_i(z_{N_z}))$ are independent of $R_i$, or $(Y_i(z_1),\dots, Y_i(z_{N_z})) \indy R_i \mid S^{P_0}_i, A_i, X_i$.
\end{assumption}
These two assumptions are crucial for the nonparametric identifiability of vaccine efficacy against post-infection outcomes.
Both assumptions depend on the validity of the measurements $A_i$ being a good proxy for infection risk frailty and post-infection outcome frailty conditional on principal stratum.

\Cref{cond:a-indy} is equivalent to assuming that individuals' covariate measurements $A_i$ are exchangeable within strata defined by $(X_i = x, S_i^{P_0} = u)$; it is thus similar to the assumption of no unmeasured confounders in observational trials.
The randomization of treatment assignment assures the independence of $A_i$ and $Z_i$.
The conditional independence of $A_i$ and $R_i$ given other baseline covariates and principal stratum means that $A_i$ must measure baseline biological predisposition to infection and predisposition for vaccine efficacy against infection; if these characteristics are captured by $A_i$, then conditional on principal stratum, or the joint measurement of infection under placebo and infection under treatment, there should be no variation between study sites.
The independence assumption may be characterized as $A_i$ being a measure of susceptibility to infection, while $R_i$ affects the frailty of infection conditional on being susceptible.
A priori, there is no reason to assume that susceptibility to infection, conditional on principal stratum varies by study site.
\citep{longiniFrailtyMixtureModel1996,aalen_survival_2008}.

\Cref{cond:homogeneity} is equivalent to assuming conditional exchangeability of $Y_i(z)$ within strata defined by $(X_i = x, S_i^{P_0} = u, A_i = k)$ \citep{saarela_exchangeability}.
\Cref{cond:homogeneity} is commonly deemed reasonable in vaccine efficacy trials \citep{tsiatis_estimating_2022}, and is justifiable given the conditioning on $A_i$ and other pretreatment covariates.
We may loosen \Cref{cond:a-indy} by employing a parametric model for $A_i \mid R_i, S^{P_0}, X_i$.
This is discussed more in \Cref{subsec:sens}.

Now we will define the joint distribution of influenza test results, reported severe illness, and pre-season antibody concentration given treatment assignment, study site membership, and baseline covariates under \Crefrange{cond:nondiff}{cond:homogeneity}.

Let $\theta^{r,x}_{u} = P(S^{P_0}_i = u \mid R_i = r, X_i = x)$ where $u \in \{0,1\}^{N_z}$, and $\beta^{u,x}_{j,k} =  P(Y_i(z_j) = 1 \mid S^{P_0}_i = u, A_i = k, X_i = x)$.
Let $u_j = P(S_i(z_j) = 1 \mid S_i^{P_0} = u)$.
Then $\beta_{j,k}^{u,x}$ is only defined for $j$ such that $u_j = 1$.
Let $a^{u,x}_{k} = P(A_i = k \mid S^{P_0}_i = u, X_i = x)$.
Further, recall $\mathrm{sn}_S = P(\tilde{S}_i = 1 \mid S_i = 1), \mathrm{sp}_S = P(\tilde{S}_i = 0 \mid S_i = 0)$ and $\mathrm{sn}_Y = P(\tilde{Y}_i = 1 \mid Y_i = 1), \mathrm{sp}_Y = P(\tilde{Y}_i = 0 \mid Y_i = 0)$.
Let the probability of observing an infection test result $s$ and a preseason antibody titer level $k$ given treatment assignment $z$, study site membership $r$ and pretreatment covariates $x$, or $q_{sk\mid zrx} = P(\tilde{S}_i = s, A_i = k \mid Z_i = z, R_i = r, X_i = x)$, be defined as:
\begin{align*}
  q_{sk\mid jrx} & = \mathrm{sn}_S^s(1 - \mathrm{sn}_S)^{1 - s} \textstyle\sum_{u \mid u \in \mathcal{S}, u_j = 1}a^{u,x}_k\theta^{r,x}_{u}
                   + \mathrm{sp}_S^{1-s}(1 - \mathrm{sp}_S)^{s}  \textstyle\sum_{u \mid u \in \mathcal{S}, u_j = 0}a^{u,x}_k\theta^{r,x}_{u}.
\end{align*}

Similarly, we define the probability of observing a reported severe illness outcome $y$ given treatment assignment $z$, study site membership $r$, and pretreatment covariates $x$, or $q_{y \mid kzrx} = P(\tilde{Y}_i = y \mid Z_i = z, R_i = r, A_i = k, X_i = x)$, as:
\begin{align*}
  q_{y \mid kjrx} & = \mathrm{sn}_Y^y(1 - \mathrm{sn}_Y)^{1 - y} \textstyle\sum_{u \mid u \in \mathcal{S}, u_j = 1} \beta^{u,x}_{j,k}\theta^{r,x}_{u} \\
                  &  + \mathrm{sp}_Y^{1-y}(1 - \mathrm{sp}_Y)^{y} (\textstyle\sum_{u \mid u \in \mathcal{S}, u_j = 1} (1 - \beta^{u,x}_{j,k})\theta^{r,x}_{u} + \textstyle\sum_{u \mid u \in \mathcal{S}, u_j = 0}\theta^{r,x}_{u})
\end{align*}

In designing a clinical trial in which a primary or secondary endpoint is severe illness, limiting the definition of severe illness to encompass only the most extreme illness can at once increase the sensitivity and specificity of reporting.
In the next section we will show which model parameters are identifiable from observed data, as well as showing that this identifies $\mathrm{VE}^u_{I,jk}$.

\subsubsection{Identifiability of expanded model}

Now we will show that the joint variation in observed antibody concentrations and infection rates across study sites identifies the joint distribution of principal strata proportions and covariate values by study site under Assumptions \Crefrange{cond:nondiff}{cond:homogeneity}.
Then, given sufficient variation in principal strata proportions between study sites, the distribution of post-infection potential outcomes can be identified as well, along with $\mathrm{VE}^u_{I,jk}$, the VE against post-infection outcomes in principal stratum $u$ for treatments $z_j$ (assuming $z_k$ is the placebo).

The proof depends on representing the observed distribution $q_{sk \mid jrx}$ for a fixed $x$ as a $3$-way array and using a modified tensor decomposition uniqueness theorem from \cite{kruskalThreewayArraysRank1977}.
\citeauthor{kruskalThreewayArraysRank1977}'s theorem defines sufficient conditions for the uniqueness of the \emph{triple product} decomposition of $L$, where this product is defined in \Cref{defn:triple-prod}.
\begin{definition}[Array triple product]\label{defn:triple-prod}
  Let the array triple product with resulting array $L \in \R^{I \times J \times K}$ be defined between matrices $A \in \R^{I \times M}$,
  $B \in \R^{J \times M}$, $C \in \R^{K \times M}$.
  The operation is represented as $L = [A, B, C]$. 
  As a result, the $(i,j,k)^\mathrm{th}$ element of $L$, $L_{ijk}$, is defined the sum of three-way-products of elements $a_{im}, b_{jm}, c_{km}$, i.e.:
  $$
  L_{ijk} = \sum_{m=1}^M a_{im}b_{jm}c_{km}.
  $$
\end{definition}
The sufficient conditions concern the \emph{Kruskal ranks} of the matrices $A, B, C$, defined in \Cref{defn:krank}.
\begin{definition}[Kruskal rank]\label{defn:krank}
  Let the Kruskal rank of a matrix $B \in \R^{I \times M}$ be $k_B \in [0, 1, 2, \dots, M]$, and let $k_B$ be the maximum integer such that every set of $k_B$ columns of $B$ are linearly independent.
\end{definition}
Kruskal rank is stricter than matrix rank. 
To see why, consider a matrix with $M$ columns of which two are repeated. 
The rank of the matrix is at most $M-1$, but Kruskal rank is at most $1$.
A corollary of the definition is that if a matrix is full column rank, its Kruskal rank equals its column rank.

Let $L$ be the $3$-way array representing $q_{sk \mid zrx}$, where we fix $X_i = x$ for each unique value of $X_i$. The array's dimensions are $2N_z \times N_a \times N_r$ and is defined so that the $(j,k,r)^\mathrm{th}$ element $P(S_i = \ind{j \leq N_z}, A_i = k \mid Z = z_{j - 1 \bmod N_z + 1}, R_i = r, X_i = x)$.
If we look at the matrix that results from fixing the third array index, also known as the \emph{3-slab} and denoted as $L_r \in \R^{2N_z \times N_a}$, we can see a possible decomposition of this array.
Let $\sum_{u_z=s}{a^{u,x}_k} \theta_u^{r,x}$ denote the sum over elements $u \in \mathcal{S}$ such that $P(S_i(z) = 1 \mid S_i^{P_0} = u)$.
The let $L_r$ be defined as 
{\scriptsize\makeatletter\setlength\BA@colsep{1pt}\makeatother
\begin{align*}
  \begin{blockarray}{cccc}
    (a = 1) & \dots & (a = N_a)  \\ 
    \begin{block}{[ccc]c}
      (1 - \mathrm{sp}_S)\sum_{u_{z_1} = 0}{a^{u,x}_1} \theta_u^{r,x} +  \mathrm{sn}_S\sum_{u_{z_1} = 1}{a^{u,x}_1} \theta_u^{r,x} & \dots & (1 - \mathrm{sp}_S)\sum_{u_{z_1} = 0}{a^{u,x}_{N_a}} \theta_u^{r,x} +  \mathrm{sn}_S\sum_{u_{z_1} = 1}{a^{u,x}_{N_a}} \theta_u^{r,x} & (s=1,z=z_1) \\
      (1 - \mathrm{sp}_S)\sum_{u_{z_2} = 0}{a^{u,x}_{1}} \theta_u^{r,x} +  \mathrm{sn}_S\sum_{u_{z_2} = 1}{a^{u,x}_{1}} \theta_u^{r,x}  & \dots & (1 - \mathrm{sp}_S)\sum_{u_{z_2} = 0}{a^{u,x}_{N_a}} \theta_u^{r,x} +  \mathrm{sn}_S\sum_{u_{z_2} = 1}{a^{u,x}_{N_a}} \theta_u^{r,x}  & (s=1,z=z_2) \\
      \vdots & \ddots & \vdots & \vdots \\ 
      (1 - \mathrm{sp}_S)\sum_{u_{z_{N_z}} = 0}{a^{u,x}_{1}} \theta_u^{r,x} +  \mathrm{sn}_S\sum_{u_{z_{N_z}} = 1}{a^{u,x}_{1}} \theta_u^{r,x}  & \dots & (1 - \mathrm{sp}_S)\sum_{u_{z_{N_z}} = 0}{a^{u,x}_{N_a}} \theta_u^{r,x} +  \mathrm{sn}_S\sum_{u_{z_{N_z}} = 1}{a^{u,x}_{N_a}} \theta_u^{r,x}  & (s=1,z=z_{N_z}) \\
      \mathrm{sp}_S\sum_{u_{z_1} = 0}{a^{u,x}_1} \theta_u^{r,x} +  (1 - \mathrm{sn}_S)\sum_{u_{z_1} = 1}{a^{u,x}_1} \theta_u^{r,x}      & \dots &  \mathrm{sp}_S\sum_{u_{z_1} = 0}{a^{u,x}_{N_a}} \theta_u^{r,x} +  (1 - \mathrm{sn}_S)\sum_{u_{z_1} = 1}{a^{u,x}_{N_a}} \theta_u^{r,x} & (s=0,z=z_1) \\
      \mathrm{sp}_S\sum_{u_{z_2} = 0}{a^{u,x}_{1}} \theta_u^{r,x} +  (1 - \mathrm{sn}_S)\sum_{u_{z_2} = 1}{a^{u,x}_{1}} \theta_u^{r,x}      & \dots & \mathrm{sp}_S\sum_{u_{z_2} = 0}{a^{u,x}_{N_a}} \theta_u^{r,x} +  (1 - \mathrm{sn}_S)\sum_{u_{z_2} = 1}{a^{u,x}_{N_a}} \theta_u^{r,x}& (s=0,z=z_2) \\
      \vdots & \ddots & \vdots & \vdots \\ 
      \mathrm{sp}_S\sum_{u_{z_{N_z}} = 0}{a^{u,x}_{1}} \theta_u^{r,x} +  (1 - \mathrm{sn}_S)\sum_{u_{z_{N_z}} = 1}{a^{u,x}_{1}} \theta_u^{r,x}      & \dots & \mathrm{sp}_S\sum_{u_{z_{N_z}} = 0}{a^{u,x}_{N_a}} \theta_u^{r,x} +  (1 - \mathrm{sn}_S)\sum_{u_{z_{N_z}} = 1}{a^{u,x}_{N_a}} \theta_u^{r,x}& (s=0,z=z_{N_z}) \\
    \end{block}
  \end{blockarray}
\end{align*}
}
We can define matrices encoding the distribution of principal strata by study site, and the distribution of pre-season titers by principal stratum.
In order to define these matrices such that they share a common ordering along the axes defined by the principal strata, we shall use the the natural ordering of the binary vectors $S_i^{P_0}$: the base-10 representation of the principal stratum.
In order to formalize this ordering, we define a map, $\varpi_m(j)$, which generates the $m$-digit binary representation of the integer $j$ as a length-$m$ binary vector.
We also define its inverse, $\varpi_m(u)^{-1}$, where $u$ is an element of $\mathcal{S}$.
\begin{definition}[Base-10 to binary map]\label{defn:bin-map}
  Let the operator $\varpi_m$ be defined as $\varpi_m(\cdot): j \to \{0,1\}^m, j \in \mathbb{N}, j \leq 2^{m}-1$ with elements $\varpi_m(j)_i \in \{0,1\}$, so $\varpi_m(j)$ is the base-2 representation of $j$ with $m$ digits represented as a binary $m$-vector.
  The binary representation is indexed so the $i^{\mathrm{th}}$ element of the vector corresponds to the digit for $2^{i-1}$.
  Let the inverse operator $\varpi_m^{-1}(\cdot):\{0,1\}^{m} \to j$, or the binary to base-10 conversion. 
  Let digit $i$ of $\varpi_m(\cdot)$ represent the digit for $2^{i - 1}$.
\end{definition}
For example, $\varpi_3(4) = (0,0,1)$, $\varpi_5(4) = (0,0,1,0,0)$, and $\varpi_3((0,0,1))^{-1} = \varpi_5((0,0,1,0,0))^{-1} = 4$.
In order to see how $\mathcal{S}$ is ordered, suppose that $N_z = 3$ so $\mathcal{S} \equiv \{0,1\}^3$. Then the third and fourth elements of the ordered set of principal strata are $(0,1,0)$ and $(1,1,0)$ respectively.

With this ordering defined, let the matrix $P_{N_z}^x(A \mid S^{P_0})$ in $\R^{N_a \times 2^{N_z}}$ encode the distributions $A_i \mid S^{P_0}_i, X_i$
with $(k,j)^\mathrm{th}$ element $P_{N_z}(A_i = k\mid S^{P_0}_i = \varpi_{N_z}(j-1), X_i = x)$. 
Let the matrix $P_{N_z}^x(S^{P_0} \mid R)$ in $\R^{2^{N_z} \times N_r}$ encode the distribution $S^{P_0}_i \mid R_i, X_i$ with $(k,j)^\mathrm{th}$ element 
$P_{N_z}^x(S^{P_0}_i = \varpi_{N_z}(k-1)\mid R_i = j, X_i = x)$. 

Finally, let matrix $P_{N_z}(\tilde{S} \mid Z, S^{P_0}) \in \R^{2N_z \times 2^{N_z}}$ with $(k,j)^\mathrm{th}$ element
\begin{align*}
\mathrm{sn}_S^{\varpi_{N_z}(j-1)_k}(1 - \mathrm{sp}_S)^{1 - \varpi_{N_z}(j-1)_k} \ind{k \leq N_z} + (1-\mathrm{sn}_S)^{\varpi_{N_z}(j-1)_{k-N_z}}\mathrm{sp}_S^{1 - \varpi_{N_z}(j-1)_{k-N_z}} \ind{k > N_z}.
\end{align*}

Let $P_{N_z}^x(S^{P_0} \mid R = r)$ be the $r^\mathrm{th}$ column of matrix $P_{N_z}^x(S^{P_0} \mid R)$. Then $L_{r}$ can be represented in matrix form as
$$
L_r = P_{N_z}(\tilde{S} \mid Z, S^{P_0})\mathrm{diag}(P_{N_z}^x(S^{P_0} \mid R = r))  P_{N_z}^x(A \mid S^{P_0})^T
$$
This structure again allows us to define $L$ as the triple product of these three matrices, each of which have columns that correspond to principal strata:
$$ L = [P_{N_z}(\tilde{S} \mid Z, S^{P_0}),P_{N_z}^x(A \mid S^{P_0}), P_{N_z}^x(S^{P_0} \mid R)^T].$$ 

The conditions for the identifiability of the model parameters are outlined below:
\begin{theorem}\label{thm:id-noisy}
Let $N_z \geq 2$. 
Suppose \Crefrange{cond:SUTVA}{cond:homogeneity} hold.
If both $\mathrm{sn}_S, \mathrm{sp}_S$ lie in $[0, 1/2)$ or both lie in $(1/2, 1]$, $P_{N_z}^x(A \mid S^{P_0})$ is at least Kruskal rank $2^{N_z} - 1$ and $P_{N_z}^x(S^{P_0} \mid R)$ is rank $2^{N_z}$ for all $x$ then the counterfactual distributions $P(S^{P_0}_i = u \mid R_i = r, X_i = x)$,  $P(A_i = k \mid S^{P_0}_i = u, X_i = x)$ are identifiable as are the quantities $\mathrm{sn}_S, \mathrm{sp}_S, \mathrm{sp}_Y, \mathrm{VE}^u_{I,jk}(k)$, and $ \mathrm{VE}^u_{I,jk}$.
Furthermore, if $\mathrm{sn}_Y$ is unknown (known), distributions $P(Y_i(z_j) = 1 \mid S^{P_0}_i = u, A_i = k, X_i = x)$ are identifiable up to an unknown (known) common constant, $r_Y = \mathrm{sn}_Y + \mathrm{sp}_Y - 1$.  
\end{theorem}

Given that the marginal distribution of reported severe illness, or $P(\tilde{Y}_i(z_j) = 1 \mid S^{P_0}_i = u, A_i = k, X_i = x)$ is identified, along with $\mathrm{sp}_Y$, we can write the estimand of interest, the vaccine efficacy against severe illness within the always-infected stratum explicitly.
To see why, note that for any latent binary random variable $Q$ and its measurement, $\tilde{Q}$ with associated sensitivity $\mathrm{sn}_Q$ and specificity $\mathrm{sp}_Q$, we have the following identity:
\begin{align}\label{eq:id-binary}
  P(Q = 1) = \frac{P(\tilde{Q} = 1) - (1 - \mathrm{sp}_Q)}{\mathrm{sn}_Q + \mathrm{sp}_Q - 1}.
\end{align}
Employing \Cref{eq:id-binary} to \Cref{defn:ve-i-multi} yields:
\begin{align*}
  \mathrm{VE}^u_{I,jk} = \frac{\Exp{\tilde{Y}_i(z_k) - \tilde{Y}_i(z_j)  \mid S^{P_0}_i = u}}{\Exp{\tilde{Y}_i(z_k) \mid S^{P_0}_i = u} - (1 - \mathrm{sp}_Y)}.
\end{align*}
Thus $\mathrm{VE}^u_{I,jk}$ is point identified by observed data without needing to identify $\mathrm{sn}_Y$.
This estimand marginalizes over the population distribution of $X_i$, which may be known, or may be estimated.

Furthermore, the identifiability of the conditional counterfactual distributions $P(\tilde{Y}_i(z_j) = 1 \mid S^{P_0}_i = u, A_i = k, X_i = x)$ allows for causal effect heterogeneity by covariate $A_i$. 
\begin{definition}[Conditional VE against post-infection outcome $Y$]\label{defn:hetero-treat}
  \[\mathrm{VE}_{I,jk}(\ell) = \frac{\Exp{\tilde{Y}_i(z_k) - \tilde{Y}_i(z_j)  \mid S^{P_0}_i = u, A_i = \ell}}{\Exp{\tilde{Y}_i(z_k) \mid S^{P_0}_i = u, A_i = \ell} - (1 - \mathrm{sp}_Y)}.\]
\end{definition}
Again this estimand marginalizes over $X_i$.

\Cref{thm:id-noisy} allows for a more realistic model of infection measurement than \cite{hudgensCausalVaccineEffects2006} and does not require any restrictions on the space of principal strata. 
The primary benefit of an unrestricted principal strata distribution is that we can jointly infer vaccine efficacy against infection and vaccine efficacy against a post-infection outcome.
This will aid in designing comprehensive randomized trials for vaccine efficacy.

The proof of \Cref{thm:id-noisy}, shown in \Cref{proof:id-noisy} is related to the methods in \cite{jiangPrincipalCausalEffect2016} and \cite{dingIdentifiabilityEstimationCausal2011}.
\cite{dingIdentifiabilityEstimationCausal2011} addresses problems of identifiability in survivor average treatment effects, which is mathematically analogous to vaccine efficacy for post-infection outcomes, by measuring covariates that are related to the principal strata.
\cite{jiangPrincipalCausalEffect2016} identifies principal causal effects in binary surrogate endpoint evaluations. Despite not being mathematically identical to vaccine efficacy, binary surrogacy endpoint evaluation is ultimately a problem in identification of principal causal effects.
Most importantly, the proof does not require any restrictions on the distribution of post-infection outcomes.
This makes the result applicable to categorical post-infection outcomes like symptom scores, or continuous outcomes like viral load.
More broadly, our method is applicable to principal stratification problems outside the scope of vaccine efficacy. 

The identifiability results in \Cref{thm:id-noisy} suggest the following so-called transparent parameterization\footnote{See \citep{gustafson2015bayesian} for more details on inference in partially identified Bayesian models}: $(\beta^{u,x}_{j,k}, \mathrm{sp}_Y,\mathrm{sn}_Y) \to (\tilde{p}^{u,x}_{j,k} = (\mathrm{sn}_Y + \mathrm{sp}_Y - 1) \beta_{j,k}^{u,x} + (1 - \mathrm{sp}_Y), \mathrm{sp}_Y, \mathrm{sn}_Y)$ .
The quantities 
$
\tilde{p}^{u,x}_{j,k} = P(\tilde{Y}_i = 1 \mid Z_i = z_j, S_i^{P_0} = u, A_i = k)
$
and $\mathrm{sp}_Y$ are identified by the data, while $\mathrm{sn}_Y$ is not.
This yields the following asymptotic identification regions for $\mathrm{sn}_Y$ and $\beta_{j,k}^{u,x}$:
\begin{align}
    \mathrm{sn}_Y \in \lp \max_{x,u,j,k}(\tilde{p}^{u,x}_{j,k}) ,  1\rp, \quad \beta_{j,k}^{u,x} \in \lp\frac{\tilde{p}^{u,x}_{j,k} - (1 - \mathrm{sp}_Y)}{\mathrm{sp}_Y}, \frac{\tilde{p}^{u,x}_{j,k} - (1 - \mathrm{sp}_Y)}{\max_{x,u,j,k}(\tilde{p}^{u,x}_{j,k}) + \mathrm{sp}_Y - 1} \rp
\end{align}
This may be useful for policymakers interested in absolute risk of post-infection outcomes to forecast the burden on healthcare centers under different vaccination policies.

We will present a final corollary that will be useful in our applied examples:
\begin{corollary}\label{cor:id-noisy-A}
  Suppose in addition to \Crefrange{cond:SUTVA}{cond:homogeneity}, researchers do not directly observe $A_i$, but instead observe a misclassified version of $A_i$, $\tilde{A}_i$, such that the following nondifferential error assumption holds: $\tilde{A}_i \indy \tilde{S}_i, \tilde{Y}_i, Y_i(z_j,S(z_j)), R_i, Z_i, S^{P_0}_i \mid A_i, X_i$.
  If both $\mathrm{sn}_S, \mathrm{sp}_S$ lie in $[0, 1/2)$ or both lie in $(1/2, 1]$, $P_{N_z}^x(\tilde{A} \mid S^{P_0})$ is at least Kruskal rank $2^{N_z} - 1$ and $P_{N_z}^x(S^{P_0} \mid R)$ is rank $2^{N_z}$ for all $x$ then the counterfactual distributions $P(S^{P_0}_i = u \mid R_i = r, X_i = x)$,  $P(\tilde{A}_i = k \mid S^{P_0}_i = u, X_i = x)$ are identifiable as are the quantities $\mathrm{sn}_S, \mathrm{sp}_S, \mathrm{sp}_Y$, and $ \mathrm{VE}^u_{I,jk}$.
  Furthermore, if $\mathrm{sn}_Y$ is unknown (known), distributions $P(Y_i(z_j) = 1 \mid S^{P_0}_i = u, \tilde{A}_i = k, X_i = x)$ are identifiable up to an unknown (known) common constant, $r_Y = \mathrm{sn}_Y + \mathrm{sp}_Y - 1$.  
\end{corollary}
The proof, shown in \Cref{proof:id-noisy-A}, follows directly from the proof of \Cref{thm:id-noisy} and the nondifferential misclassification error assumption for $A$. 

While misclassified $\tilde{A}$ precludes learning heterogeneous treatment effects like in \Cref{defn:hetero-treat}, marginalizing over the identifiable distribution $\tilde{A}_i \mid S^{P_0}, X_i$ will yield the average post-infection vaccine efficacy.

\subsection{Models, priors and sensitivity analyses}\label{subsec:sens}

Under the assumptions laid out in \Cref{sec:multi-arm} the observed data likelihood is multinomial for each study site and treatment group.

Let $\tilde{n}_{syk}(j,r,x)$ be $\sum_{i=1}^n \ind{\tilde{S}_i = s}\ind{\tilde{Y}_i = y}\ind{Z_i = z_j}\ind{R_i = r}\ind{A_i = k}\ind{X_i = x}$, and let the error-free partially-observed causal model probabilities
$p_{syk \mid jrx} = P(S_i = s, Y_i = y, A_i = k \mid Z_i = z_j, R_i = r, X_i = x)$ be defined as:
\begin{align}\label{eq:misclass-obs-model-prob-multi-site-multi-z}
  p_{1yk \mid jrx} = \textstyle\sum_{u \mid u \in \mathcal{S}, u_j = 1}a^{u,x}_k\theta^{r,x}_{u} (\beta^{u,x}_{j,k})^y(1 - \beta^{u,x}_{j,k})^{1-y}, \quad p_{0*k \mid jr} = \textstyle\sum_{u \mid u \in \mathcal{S}, u_j = 0}a^{u,x}_k\theta^{r,x}_{u},
\end{align}
where we note that $p_{01k \mid jrx} = 0$ for all $k,j,r$.
Then we define the observable joint probabilities $q_{syk \mid jrx} = P(\tilde{S}_i = s, \tilde{Y}_i = y, A_i = k \mid Z_i = z_j, R_i = r, X_i = x)$ as
\begin{align*}
  q_{syk \mid jrx} & = \mathrm{sn}_S^s(1 - \mathrm{sn}_S)^{1 - s} \mathrm{sn}_Y^y(1 - \mathrm{sn}_Y)^{1 - y} p_{11kjrx}
              + \mathrm{sn}_S^s(1 - \mathrm{sn}_S)^{1 - s} \mathrm{sp}_Y^{1-y}(1 - \mathrm{sp}_Y)^{y} p_{10kjrx} \\
            & + \mathrm{sp}_S^{1-s}(1 - \mathrm{sp}_S)^{s} \mathrm{sp}_Y^{1-y}(1 - \mathrm{sp}_Y)^{y} p_{0*kjrx},
\end{align*}
This allows us to define the observational model as:

{\scriptsize
  \begin{align}
    \begin{split}\label{eq:misclass-obs-model-multi-site-multi-z}
      (&\tilde{n}_{00 1}(j,r,x),\tilde{n}_{01 1}(j,r,x), \tilde{n}_{101}(j,r,x), \tilde{n}_{111}(j,r,x), \dots, \tilde{n}_{00 N_a}(j,r,x),\tilde{n}_{01 N_a}(j,r,x), \tilde{n}_{10N_a}(j,r,x), \tilde{n}_{11N_a}(j,r,x)) \sim \\
       & \text{Multinomial}(n(j,r,x) \mid q_{001 \mid jrx}, q_{011\mid jrx}, q_{101\mid jrx}, q_{111\mid jrx},\dots,q_{00N_a\mid jrx},q_{01N_a\mid jrx}, q_{10N_a\mid jrx}, q_{11N_a\mid jrx}),\\
       &j \in \{1, \dots, N_z\}, \, r \in \{1,\dots,N_r\}, \, x \in \{1, \dots, N_x\}
    \end{split}
  \end{align}
}

The post-infection severe-illness models can be formulated as saturated logistic regressions:
\begin{align*}
  \log\frac{P(Y_i(z_j) = 1 \mid S^{P_0}_i = u, A_i = k, X_i = x)}{P(Y_i(z_j) = 0 \mid S^{P_0}_i = u, A_i = k, X_i = x)} = \alpha_j^{u,x} +  \delta^{u,x}_{j,k},\,\,\beta^{u,x}_{j,k} & = \frac{e^{\alpha^{u,x}_j + \delta^{u,x}_{j,k}}}{1 + e^{\alpha^{u,x}_j + \delta^{u,x}_{j,k}}}, \\
  \delta^{u,x}_{j,1} & = 0 \, \forall \,j,\, u,\, x.
\end{align*}
We may also implement a James-Stein-type estimator that asymptotically reduces to the saturated logistic regression:
\begin{align*}
  \log\frac{P(Y_i(z_j) = 1 \mid S^{P_0}_i = u, A_i = k, X_i = x)}{P(Y_i(z_j) = 0 \mid S^{P_0}_i = u, A_i = k, X_i = x)} = \alpha_j^{u} + \epsilon^{u,x}_{j,k}, \,
  \epsilon^{u,x}_{j,k} \sim N(0, \tau_{\epsilon}^2) \, \forall \,j,\, u,\, x,\,k.
\end{align*}
The model can accommodate deviations from \Cref{cond:homogeneity} through an additive term $\varepsilon^{u,x}_{r}$ capturing heterogeneity between study sites:
\begin{align*}
  \log\frac{P(Y_i(z_j) = 1 \mid S^{P_0}_i = u, A_i = k, R_i = r, X_i = x)}{P(Y_i(z_j) = 0 \mid S^{P_0}_i = u, A_i = k, R_i = r, X_i = x)} & = \alpha_j^{u,x} +  \delta^{u,x}_{j,k} + \varepsilon^{u,x}_{r},\\
  \varepsilon^{u,x}_r & \sim \mathrm{Normal}(0, (\tau^u_{\varepsilon})^2).
\end{align*}
We can fix $\tau^{u,x}_{\varepsilon}$ to several values for sensitivity analysis, as developed in \cite{jiangPrincipalCausalEffect2016}. 

We may write the probability models for $S_i^{P_0} \mid R_i, X_i$  and $A_i \mid S_i^{P_0}, X_i$  as two saturated multinomial regressions, given \Cref{cond:a-indy} that $A_i \indy R_i, Z_i \mid  S^{P_0}_i, X_i$. 
\begin{align*}
  \log \frac{P(S^{P_0}_i = u\mid R_i = r, X_i = x)}{P(S^{P_0}_i = u_0\mid R_i = r, X_i = x)} & =  \mu_{u}^r + \eta^{x}_{u} + \eta^{r,x}_{u} \\
  \log \frac{P(A_i = k\mid S^{P_0}_i = u, X_i = x)}{P(A_i = k_0\mid S^{P_0}_i = u, X_i = x)} & =  \nu^u_k + \gamma^x_k +  \gamma^{u,x}_{k},
\end{align*}
where 
\begin{equation*}
\arraycolsep=14.4pt
\begin{array}{cc}
  \theta^{r,x}_{u} = \frac{e^{\mu_u^r + \eta^x_{u} + \eta^{r,x}_u}}{\sum_{w \in \mathcal{S}}e^{\mu_{w}^r +  \eta^x_{u} + \eta^{r,x}_u}}, \,\mu_{u_0}^r = 0 \,\forall\, r, & a^{u,x}_k = \frac{e^{\nu^u_k + \gamma_k^x + \gamma_k^{u,x}}}{\sum_{m = 1}^{N_a} e^{\nu^{u}_m + \gamma_k^x + \gamma_k^{u,x}}},\,\nu^{u}_{k_0} = 0\, \forall\, u.
\end{array}
\end{equation*}
Note that $\eta_{u_0}^x, \gamma_{k_0}^x, \eta_{u_0}^{r,x}, \gamma_{k_0}^x, \gamma_{k_0}^{u,x}$ are all zero for all $x$.
Furthermore, for given reference categories $x_0, u_0, r_0$, $\eta_u^{x_0},\gamma_{k}^{x_0}$ are zero for all $u,k$, while $\eta_{u}^{r_0,x}$ is zero for all $x$, $\eta_{u}^{r,x_0}$ is zero for all $r$, $\gamma_k^{u,x_0}$ is zero for all $u$ and  $\gamma_k^{u_0,x}$ is zero for all $x$.
This leads to a tidy representation of the log-odds of belonging to stratum $u$ vs. $u_0$ conditional on $A_i=k, R_i=r, X_i=x$: 
\begin{equation*}
  \log \frac{P(S^{P_0}_i = u \mid A_i = k, R_i = r, X_i = x)}{P(S^{P_0}_i = u_0 \mid A_i = k, R_i = r, X_i = x)}  = 
  \mu_{u}^r + \nu^u_k - \nu^{u_0}_k + \gamma^{u,x}_k - \gamma^{u_0,x}_k + \eta^x_u + \eta^{r,x}_u.
\end{equation*}

If we suspect deviations from \Cref{cond:a-indy}, we can add an interaction between $A_i$ and $R_i$ in the multinomial regression model for $A_i$:
\begin{equation}\label{eq:a-interact-mod}
  \log \frac{P(A_i = k\mid R_i = r, S^{P_0}_i = u, X_i = x)}{P(A_i = k_0\mid R_i = r, S^{P_0}_i = u, X_i = x)} =  \nu^u_k + \gamma^x_k +  \gamma^{u,x}_{k} + \epsilon^{r}_{k},    \quad  \epsilon^{r}_{k} \sim \mathrm{Normal}(0, (\tau^k_{\epsilon})^2) \,\forall r.
\end{equation}

\section{Design and analysis of vaccine efficacy studies}\label{sec:ve-studies}
There are several real-world applications for \Cref{thm:id-noisy} in vaccine efficacy studies. 
The first is for quantifying vaccine efficacy against post-infection outcomes like severe illness, medically-attended illness or death, which is the primary motivation for the methods we have developed here. 
A second is to quantify the impact on vaccination on secondary transmission to household contacts.
In both of these hypothetical trials, we imagine that participants are prospectively monitored for infection as well as the post-infection outcome of interest.
The infection monitoring might involve regular diagnostic testing or analysis of blood specimens for signs of infection.

A guiding philosophy we follow below is in using Bayesian models to design clinical trials with good Frequentist properties, as is discussed in \cite{berry_bayesian_2011}.
Thus, we chose our rejection region for our test statistic so as to limit our Type $1$ error to no more than $0.05$.

\subsection{Vaccine efficacy against severe illness trial design}\label{subsec:ve-p}
To show how our model can be used to design a vaccine efficacy study, we consider determining the sample size for two hypothetical clinical trials: one three-arm trial inspired by \cite{montoComparativeEfficacyInactivated2009}, and a two-arm trial inspired by \cite{polackSafetyEfficacyBNT162b22020}.
\cite{montoComparativeEfficacyInactivated2009} investigated vaccine efficacy against symptomatic influenza infection in a three-arm, double-blind placebo-controlled randomized trial.
\cite{polackSafetyEfficacyBNT162b22020} presented the results of the COVID-19 Pfizer vaccination trial, which measured vaccine efficacy against symptomatic infection using a two-arm double-blind placebo-controlled randomized trial.
All of our hypothetical trials are designed so as to jointly test the efficacy against infection and the efficacy against severe symptoms for the always-infected group.

In order to design our hypothetical trials, we simulate 200 datasets under the alternative hypothesis for each sample size and measure the proportion of datasets in which we reject the null hypothesis.
For both trials, we will target a power of $0.8$ against an alternative hypothesis that the vaccine efficacy against symptoms is equal to $0.6$ for the always-infected stratum (i.e. $S^{P_0}_i = (1,1,1)$ and $S^{P_0}_i = (1,1)$).
We reject the null when the posterior probability is $0.85$ or larger that vaccine efficacy against severe illness is above $0.1$ and that the vaccine efficacy against infection is greater than $0.3$. 
We can write the rejection region for $\mathrm{Data} = \{(\tilde{S}_i, \tilde{Y}_i, Z_i, R_i, A_i, X_i), \, 1 \leq i \leq n\}$ as $\{\mathrm{Data}: P(\mathrm{VE}^{(1,1,1)}_{I,31} > 0.1, \mathrm{VE}_{S,31} > 0.3 \mid \mathrm{Data}) \geq C\}$ for the three-arm trial and $\{\mathrm{Data}: P(\mathrm{VE}^{(1,1)}_{I,21} > 0.1, \mathrm{VE}_{S,21} > 0.3 \mid \mathrm{Data}) \geq C\}$ for the two-arm trial.
More details on the choice of $C$ is given in the next two subsections.
Broadly, our decision criterion is akin to that used in \cite{polackSafetyEfficacyBNT162b22020}, namely that the posterior probability is greater than $0.986$ that vaccine efficacy against confirmed COVID-19 is greater than $0.3$.
For example, in the two-arm scenario $C = 0.85$ adequately controls the Type $1$ error for a null hypothesis of no vaccine efficacy against severe illness.

We use the model defined in \Cref{eq:misclass-obs-model-multi-site-multi-z}; the computational details are discussed in \Cref{sec:numerical-deets}. 
 Given the results of \Cref{thm:id-noisy}, we can determine the number of study sites and the number of levels for $A_i$ that need to be observed in order to point identify the causal estimand of interest. 
 For the three-arm trial, we need at least $8$ study sites and a covariate with at least $7$ levels, while for the two-arm trial we need only $4$ study sites and a covariate with at least $3$ levels.
 In the two-arm study we simulated data from $8$ sites, while in the three-arm trial we simulated data from $16$ sites.

 \subsubsection{Two-arm trial}

 In the two arm trial, we sample the study-site-specific principal strata proportions from a Dirichlet distribution.
 The distribution's mean corresponds to the vector:
 $$
 (P(S_i^{P_0} = (0,0)), P(S_i^{P_0} = (0,1)), P(S_i^{P_0} = (1,0)), P(S_i^{P_0} = (1,1))),
 $$
 where
 $S_i^{P_0} = (S_i(1) = i, S_i(0) = j)$.
 The mean is set to $(0.78, 0.1, 0.02, 0.1)$, while the variance for each category is 
 $
 P(S_i^{P_0} = (i,j))(1 - P(S_i^{P_0} = (i,j)))/101.
 $
 These parameter settings result in a cumulative incidence of $16\%$.

 We fitted three models, which are described in detail in the appendix. Both included terms to account for possible deviations from assumption \Cref{cond:a-indy}, and both employ a hierarchical model for the probability of severe disease conditional on principal stratum and covariates.
 One model allows for an unrestricted set of principal strata, while the other two assume monotonicity, namely that $P(S_i^{P_0} = (1, 0)) = 0$.
 Of the two models that assume monotonicity, one assumes severe disease and infection are observed without error, while the other employs a measurement error model identical to the model proposed in this paper.
 In the results below, the unrestricted model is \emph{full}, while the model that incorrectly assumes monotonicity and assumes perfect measurements is \emph{mono-wo-meas-error}.
 The model that assumes monotonicity but allows for measurement error is called \emph{mono-w-meas-error}.
 Given the monotonicity assumption, neither monotonic model learns the distribution for $A_i$; they do, however, condition on $A_i$ when learning the distribution of the post-infection outcome.
 We measured coverage of the $95\%$ posterior credible intervals, bias and mean squared error of the posterior median of the target estimand, namely vaccine efficacy against severe illness in the always-infected stratum $S_i^{P_0} = (1,1)$, and vaccine efficacy against infection.
 The data were simulated under two alternative hypotheses and two null hypotheses.
 VE against severe illness was $\approx 0.6$ under one alternative hypothesis, and the alternative hypothesis had a VE against severe illness of $\approx -0.6$.
 Alternatively, under both null hypotheses VE against severe disease is equal to zero, while one has a VE against infection of $\approx 0.5$, and the other has a VE against infection $\approx 0.06$.
 Within these hypotheses we further stratified the simulations by two conditions: one in which \Cref{cond:a-indy} held and one in which \Cref{cond:a-indy} was violated.

 Bias and MSE under both alternative hypotheses are presented in \Crefrange{fig:bias}{fig:mse}.
 The figures show that under most scenarios, the full model has lower bias and MSE compared to the other models, though there are exceptions.
 The notable exceptions are at sample sizes of $20{,}000$, and $40{,}000$ under the null hypotheses.
 In these cases, monotonic models have smaller MSEs.
 This is due to the fact that the models which assume monotonicity have far fewer parameters than the joint model, and the increased bias is small compared to the decrease in variance.
\begin{figure}
  \centering
  \includegraphics[scale=0.7]{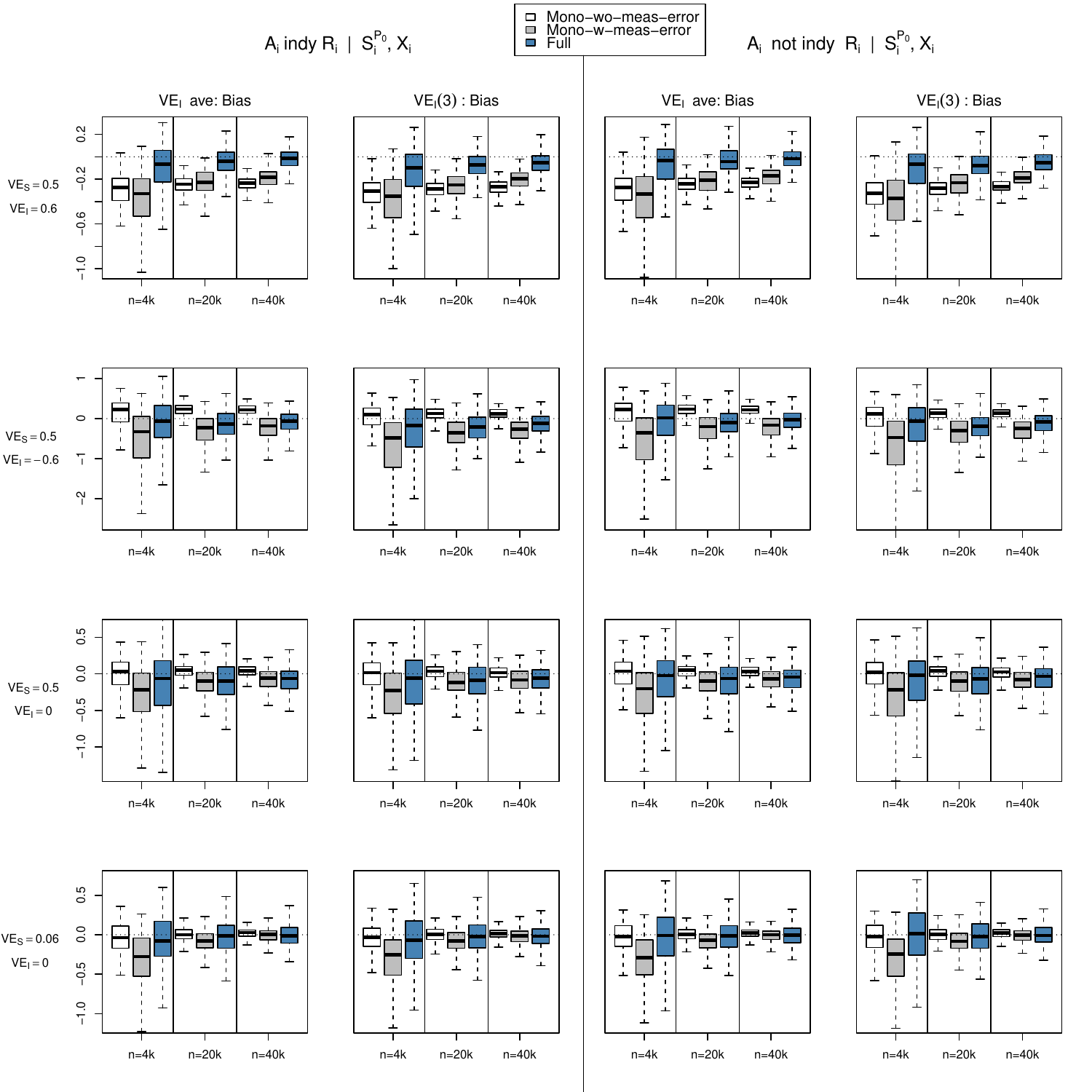}
  \caption{Boxplots collecting bias across simulated datasets for posterior median $\mathrm{VE}_I$ and $\mathrm{VE}_{I}(3)$.
    Each boxplot bar summarizes the bias of an estimator across $200$ datasets under a single sample size and hypothesis combination.
    The blue bar indicates bias under the full model, while the white bar indicates bias under the model employing a monotonicity assumption.
  Rows of the graph correspond to different hypotheses, while the columns correspond to whether or not \Cref{cond:a-indy} holds.}
  \label{fig:bias}
\end{figure}

\begin{figure}
  \centering
  \includegraphics[scale=0.7]{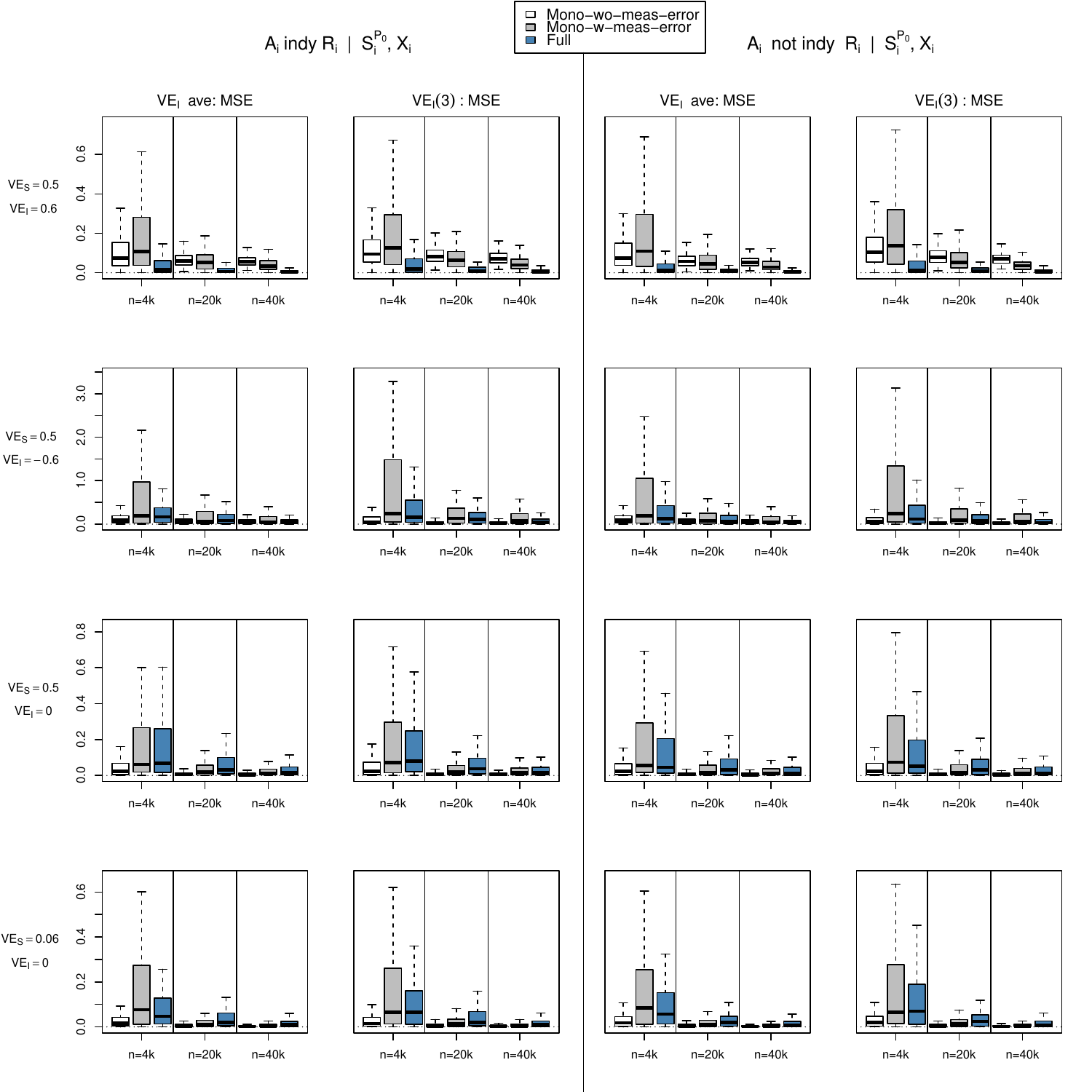}
  \caption{Boxplots collecting mean squared error across simulated datasets for posterior median $\mathrm{VE}_I$ and $\mathrm{VE}_{I}(3)$. 
    Each boxplot bar summarizes the bias of an estimator across $200$ datasets under a single sample size and hypothesis combination.
    The blue bar indicates bias under the full model, while the white bar indicates bias under the model employing a monotonicity assumption.
    Rows of the graph correspond to different hypotheses, while the columns correspond to whether or not \Cref{cond:a-indy} holds. }
  \label{fig:mse}
\end{figure}

\begin{figure}
  \centering
  \includegraphics[scale=0.7]{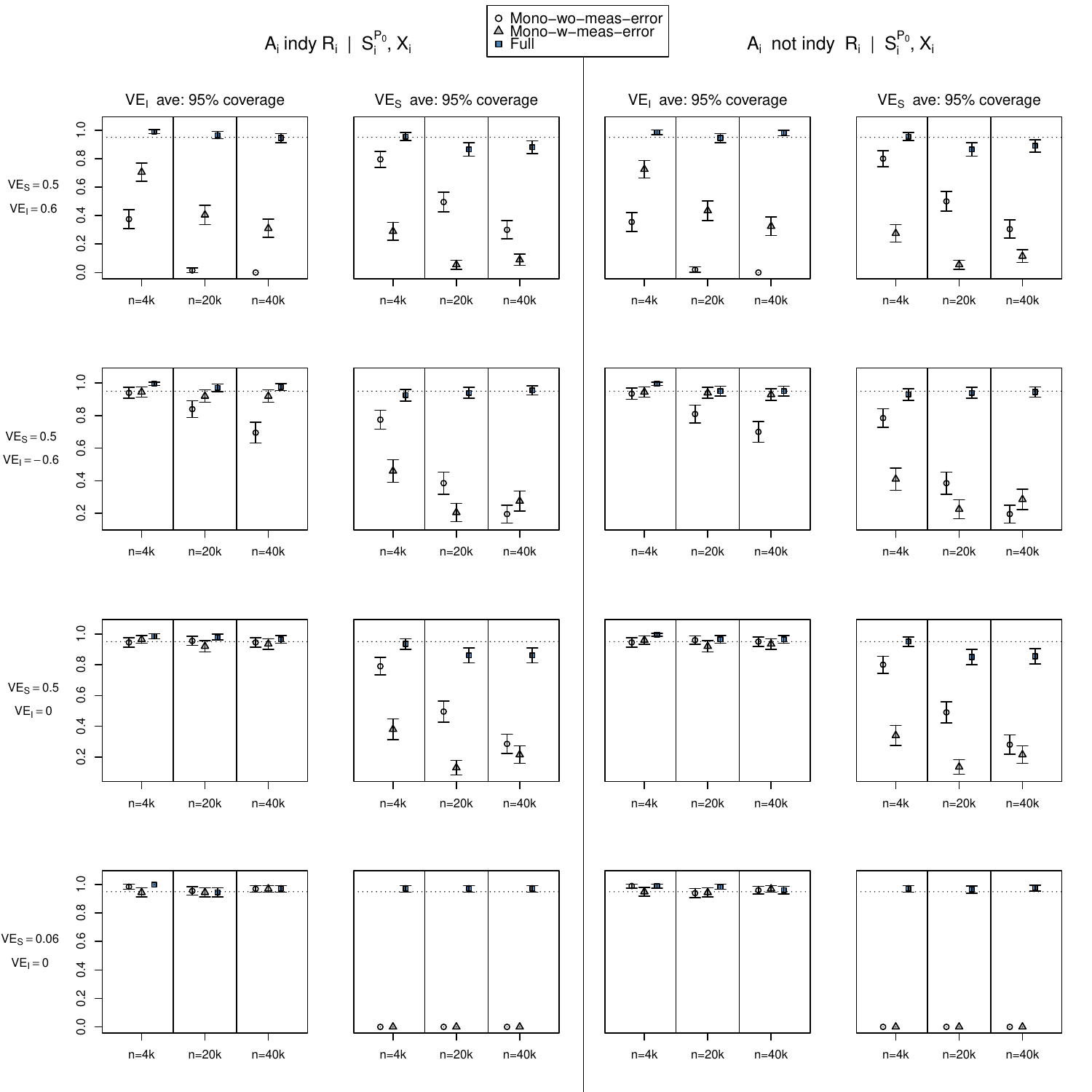}
  \caption{Posterior credible interval coverage across simulated datasets for $\mathrm{VE}_I$ and $\mathrm{VE}_{S}$ (95\% intervals).
    Each point and error bar summarizes the coverage of the credible $200$ datasets under a single sample size and hypothesis combination.
    The blue square points indicate the full model coverage, while the white points indicate the coverage under the model employing a monotonicity assumption.
    Rows of the graph correspond to different hypotheses, while the columns correspond to whether or not \Cref{cond:a-indy} holds.}
  \label{fig:coverage}
\end{figure}

The coverage plots in \Cref{fig:coverage} show that the posterior credible intervals generated by the incorrect models have coverage far below the nominal rates under both alternative hypotheses, but the coverage is at or above the nominal rates under both null hypotheses.
The full model achieves or exceeds the nominal coverage in all scenarios.
Because it is a stated goal of the model to jointly infer $\mathrm{VE}_S$ and $\mathrm{VE}_I$, the coverage plots show that it is at a minimum necessary to use a model that allows an unrestricted causal model to achieve nominal coverage.

The power calculations are presented in \Cref{tab:z2-size}, which shows power as a function of the sample size and the various hypotheses.
Despite the increased MSE under the null hypothesis, the full model does not have significantly inflated Type I errors compared to the \emph{mono} models.
The full model displays higher power uniformly across scenarios and sample sizes.
\begin{table}[H]
\centering
\caption{Power and Type I error rates for sample sizes of $4{,}000$ through $40{,}000$.
  Hypothesis indicates whether the alternative or the null hypothesis was used to generate the datasets, the column $A_i \indy R_i \mid S_i^{P_0}$ indicates whether \Cref{cond:a-indy} holds, and the Model column indicates whether the full model or the models that incorrectly assume monotonicity were fitted.
  The $k$ column indicates the cutoff value used to control Type I error; these values are specific to each model and are set so as to control the Type I error under the null across all sample sizes and $A_i \indy R_i \mid S_i^{P_0}$ scenarios.
  Values for $\mathrm{VE}_I = 0.6$ indicate power, while  $\mathrm{VE}_I = 0$ indicate Type I error.
}
\bgroup
\def\arraystretch{0.5}
\begin{tabular}{llll|rrr}
  \hline
  Hypothesis & $A_i \indy R_i \mid S^{P_0}$? & Model & $C$ & $4{,}000$ & $20{,}000$ & $40{,}000$ \\ 
  \hline
  \multirow{6}{*}{$\mathrm{VE}_I = 0.6$} & \multirow{3}{*}{$A_i \indy R_i \mid S^{P_0}$} & \emph{mono-wo-meas-error} & 0.925   & 0.23 & 0.77 & 0.90 \\ 
             &  & \emph{mono-w-meas-error} & 0.75 &  0.34 & 0.80 & 0.95\\ 
             &  & \emph{full}  & $0.85$ & 0.53 & 0.92 & 1.00 \\ 
             & \multirow{3}{*}{$A_i \not\indy R_i \mid S^{P_0}$} & \emph{mono-wo-meas-error} & 0.925 & 0.23 & 0.78 & 0.92\\ 
             & & \emph{mono-w-meas-error} & 0.75 & 0.35 & 0.81 & 0.96 \\ 
             & & \emph{full} & $0.85$ & 0.59 & 0.95 & 1.00 \\ 
  \hline
  \multirow{6}{*}{$\mathrm{VE}_I = 0$} & \multirow{3}{*}{$A_i \indy R_i \mid S^{P_0}$} & \emph{mono-wo-meas-error} & 0.925 &0.03 & 0.01 & 0.01 \\ 
             &  & \emph{mono-w-meas-error} &0.75 &0.03 & 0.02 & 0.04  \\ 
             &  & \emph{full} & 0.85 & 0.03 & 0.04 & 0.01 \\ 
             & \multirow{3}{*}{$A_i \not\indy R_i \mid S^{P_0}$} & \emph{mono-wo-meas-error} & 0.925 & 0.05 & 0.01 & 0.01\\ 
             &  & \emph{mono-w-meas-error} & 0.75 & 0.04 & 0.03 & 0.03\\ 
   &  & \emph{full} & 0.85 & 0.03 & 0.04 & 0.01 \\ 
  \hline
\end{tabular}
\egroup
\label{tab:z2-size}
\end{table}
While these results show that one needs fairly large sample sizes to achieve 80\% power for the estimands of interest in both scenarios, this is expected because the always-infected principal strata, $(1,1)$, are only $10\%$ of their respective populations in our simulation studies.
Despite the full model having more parameters than either the \emph{mono-w-meas-error} or \emph{mono-wo-meas-error}, it generates the most powerful test in all scenarios while adequately controlling the type I error.
In order to achieve Type I error control, the \emph{mono-wo-meas-error} model requires a much smaller rejection region, or, equivalently, a larger value of $C$ compared to the other two models.
This reduces its power, and results in the \emph{mono-wo-meas-error} having the smallest power among the models.

This highlights the extent to which power calculations, and bias and MSE for our models are dependent on principal strata proportions.
Even though the proportion of individuals who are harmed by the vaccine is only 2\%, the effect on models that assume monotonicity is quite drastic in terms of bias, MSE, and especially coverage.

Furthermore, though the sample sizes are large, randomized vaccine trials of similar magnitude have been run.
For example, the trial presented in \cite{polackSafetyEfficacyBNT162b22020} included approximately $43{,}500$ participants.
This highlights the fact that our model can be used to infer post-infection outcome vaccine efficacy from large real-world studies.

\subsubsection{Three-arm trial}

In the three arm trial, we sample the study-site-specific principal strata proportions from a Dirichlet distribution.
The distribution's mean corresponds to the vector:
\begin{align*}
  (&P(S_i^{P_0} = (0,0,0)), P(S_i^{P_0} = (1,0,0)), P(S_i^{P_0} = (0,1,0)), P(S_i^{P_0} = (1,1,0)), \\
  & P(S_i^{P_0} = (0,0,1)), P(S_i^{P_0} = (1,0,1)), P(S_i^{P_0} = (0,1,1)), P(S_i^{P_0} = (1,1,1))
\end{align*}
where
$S_i^{P_0} = (S_i(0) = i, S_i(1) = j, S_i(2) = k)$.
The mean is set to $(0.7, 0.13, 0.01, 0.01, 0.01, 0.01, 0.01, 0.12)$, while the variance for each category is 
$
P(S_i^{P_0} = (i,j,k))(1 - P(S_i^{P_0} = (i,j,k)))/101,
$
as above.
These parameter settings result in a population-average cumulative incidence of $19\%$.

We simulated data under two scenarios of varying study population size of $4{,}000$, $40{,}000$, and $80{,}000$ participants.
In all scenarios, $\mathrm{VE}_{S,21} = 0.45, \mathrm{VE}_{S,31} = 0.41$ and $\mathrm{VE}^{(1,1,1)}_{I,21} = 0$.
In one scenario,  $\mathrm{VE}^{(1,1,1)}_{I,31} = 0.52$, while $\mathrm{VE}^{(1,1,1)}_{I,31} = 0$ in the null scenario.
The data is generated so that \Cref{cond:a-indy} holds.

We fitted the same three models we fitted above, but the monotonicity assumption is more severe for the three-arm trial compared to the two-arm trial because only four parameters can be estimated from the observed data.
The monotonic models assume that
$$P(S_i^{P_0} = (0,1,0)) = P(S_i^{P_0} = (0,0,1)) = P(S_i^{P_0} = (1,0,1)) = P(S_i^{P_0} = (0,1,1)) = 0,$$
thus misclassifying about 4\% of the population.

In order to compare the models' performance, we measured $\mathrm{VE}^{(1,1,1)}_{I,21}$ and $\mathrm{VE}^{(1,1,1)}_{I,31}$, as well as conditional VE $\mathrm{VE}^{(1,1,1)}_{I,21}(\ell)$ and $\mathrm{VE}^{(1,1,1)}_{I,31}(\ell)$ for $\ell = 5$.
We also compared inferences for the VE against infection estimands, $\mathrm{VE}_{S,21}$ and $\mathrm{VE}_{S,31}$ respectively.
\begin{figure}
  \centering
  \includegraphics[scale=0.7]{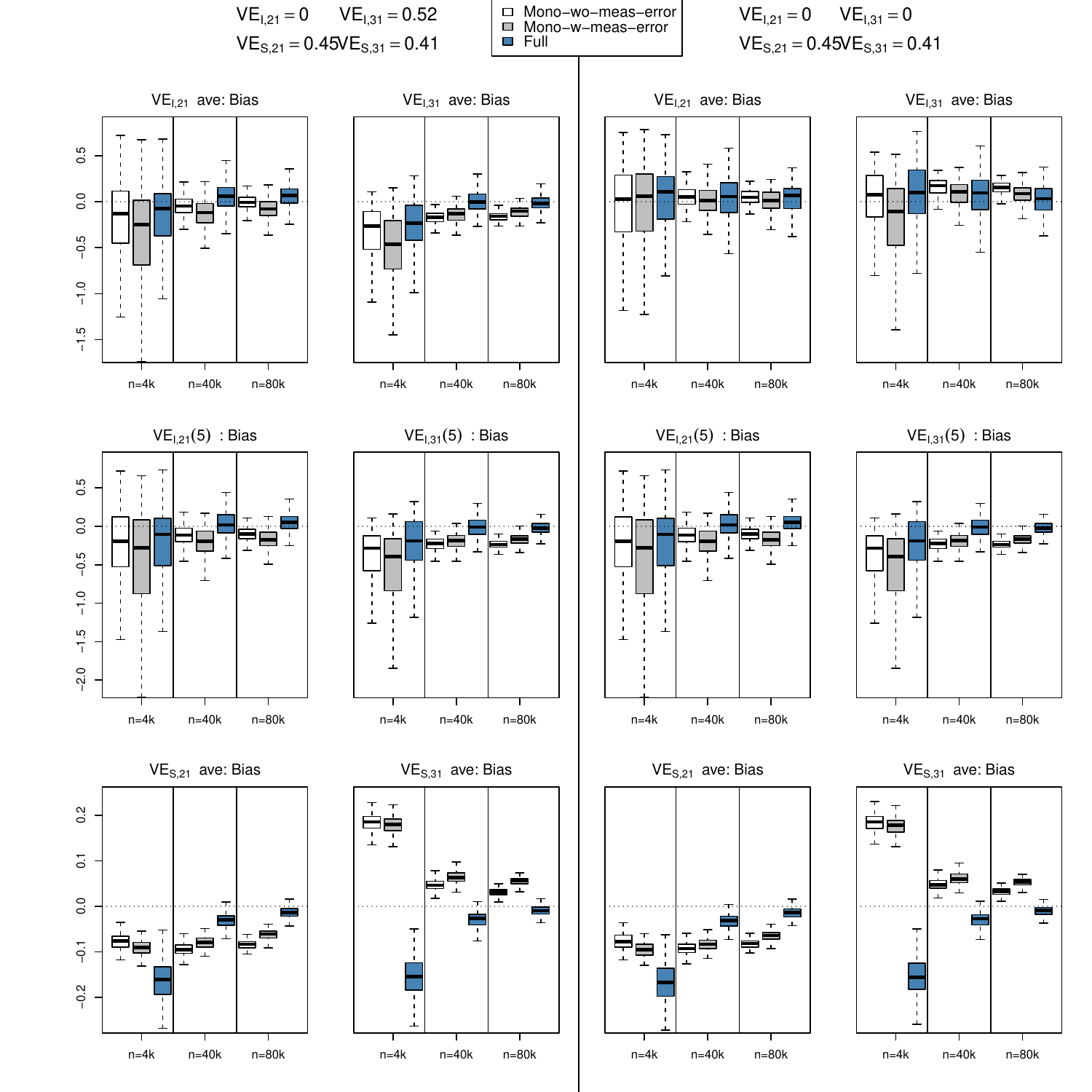}
  \caption{Boxplots collecting bias across simulated datasets for posterior median $\mathrm{VE}^{(1,1,1)}_{I,21}$ and $\mathrm{VE}^{(1,1,1)}_{I,31}$ in the top row,
    $\mathrm{VE}^{(1,1,1)}_{I,21}(k)$ and $\mathrm{VE}^{(1,1,1)}_{I,31}(k)$ in the middle row, and $\mathrm{VE}_{S,21}$ and $\mathrm{VE}_{S,31}$ in the final row.
    Each boxplot bar summarizes the bias of an estimator across $200$ datasets under a single sample size and hypothesis combination.
    The blue bar indicates bias under the full model, while the white bar indicates bias under the model employing a monotonicity assumption.
    The columns correspond to whether the null hypothesis of $\mathrm{VE}^{(1,1,1)}_{I,31} = 0$ holds.}
  \label{fig:bias-z3}
\end{figure}

\begin{figure}
  \centering
  \includegraphics[scale=0.7]{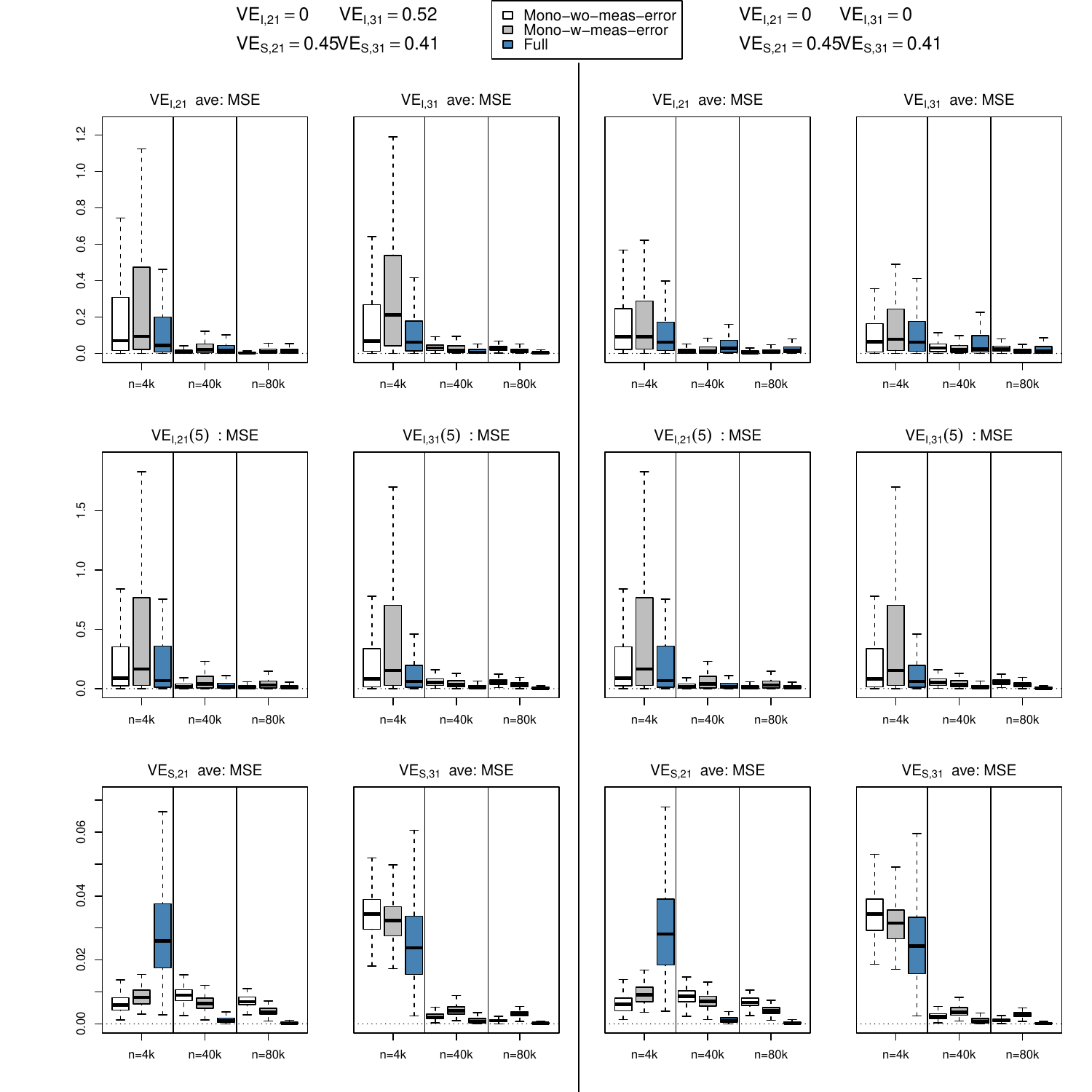}
  \caption{Boxplots collecting mean squared error across simulated datasets for posterior median $\mathrm{VE}^{(1,1,1)}_{I,21}$ and $\mathrm{VE}^{(1,1,1)}_{I,31}$ in the top row,
    $\mathrm{VE}^{(1,1,1)}_{I,21}(k)$ and $\mathrm{VE}^{(1,1,1)}_{I,31}(k)$ in the middle row, and $\mathrm{VE}_{S,21}$ and $\mathrm{VE}_{S,31}$ in the final row. 
    Each boxplot bar summarizes the bias of an estimator across $200$ datasets under a single sample size and hypothesis combination.
    The blue bar indicates bias under the full model, while the white bar indicates bias under the model employing a monotonicity assumption.
    The columns correspond to whether the null hypothesis of $\mathrm{VE}^{(1,1,1)}_{I,31} = 0$ holds.}
  \label{fig:mse-z3}
\end{figure}

\begin{figure}
  \centering
  \includegraphics[scale=0.7]{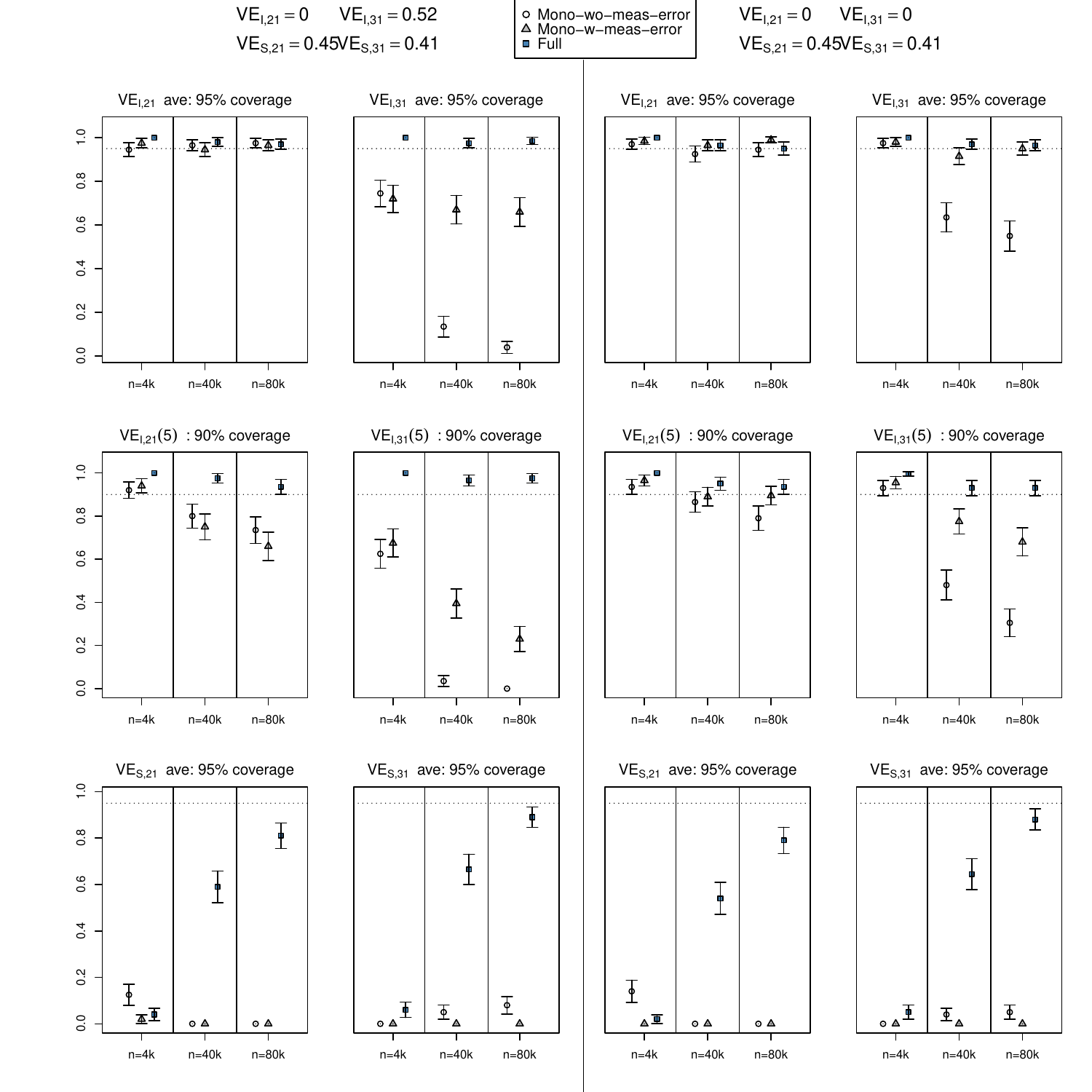}
  \caption{Posterior credible interval coverage across simulated datasets for $\mathrm{VE}^{(1,1,1)}_{I,21}$ and $\mathrm{VE}^{(1,1,1)}_{I,31}$ (95\% intervals)  in the top row,
    $\mathrm{VE}^{(1,1,1)}_{I,21}(k)$ and $\mathrm{VE}^{(1,1,1)}_{I,31}(k)$ in the middle row, and $\mathrm{VE}_{S,21}$ and $\mathrm{VE}_{S,31}$ in the final row (all 90\% intervals). 
    Each point and error bar summarizes the coverage of the credible $200$ datasets under a single sample size and hypothesis combination.
    The blue square points indicate the full model coverage, while the white points indicate the coverage under the model employing a monotonicity assumption.
    The columns correspond to whether the null hypothesis of $\mathrm{VE}^{(1,1,1)}_{I,31} = 0$ holds.}
  \label{fig:coverage-z3}
\end{figure}

\Cref{fig:bias-z3} shows that the bias of the \emph{full} model was the smallest in nearly every scenario for the chosen parameters.
Meanwhile, the \Cref{fig:mse-z3} results are more mixed; the figure shows that the MSE is smallest for the full model when sample size exceeds $4{,}000$ for the $\mathrm{VE}_S$ parameters, but that the MSE for the $\mathrm{VE}_{I}$ and conditional causal estimand is larger when the causal null hypothesis holds.
In \Cref{fig:coverage-z3} we see that the full model attains the nominal coverage for the $\mathrm{VE}_I$ parameters, while the monotonic models fail to achieve nominal coverage for these parameters.
All models fail to achieve the nominal coverage for the $\mathrm{VE}_S$ parameters, but the \emph{full} model's coverage improves as the sample size increases.

\begin{table}[H]
\centering
\caption{Power and Type I error rates for sample sizes of $4{,}000$ through $80{,}000$ for the three-arm treatment simulation study.
  Hypothesis indicates whether the alternative or the null hypothesis was used to generate the datasets, the column $A_i \indy R_i \mid S_i^{P_0}$ indicates whether \Cref{cond:a-indy} holds, and the Model column indicates whether the full model or the two models that incorrectly assume monotonicity were fitted.
  The $k$ column shows the cutoff value that determines the rejection region; the value is chosen so as to control Type I error across all sample sizes for each model.
  Values for $\mathrm{VE}^{(1,1,1)}_{I,31} = 0.52$ indicate power, while  $\mathrm{VE}^{(1,1,1)}_{I,31} = 0$ indicate Type I error. By design, all methods will have Type I error less than or equal to 5\%.
}
\bgroup
\def\arraystretch{0.5}
\begin{tabular}{llll|rrr}
  \hline
  Hypothesis & $A_i \indy R_i \mid S^{P_0}$? & Model & $C$ & $4{,}000$ & $40{,}000$ & $80{,}000$ \\ 
  \hline
  \multirow{3}{*}{$\mathrm{VE}^{(1,1,1)}_{I,31} = 0.52$} & \multirow{3}{*}{$A_i \indy R_i \mid S^{P_0}$} & \emph{mono-wo-meas-error} & 0.99  & 0.02 & 0.68 & 0.94\\ 
              &  & \emph{mono-w-meas-error} & 0.94 & 0.02 & 0.75 & 0.97\\ 
              &  & \emph{full}& 0.93 & 0.00 & 0.83 & 0.98  \\ 
  \hline
  \multirow{3}{*}{$\mathrm{VE}^{(1,1,1)}_{I,31} = 0$} & \multirow{3}{*}{$A_i \indy R_i \mid S^{P_0}$} & \emph{mono-wo-meas-error} & 0.99 & 0.00 & 0.03 & 0.03 \\ 
              &  & \emph{mono-w-meas-error}& 0.94 &  0.00 & 0.04 & 0.03\\ 
              &  & \emph{full} & 0.93 & 0.00 & 0.04 & 0.03 \\ 
  \hline
\end{tabular}
\egroup
\label{tab:z3-size}
\end{table}

\Cref{tab:z3-size} shows the pitfalls of the smaller MSE values from the \emph{mono-wo-meas-error} model, namely that in order to control the Type I error we need a much smaller rejection region.
This leads to lower power compared to the other two models.
In keeping with the results in the two-arm trial, the full model has the highest power among all the models, with the exception of the $4{,}000$ participant scenario.
The results show that one could run a trial with $40{,}000$ to achieve 80\% power.
Recent vaccine trials show that trials of $40{,}000$ are feasible and successful, further demonstrating the applicability of our model to real-world trials.

\section{Discussion}\label{sec:discuss-ve}

Policymakers and public health experts can use vaccine efficacy for post-infection outcomes to design more precise vaccination programs.
Our method makes inferring these causal estimands feasible in real-world multi-arm trials where outcomes are measured with error and vaccines cannot be assumed to have a nonnegative effect on infection for every individual.
The power of our method lies in its ability to be applied to vaccine trials with multiple treatments as well as various post-infection outcomes.
Although we focus much of the paper on the binary severe illness outcome, our method is readily extensible to other binary post-infection outcomes like secondary transmission, hospitalization, or death, and to ordinal and continuous measures like immune response as measured by antibody titer or viral load.
Accordingly, when paired with a parametric likelihood for continuous post-infection outcomes, our method may be more statistically efficient than models identified by likelihood assumptions alone, like that of \cite{zhangLikelihoodBasedAnalysisCausal2009}. 
Furthermore, our identifiability results are nonparametric, though we use parametric Bayesian models to design estimators in our examples.
One can use these methods to design and analyze clinical trials, as we show in \Cref{sec:ve-studies}.

\subsection{Limitations and extensions}

The nondifferential error assumptions may be violated in certain scenarios, so extending the measurement error model to reflect more complex measurement processes is an important future direction.
A simple extension would be to allow infection test sensitivity and specificity to condition on $X$ and $Z$.
For instance, if a vaccine changes how the virus populates the nasal cavity, we might expect that PCR tests from nasopharyngeal swabs will be less sensitive in the vaccinated group.
The model currently assumes that post-infection outcome specificity is constant across study sites, but this may not reflect the reality of some vaccine trials.
An extension to the model would be to allow misclassification rates that differ by study site $R$, but more work is needed to understand the impact of this model expansion.
Further work is also needed to generalize the procedure to categorical intermediate outcomes, which would allow for more general vaccine efficacy against transmission study designs \citep{vanderweeleBoundingInfectiousnessEffect2011}, as well as applications beyond vaccine efficacy to noncompliance in multi-arm trials where the exclusion restriction could be violated \citep{chengBoundsCausalEffects2006}. 

\begin{appendices}
\section{Proofs and further details for simulation studies}

We define our notation for principal stratification in vaccine efficacy (VE) in section \ref{subsec:defn}.
In section \ref{sec:krank}, we give general properties of the Kruskal rank, and extensions to \cite{kruskalThreewayArraysRank1977} theorems that we derived.
We apply these extensions in the context of principle stratification for VE in section \ref{subsec:krank-ve}.
The proof of our main result, \Cref{thm:id-noisy}, is given in section \ref{sec:main-res}.
These proofs are based on results in \Cref{sec:krank} and \Cref{subsec:krank-ve}.

\section{Notation and definitions}\label{subsec:defn}

In the following proofs, we have omitted the subscript $i$ from random variables to simplify our notation.
We have also elided conditioning on $X_i = x$; the proofs shown in \Cref{proof:id-noisy} are understood to be conditional on $X_i = x$.
Let $z$ be the $N_z$-category discrete variable taking values in the set $\{z_1,\dots, z_{N_z}\}$ representing treatment, and let $Z$ be treatment assignment.  
The principal stratum, $S^{P_0}$ is defined as $(S(z_1),\dots,S(z_{N_z}))$, $S(z_j) \in \{0,1\}$.
Let $\mathcal{S}$ be the set of principal strata, which is equal to $\{0,1\}^{N_z}$ when there are no monotonicity assumptions; let $u \in \mathcal{S}$.

Let the set of treatments be $\{z_1, \dots, z_{N_z}\}$, with $z \in \{z_1, \dots, z_{N_z}\}$

Let $A$ have $N_a$ levels and take values in the set $\{1, \dots, N_a\}$. 
Let $P(A \mid R)$ be the $N_a \times N_r$ matrix with $(i, j)^\mathrm{th}$ element equal to $P(A = i \mid R = j)$ and let $P(A = k \mid R)$ be the $N_r \times N_r$ diagonal matrix with $(i,i)^\mathrm{th}$ diagonal element $P(A = k \mid R = i)$. Let $P(\tilde{A} = k \mid R)$ be defined similarly.
Let $P_{N_z}(A \mid S^{P_0})$ be the $N_a \times 2^{N_z}$ matrix with $(i, j)^\mathrm{th}$ element equal to $P(A = i \mid S^{P_0} = \varpi_{N_z}(j-1))$, and let $P_{N_z}(S^{P_0} \mid R)$ be the $2^{N_z} \times N_r$ matrix with $(i,j)^\mathrm{th}$ element equal to $P(S^{P_0} = \varpi_{N_z}(i-1) \mid R = j)$. 
Let $P_{N_z}(A = k \mid S^{P_0})$ be the $2^{N_z} \times 2^{N_z}$ diagonal matrix with $(i,i)^{\mathrm{th}}$ element $P(A = k \mid S^{P_0} = \varpi_{N_z}(i-1))$ and let $P_{N_z}(\tilde{A} = k \mid S^{P_0})$ be defined similarly.
Let $P(y \mid R, Z = z)$ be the $1 \times R$ matrix with element $(1, j)^\mathrm{th}$ equal to $P(y \mid R = j, Z = z)$, and similarly let $P_{N_z}(y \mid S^{P_0}, Z = z)$ be the $1 \times 2^{N_z}$ matrix with element $(1, j)^\mathrm{th}$ equal to $P(y \mid S^{P_0} = \varpi_{N_z}(j-1), Z = z)$.
Let $P(y \mid R, Z = z, A = k)$ be the $1 \times R$ matrix with $(1, j)^\mathrm{th}$ element equal to $P(y \mid R = j, Z = z, A = k)$, and similarly let $P_{N_z}(y \mid S^{P_0}, Z = z, A = k)$ be the $1 \times 2^{N_z}$ matrix with element $(1, j)^\mathrm{th}$ equal to $P(y \mid  S^{P_0} = \varpi_{N_z}(j-1), Z = z, A = k)$.
Let the matrix $P_{N_z}(S \mid Z, S^{P_0})$ be in $\R^{2 N_z \times 2^{N_z}}$ where column denotes principal stratum $S^{P_0} = \varpi_{N_z}(j-1)$ and row represents a combination $(s, z) \in \{(1,1),(1,2),\dots,(1,N_z),(0,1),\dots,(0,N_z)\}$, with $(i,j)^\mathrm{th}$ element denoted $P_{N_z}(S \mid Z, S^{P_0})_{ij}$ defined as
\[
P_{N_z}(S \mid Z, S^{P_0})_{ij} = \varpi_{N_z}(j-1)_i \ind{i \leq N_z} + (1 - \varpi_{N_z}(j-1)_{i-N_z}) \ind{i > N_z},
\]
and let $P_{N_z}(\tilde{S} \mid Z, S^{P_0})$ be in $\R^{2 N_z \times 2^{N_z}}$ with $(i,j)^\mathrm{th}$ element denoted $P_{N_z}(\tilde{S} \mid Z, S^{P_0})_{ij}$ defined:
\begin{align*}
P_{N_z}(\tilde{S} \mid Z, S^{P_0})_{ij} & = \mathrm{sn}_S^{\varpi_{N_z}(j-1)_i}(1 - \mathrm{sp}_S)^{1 - \varpi_{N_z}(j-1)_i} \ind{i \leq N_z} \\ 
& + (1-\mathrm{sn}_S)^{\varpi_{N_z}(j-1)_{i-N_z}}\mathrm{sp}_S^{1 - \varpi_{N_z}(j-1)_{i-N_z}} \ind{i > N_z}.
\end{align*}
Let $B^+$ be the Moore-Penrose inverse of the matrix $B$, $\mathbf{1}_m$ be the $m$-vector of $1$s, $\mathbf{0}_m$ be the $m$-vector of $0$s, and $\mathbf{I}_m$ be the $m \times m$ dimensional identity matrix.

\section{Kruskal rank properties related to VE}\label{subsec:krank-ve}
In this section, we show that (a) the Kruskal rank of the matrix $P_{N_z}(\tilde{S} \mid Z, S^{P_0})$ is $3$ for $N_z \geq 2$ when $\mathrm{sn}_S + \mathrm{sp}_S \neq 1$ and (b) the column domains of $P_{N_z}(\tilde{S} \mid Z, S^{P_0})$ are not invariant to column permutation when  $\mathrm{sn}_S, \mathrm{sp}_S > 0.5$ or $\mathrm{sn}_S, \mathrm{sp}_S < 0.5$ for $N_z \geq 2$.
\begin{lemma}[Kruskal rank $P_2(\tilde{S} \mid Z,S^{P_0})$ ]\label{lemma:p_s_tilde}
The Kruskal rank of 
\begin{align}\label{eq:p_s_tilde}
  \begin{blockarray}{ccccc}
 (0,0) & (1,0) & (0,1) & (1,1)  \\
\begin{block}{[cccc]c}
 1 - \mathrm{sp}_S &  \mathrm{sn}_S & 1-\mathrm{sp}_S & \mathrm{sn}_S       & (s=1,z=1) \\
 1 - \mathrm{sp}_S &  1 - \mathrm{sp}_S & \mathrm{sn}_S & \mathrm{sn}_S     & (s=1,z=2) \\
 \mathrm{sp}_S     &  1 - \mathrm{sn}_S & \mathrm{sp}_S & 1 - \mathrm{sn}_S & (s=0,z=1) \\
 \mathrm{sp}_S     &  \mathrm{sp}_S & 1 - \mathrm{sn}_S & 1 - \mathrm{sn}_S & (s=0,z=2) \\
\end{block}
\end{blockarray}
\end{align}

is $3$ as long as $\mathrm{sn}_S + \mathrm{sp}_S \neq 1$.
\end{lemma}
\begin{proof}
All subsets of $3$ columns of the matrix $P_2(\tilde{S} \mid Z, S^{P_0})$ are of the form: 
\begin{align} \label{eq:minor-1}
  \begin{bmatrix}
    a & c & e \\
    b & d & f \\
    1 - a & 1 - c & 1 - e  \\
    1 - b & 1 - d & 1 - f 
  \end{bmatrix}.
\end{align}
These submatrices have a common maximal minor of 
$$ 
a (d - f) - c(b - f) + e(b - d).
$$
The quantities $a,b,c,d,e,f$ are the elements of the $2\times 3$ matrix
\begin{align}
  \begin{blockarray}{ccc}
\begin{block}{[ccc]}
  a & c & e\\
  b & d & f \\
\end{block}
\end{blockarray}
\end{align}
in which $(a,b)^T,(c,d)^T,(e,f)^T$ are any $3$ columns drawn without replacement from the $2 \times 4$ submatrix of \Cref{eq:p_s_tilde}:
\begin{align}
  \begin{blockarray}{cccc}
\begin{block}{[cccc]}
 1 - \mathrm{sp}_S & \mathrm{sn}_S & 1-\mathrm{sp}_S & \mathrm{sn}_S     \\
 1 - \mathrm{sp}_S & 1 - \mathrm{sp}_S & \mathrm{sn}_S & \mathrm{sn}_S    \\
\end{block}
\end{blockarray}.
\end{align}
These minors are all equal to (up to a factor of $-1$):
$$
(1 - \mathrm{sn}_S  - \mathrm{sp}_S)^2,
$$
which can be seen after a brute-force calculation.
The minors are nonzero for all $\mathrm{sn}_S, \mathrm{sp}_S \in [0,1]$ such that $\mathrm{sn}_S + \mathrm{sp}_S \neq 1$. Thus, by the determinantal rank definition, all $3$ column matrices are rank $3$.
In contrast, the determinant of $P_2(\tilde{S} \mid Z,S^{P_0})$ is $0$ for all values of $\mathrm{sn}_S,\mathrm{sp}_S$.
Thus by the definition of Kruskal rank in \Cref{defn:krank}, $k_{P_2(\tilde{S} \mid Z,S^{P_0})} = 3$.
\end{proof}

\begin{lemma}[Kruskal rank $P_{N_z}(\tilde{S} \mid Z, S^{P_0}),\, N_z \geq 2$ ]\label{lemma:p_s_tilde-gen-z}
The Kruskal rank of $P(\tilde{S} \mid Z, S^{P_0})$ for $N_z \geq 2$ is $3$
as long as $\mathrm{sn}_S + \mathrm{sp}_S \neq 1$.
\end{lemma}
\begin{proof}
We proceed by induction. For $N_z = 2$, \Cref{lemma:p_s_tilde} shows that the Kruskal rank is $3$.
Let $N_z = n$ for $n > 2$. Recall that $P_n(\tilde{S} \mid Z, S^{P_0})$ is the $2n \times 2^{n}$ matrix with column $j$
$$
\begin{bmatrix}
s_j \\
\mathbf{1}_n - s_j
\end{bmatrix}
$$
with the $i^\mathrm{th}$ element of $s_j$ denoted $s_{ij}$ and defined as:
$$
s_{ij} = \mathrm{sn}_S^{\varpi_n(j-1)_i} (1 - \mathrm{sp}_S)^{1-\varpi_n(j-1)_i}.
$$
The induction hypothesis is that the Kruskal rank of $P_{n}(\tilde{S} \mid Z, S^{P_0})$ is $3$.
The columns of $P_{n+1}(\tilde{S} \mid Z, S^{P_0})$ are of the form
$$
\begin{bmatrix}
s_j \\
1 - \mathrm{sp}_S\\
\mathbf{1}_n - s_j \\
\mathrm{sp}_S
\end{bmatrix}
$$
for $j \in \{1, \dots, 2^{n}\}$,
and
$$
\begin{bmatrix}
s_{j-2^{n}} \\
\mathrm{sn}_S\\
\mathbf{1}_n - s_{j-2^{n}} \\
1 - \mathrm{sn}_S
\end{bmatrix}
$$
for $j \in \{2^{n} + 1, \dots, 2^{n+1}\}$.
The $3$-column submatrices of $P_{N_z}(\tilde{S} \mid Z, S^{P_0})$ made from column $j,\ell,m$ indices fall into several classes.
When $j,\ell,m \in \{1, \dots, 2^{n}\}, \,j,\ell,m \in \{2^{n} + 1, \dots, 2^{n+1}\}$ or $j,\ell \in \{1, \dots, 2^{n}\}, m \in \{2^{n} + 1, \dots, 2^{n+1}\} \setminus \{j + 2^{n}, \ell + 2^{n}\}$, $j \in \{1, \dots, 2^{n}\}, m, \ell \in \{2^{n} + 1, \dots, 2^{n+1}\} \setminus \{j + 2^{n}\}$  all matrices are rank $3$ by the induction hypothesis. 
When $j,\ell \in \{1, \dots, 2^{n}\}$ but $m \in \{j + 2^{n}, \ell + 2^{n}\}$ the submatrix is
\begin{align*}
\begin{bmatrix}
s_j & s_\ell & s_{m-2^{n}} \\
1 - \mathrm{sp}_S & 1 - \mathrm{sp}_S&\mathrm{sn}_S\\
\mathbf{1}_n - s_j & \mathbf{1}_n - s_\ell & \mathbf{1}_n - s_{m-2^{n}} \\
\mathrm{sp}_S & \mathrm{sp}_S & 1 - \mathrm{sn}_S
\end{bmatrix}.
\end{align*}
WLOG, let $m = j + 2^{n}$. This leads to the submatrix:
\begin{align*}
\begin{bmatrix}
s_j & s_\ell & s_{j} \\
1 - \mathrm{sp}_S & 1 - \mathrm{sp}_S&\mathrm{sn}_S\\
\mathbf{1}_n - s_j & \mathbf{1}_n - s_\ell & \mathbf{1}_n - s_{j} \\
\mathrm{sp}_S & \mathrm{sp}_S & 1 - \mathrm{sn}_S
\end{bmatrix}.
\end{align*}
The rank of this submatrix is 
\begin{align*}
\mathrm{rank}\begin{bmatrix}
s_j & s_\ell & s_{j} \\
1 - \mathrm{sp}_S & 1 - \mathrm{sp}_S&\mathrm{sn}_S\\
\mathbf{1}_n - s_j & \mathbf{1}_n - s_\ell & \mathbf{1}_n - s_{j} \\
\mathrm{sp}_S & \mathrm{sp}_S & 1 - \mathrm{sn}_S
\end{bmatrix} & = 
\mathrm{rank}\begin{bmatrix}
s_j & s_\ell & s_{j} \\
1 - \mathrm{sp}_S & 1 - \mathrm{sp}_S&\mathrm{sn}_S\\
\mathbf{1}_n - s_j & \mathbf{1}_n - s_\ell & \mathbf{1}_n - s_{j} \\
1 & 1 & 1
\end{bmatrix} \\
& = 
\mathrm{rank}\begin{bmatrix}
s_j & s_\ell & s_{j} \\
\mathbf{1}_n - s_j & \mathbf{1}_n - s_\ell & \mathbf{1}_n - s_{j} \\
1 - \mathrm{sp}_S & 1 - \mathrm{sp}_S&\mathrm{sn}_S\\
0 & 0 & 1 - \frac{\mathrm{sn}_S}{1 - \mathrm{sp}_S}
\end{bmatrix} \\
& \geq 
\mathrm{rank}\begin{bmatrix}
s_j & s_\ell            \\ 
\mathbf{1}_n - s_j & \mathbf{1}_n - s_\ell      \\
1 - \mathrm{sp}_S & 1 - \mathrm{sp}_S\\
0 & 0    
\end{bmatrix} + \mathrm{rank}(1 - \frac{\mathrm{sn}_S}{1 - \mathrm{sp}_S}) \\
& = 3.
\end{align*}
The inequality follows from \Cref{lemma:block-rank-alt}.
Other scenarios follow similarly.
\end{proof}

\begin{lemma}[Domain restriction lemma]\label{lemma:domain}
If $\mathrm{sn}_S, \mathrm{sp}_S \in [0,0.5)$ or $\mathrm{sn}_S, \mathrm{sp}_S \in (0.5,1]$, the matrix $P_{N_z}(\tilde{S} \mid Z, S^{P_0}) \in \R^{2N_z \times 2^{N_z}}$ has column domains that are not invariant to column permutation.
\end{lemma}
\begin{proof}
We prove \Cref{lemma:domain} by induction on $N_z$. The base case is $N_z = 2$. Let $P$ be a $4\times 4$ permutation matrix
and let $P_{2}(\tilde{S} \mid Z, S^{P_0})$ be
\begin{align}
\begin{blockarray}{ccccc}
(0,0) & (1,0) & (0,1) & (1,1)   \\
\begin{block}{[cccc]c}
1 - \mathrm{sp}_S &   \mathrm{sn}_S & 1 - \mathrm{sp}_S & \mathrm{sn}_S     & (s=1,z=1) \\
1 - \mathrm{sp}_S &   1 - \mathrm{sp}_S & \mathrm{sn}_S & \mathrm{sn}_S     & (s=1,z=2) \\
\mathrm{sp}_S     &   1 - \mathrm{sn}_S & \mathrm{sp}_S & 1 - \mathrm{sn}_S & (s=0,z=1)     \\
\mathrm{sp}_S     &   \mathrm{sp}_S & 1 - \mathrm{sn}_S & 1 - \mathrm{sn}_S &  (s=0,z=2)     \\
\end{block}
\end{blockarray}
\end{align}
Recall from the definition in \Cref{subsec:defn} that the column indices $\{1,2,3,4\}$ of $P_2(\tilde{S} \mid Z, S^{P_0})$ map to the following principal strata $S^{P_0}$: $\varpi_2(0), \varpi_2(1), \varpi_2(2), \varpi_2(3)$.
In other words, column index $j$ is mapped to $S^{P_0}$ via the relation $\varpi_2(j-1)$.
We consider permutation matrix $P$ without loss of generality, and other cases are similarly shown,
$$ 
P = 
\begin{bmatrix}
0 & 0 & 0 & 1 \\
0 & 1 & 0 & 0 \\
0 & 0 & 1 & 0 \\
1 & 0 & 0 & 0
\end{bmatrix}
$$
Let $\mathcal{C} = [0, 1]$, and let $\mathcal{A}$ be one of two half intervals of $[0,1]$: $[0,0.5)$ or $(0.5,1]$. 
Let $\mathcal{B} = \mathcal{C} \setminus \mathcal{A}$. 
Note that $P_2(\tilde{S} \mid Z, S^{P_0})$ maps $\mathcal{C}\times\mathcal{C}$ to a matrix with elements in $\mathcal{C}$. 
Let $1 - \mathrm{sp}_S \in \mathcal{A}$ and let $\mathrm{sn}_S \in \mathcal{B}$ and suppose that the column domains for $P_2(\tilde{S} \mid Z, S^{P_0})$ are not invariant after permutation by matrix $P$. 
Then we have the following domain for the map given by $P_2(\tilde{S} \mid Z, S^{P_0})$:
$$ 
P_2(\tilde{S} \mid Z, S^{P_0})\mid_{\mathcal{A} \times \mathcal{B}}: \mathcal{A} \times \mathcal{B} \to
\begin{blockarray}{cccc}
(0,0) & (1,0) & (0,1) & (1,1)   \\
\begin{block}{[cccc]}
\mathcal{A} &   \mathcal{A} & \mathcal{A} & \mathcal{A}      \\
\mathcal{A} &   \mathcal{A} & \mathcal{A} & \mathcal{A}      \\
\mathcal{B}     &   \mathcal{B} & \mathcal{B} & \mathcal{B}      \\
\mathcal{B}     &   \mathcal{B} & \mathcal{B} & \mathcal{B} )     \\
\end{block}
\end{blockarray}
$$
However, we have,
\begin{align}
    \bar{P}_2(\tilde{S} \mid Z, S^{P_0})\mid_{\mathcal{A} \times \mathcal{B}} & =P_2(\tilde{S} \mid Z, S^{P_0})\mid_{\mathcal{A} \times \mathcal{B}} P\\
    & =   
\begin{blockarray}{cccc}
(0,0) & (1,0) & (0,1) & (1,1)   \\
\begin{block}{[cccc]}
1 - \mathrm{sp}_S &   \mathrm{sn}_S & 1 - \mathrm{sp}_S & \mathrm{sn}_S     \\
1 - \mathrm{sp}_S &   1 - \mathrm{sp}_S & \mathrm{sn}_S & \mathrm{sn}_S     \\
\mathrm{sp}_S     &   1 - \mathrm{sn}_S & \mathrm{sp}_S & 1 - \mathrm{sn}_S     \\
\mathrm{sp}_S     &   \mathrm{sp}_S & 1 - \mathrm{sn}_S & 1 - \mathrm{sn}_S      \\
\end{block}
\end{blockarray}
    \begin{bmatrix}
0 & 0 & 0 & 1 \\
0 & 1 & 0 & 0 \\
0 & 0 & 1 & 0 \\
1 & 0 & 0 & 0
\end{bmatrix}\\
& = 
\begin{blockarray}{cccc}
(0,0) & (1,0) & (0,1) & (1,1)   \\
\begin{block}{[cccc]}
\mathrm{sn}_S     & 1 - \mathrm{sp}_S &   \mathrm{sn}_S & 1 - \mathrm{sp}_S \\
\mathrm{sn}_S     & 1 - \mathrm{sp}_S &   1 - \mathrm{sp}_S & \mathrm{sn}_S \\
1 - \mathrm{sn}_S & \mathrm{sp}_S     &   1 - \mathrm{sn}_S & \mathrm{sp}_S     \\
1 - \mathrm{sn}_S & \mathrm{sp}_S     &   \mathrm{sp}_S & 1 - \mathrm{sn}_S      \\
\end{block}
\end{blockarray}.
\end{align}
But we see that the column domains are invariant after column permutation:
$$ 
\bar{P}_2(\tilde{S} \mid Z, S^{P_0})\mid_{\mathcal{A} \times \mathcal{B}}: \mathcal{A} \times \mathcal{B} \to
\begin{blockarray}{ccccc}
(0,0) & (1,0) & (0,1) & (1,1)   \\
\begin{block}{[cccc]c}
\mathcal{A} &   \mathcal{A} & \mathcal{A} & \mathcal{A}     & (s=1,z=1) \\
\mathcal{A} &   \mathcal{A} & \mathcal{A} & \mathcal{A}     & (s=1,z=2) \\
\mathcal{B}     &   \mathcal{B} & \mathcal{B} & \mathcal{B} & (s=0,z=1)     \\
\mathcal{B}     &   \mathcal{B} & \mathcal{B} & \mathcal{B} &  (s=0,z=2)     \\
\end{block}
\end{blockarray}
$$
In order for the columns  $\bar{P}_2(\tilde{S} \mid Z, S^{P_0})\mid_{\mathcal{A} \times \mathcal{B}}$ to be on the same domain as $P_2(\tilde{S} \mid Z, S^{P_0})\mid_{\mathcal{A} \times \mathcal{B}}$, a necessary and sufficient condition is that $\mathrm{sn}_S$ and $1 - \mathrm{sp}_S$ are on the same domain.
In other words, $\{\mathrm{sn}_S \in \mathcal{A}, \mathrm{sp}_S \in \mathcal{B}\}$ or $\{\mathrm{sn}_S \in \mathcal{B}, \mathrm{sp}_S \in \mathcal{A}\}$. 

Thus $\bar{P}_2(\tilde{S} \mid Z, S^{P_0})\mid_{\mathcal{A} \times \mathcal{B}}$ maps $(\mathrm{sn}_S, \mathrm{sp}_S)$ to the same space that $P_2(\tilde{S} \mid Z, S^{P_0})\mid_{\mathcal{A} \times \mathcal{B}}$. 
We contradict our statement that the columns are not invariant to permutation.

The case for $N_z > 2$. Let $N_z = n > 2$ and let the column domains of $P_n(\tilde{S} \mid Z, S^{P_0})$ be not invariant to permutation.
Furthermore suppose that $\mathrm{sn}_S, \mathrm{sp}_S \in \mathcal{A}$ or $\mathrm{sn}_S, \mathrm{sp}_S \in \mathcal{B}$.
Then matrix $P_{n+1}(\tilde{S} \mid Z, S^{P_0})$ has columns 
$$
\begin{bmatrix}
s_j \\
1 - \mathrm{sp}_S\\
\mathbf{1}_n - s_j \\
\mathrm{sp}_S
\end{bmatrix}
$$
for $j \in \{1, \dots, 2^{n}\}$
and
$$
\begin{bmatrix}
s_{j-2^{n}} \\
\mathrm{sn}_S\\
\mathbf{1}_n - s_{j-2^{n}} \\
1 - \mathrm{sn}_S
\end{bmatrix}
$$
for $j \in \{2^{n} + 1, \dots, 2^{n+1}\}$.
Permuting any two columns $j, k \in \{1, \dots, 2^{n}\}$  or $j, k \in \{2^{n} + 1, \dots, 2^{n+1}\}$ yields different column domains given the induction hypothesis.
If $j \in \{1, \dots, 2^{n}\}$ and $k = j + 2^{n}$, then the columns are 
$$
\begin{bmatrix}
s_j & s_j\\
1 - \mathrm{sp}_S & \mathrm{sn}_S\\
\mathbf{1}_n - s_j &\mathbf{1}_n - s_j \\
\mathrm{sp}_S & 1 -  \mathrm{sn}_S
\end{bmatrix}
$$
Let the domain of $s_j$ be $\mathcal{D}$, and let $\mathcal{D}^{\mathsf{c}} = [0,1]^n \setminus \mathcal{D}$ be the domain of $\mathbf{1}_n - s_j$. Then the domains are
$$
\begin{bmatrix}
\mathcal{D} & \mathcal{D}\\
\mathcal{A} & \mathcal{B}\\
\mathcal{D}^{\mathsf{c}} &\mathcal{D}^{\mathsf{c}} \\
\mathcal{B}& \mathcal{A}
\end{bmatrix}
$$ if $\mathrm{sp}_S,\mathrm{sn}_S \in \mathcal{B}$ and  
$$
\begin{bmatrix}
\mathcal{D} & \mathcal{D}\\
\mathcal{B} & \mathcal{A}\\
\mathcal{D}^{\mathsf{c}} &\mathcal{D}^{\mathsf{c}} \\
\mathcal{A}& \mathcal{B}
\end{bmatrix}
$$ if $\mathrm{sp}_S,\mathrm{sn}_S \in \mathcal{A}$.  
These two columns are not invariant to permutation.
Because no two columns may be interchanged without a change in domain, right multiplying $P_{n+1}(\tilde{S} \mid Z, S^{P_0})$ by any $2^{n+1} \times 2^{n+1}$ permutation matrix $P\neq \mathbf{I}_{n+1}$ to will yield a matrix with different column domains than $P_{n+1}(\tilde{S} \mid Z, S^{P_0})$.
\end{proof}

\section{Rank properties related to VE}\label{subsec:rank-ve}

In this section we show that when $N_z \geq 2$ the rank of  $P_{N_z}(\tilde{S} \mid Z, S^{P_0}) = N_z + 1$ when $\mathrm{sn}_S + \mathrm{sp}_S \neq 1$.
\begin{lemma}[Rank $P_2(\tilde{S} \mid Z,S^{P_0})$ ]\label{lemma:rank-p_s_tilde}
  The rank of $P_2(\tilde{S} \mid Z,S^{P_0})$, defined in \Cref{eq:p_s_tilde}, is $3$ as long as $\mathrm{sn}_S + \mathrm{sp}_S \neq 1$.
\end{lemma}
\begin{proof}
The determinant of $P_2(\tilde{S} \mid Z,S^{P_0})$ is $0$.
The determinant of the $3$-minor $M_{4,4}$ is $(1 - \mathrm{sn}_S - \mathrm{sp}_S)^2$ which is nonzero as long as $\mathrm{sn}_S + \mathrm{sp}_S \neq 1$.
\end{proof}

\begin{lemma}[Rank $P_{N_z}(\tilde{S} \mid Z, S^{P_0}),\, N_z \geq 2$ ]\label{lemma:rank_p_s_tilde-gen-z}
The rank of $P(\tilde{S} \mid Z, S^{P_0})$ for $N_z \geq 2$ is $N_z + 1$
as long as $\mathrm{sn}_S + \mathrm{sp}_S \neq 1$.
\end{lemma}
\begin{proof}
We proceed by induction. For $N_z = 2$, \Cref{lemma:rank-p_s_tilde} shows that the rank is $3$.
Let $N_z = n$ for $n > 2$. Recall that $P_n(\tilde{S} \mid Z, S^{P_0})$ is the $2n \times 2^{n}$ matrix with column $j$
$$
\begin{bmatrix}
s_j \\
\mathbf{1}_n - s_j
\end{bmatrix}
$$
with the $i^\mathrm{th}$ element of $s_j$ denoted $s_{ij}$ and defined as:
$$
s_{ij} = \mathrm{sn}_S^{\varpi_n(j-1)_i} (1 - \mathrm{sp}_S)^{1-\varpi_n(j-1)_i}.
$$
The induction hypothesis is that the rank of $P_{n}(\tilde{S} \mid Z, S^{P_0})$ is $n+1$.
The columns of $P_{n+1}(\tilde{S} \mid Z, S^{P_0})$ are of the form
$$
\begin{bmatrix}
s_j \\
1 - \mathrm{sp}_S\\
\mathbf{1}_n - s_j \\
\mathrm{sp}_S
\end{bmatrix}
$$
for $j \in \{1, \dots, 2^{n}\}$,
and
$$
\begin{bmatrix}
s_{j-2^{n}} \\
\mathrm{sn}_S\\
\mathbf{1}_n - s_{j-2^{n}} \\
1 - \mathrm{sn}_S
\end{bmatrix}
$$
for $j \in \{2^{n} + 1, \dots, 2^{n+1}\}$.
After a row permutation we can express $P_{n+1}(\tilde{S} \mid Z, S^{P_0})$ as a block matrix:
\begin{align*}
    \begin{bmatrix}
        P_n(\tilde{S} \mid Z, S^{P_0}) & P_n(\tilde{S} \mid Z, S^{P_0}) \\
        (1-\mathrm{sp}_S)\mathbf{1}_{2^n}^T & \mathrm{sn}_S\mathbf{1}_{2^n}^T \\
        \mathrm{sp}_S\mathbf{1}_{2^n}^T & (1 - \mathrm{sn}_S)\mathbf{1}_{2^n}^T 
    \end{bmatrix}
\end{align*}
Recall that by construction the sum of the $i^\mathrm{th}$ row with the $(i + n)^\mathrm{th}$ row of $P_n(\tilde{S} \mid Z, S^{P_0})$ is $\mathbf{1}_{2^n}^T$ for $i \leq n$.
Then by \Cref{lemma:block-rank-ext}, $\mathrm{rank}\lp P_{n+1}(S \mid Z, S^{P_0})\rp$ is 
\begin{align}
\mathrm{rank}\lp P_{n+1}(S \mid Z, S^{P_0})\rp & = \mathrm{rank} \lp P_{n}(S \mid Z, S^{P_0})
\rp + 
\mathrm{rank} \lp
\begin{bmatrix}
        \mathrm{sn}_S\mathbf{1}_{2^n}^T \\
        (1 - \mathrm{sn}_S)\mathbf{1}_{2^n}^T 
\end{bmatrix}  - 
\begin{bmatrix}
        (1-\mathrm{sp}_S)\mathbf{1}_{2^n}^T \\
        \mathrm{sp}_S\mathbf{1}_{2^n}^T 
\end{bmatrix} \rp \\
& = n + 1 + 1
\end{align}
given that $\mathrm{sn}_S + \mathrm{sp}_S \neq 1$.
\end{proof}

\section{Main results}\label{sec:main-res}
We define identifiability as in \cite{rothenbergIdentificationParametricModels1971}; this definition is used below.
\begin{definition}[Parameter identifiability]\label{defn:id}
  Let $\theta \in \Theta$ be a parameter indexing a parametric density function $f(y \mid \theta)$. $\theta$ is identifiable if there does not exist a parameter value $\theta^\prime \in \Theta, \theta^\prime \neq \theta$ for which the density $f(y \mid \theta) = f(y \mid \theta^\prime)$ for all observations $y$.
\end{definition}

\begin{proof}{Proof of \Cref{thm:id-noisy}}\label{proof:id-noisy}\\
Define the three way array $L$ with dimensions $2N_z \times N_a \times N_r$ and $(i,j,r)^\mathrm{th}$ element $P(\tilde{S} = \ind{i \leq N_z}, A = k \mid Z = z_{i - N_z \ind{i > N_z}}, R = r)$.
Recall that the definition of matrix $P_{N_z}(\tilde{S} \mid Z, S^{P_0})$ requires that column $j$ be
$$
\begin{bmatrix}
s_j \\
\mathbf{1}_{N_z} - s_j
\end{bmatrix}
$$
where the $i^\mathrm{th}$ element of $s_j$ is denoted as $s_{ij}$ and is defined as:
$$
s_{ij} = \mathrm{sn}_S^{\varpi_{N_z}(j-1)_i} (1 - \mathrm{sp}_S)^{1-\varpi_{N_z}(j-1)_i}
$$
Let the matrices $P_{N_z}(S^{P_0} \mid R)^T, P_{N_z}(A \mid S^{P_0})$ be defined as in \Cref{subsec:defn}.
Then 
\begin{align*}
    P(\tilde{S} = \ind{i \leq N_z},& A = k \mid Z = z_{i - N_z \ind{i > N_z}}, R = r) = \\
    &\sum_{j=1}^{2^{N_z}} P_{N_z}(\tilde{S} \mid Z, S^{P_0})_{i,j} P_{N_z}(S^{P_0} \mid R)^T_{r,j}P_{N_z}(A \mid S^{P_0})_{k,j} .
\end{align*}
Given that $\mathrm{sn}_S + \mathrm{sp}_S \neq 1$, as shown in \Cref{lemma:p_s_tilde-gen-z}, $k_{P_{N_z}(\tilde{S} \mid Z, S^{P_0})} = 3$ and $\mathrm{rank}(P_{N_z}(\tilde{S} \mid Z, S^{P_0})) = N_z + 1$. 
Furthermore, by assumptions stated in \Cref{thm:id-noisy}, $\mathrm{rank}(P_{N_z}(S^{P_0} \mid R)^T) = 2^{N_z}$ and $P_{N_z}(S^{P_0} \mid R)^T \in \R^{N_r \times 2^{N_z}}$ so by \Cref{defn:krank}, $k_{P_{N_z}(S^{P_0} \mid R)^T} = 2^{N_z}$.
Given that $k_{P_{N_z}(A \mid S^{P_0})} \geq 2^{N_z} - 1$ as stated in \Cref{thm:id-noisy}, the conditions in \Cref{lemma:triple-prod-constrain} hold:
\begin{align}
\min(3,2^{N_z}) + 2^{N_z} - 1 & \geq 2^{N_z} + 2 \\    
\min(3,2^{N_z} - 1) + 2^{N_z} & \geq 2^{N_z} + 2 \\    
\end{align}
and 
\begin{align}
& \mathrm{rank}(P_{N_z}(S \mid Z, S^{P_0})) + \mathrm{rank}(P_{N_z}(S^{P_0} \mid R))  + \mathrm{rank}(P_{N_z}(A \mid S^{P_0})) \\
& \geq N_z + 1  + 2^{N_z} + 2^{N_z} - 1\\
& = N_z + 2^{N_z + 1}
\end{align}
by the fact that $\mathrm{rank}(P_{N_z}(A \mid S^{P_0})) \geq k_{(P_{N_z}(A \mid S^{P_0})}$.
Also
\begin{align}
N_z + 2^{N_z + 1} - 2(2^{N_z} - 1) &  = N_z - 1 \\
& \geq
\begin{cases}
\min(N_z - 2, \mathrm{rank}(P_{N_z}(A \mid S^{P_0})) - k_{(P_{N_z}(A \mid S^{P_0})}) \\
\min(N_z - 2, 0)
\end{cases}.
\end{align}
Given that $P_{N_z}(A \mid S^{P_0})$ has columns that sum to $1$, and $P_{N_z}(S^{P_0} \mid R)^T$ has rows that sum to $1$, we can apply \Cref{lemma:triple-prod-constrain} to the $3$-way array $L$.
Applying \Cref{lemma:triple-prod-constrain} yields that the triple-product decomposition $[P_{N_z}(\tilde{S} \mid Z, S^{P_0}),P_{N_z}(A \mid S^{P_0}),P_{N_z}(S^{P_0} \mid R)^T]$ is unique up to a common column permutation.
However, \Cref{thm:id-noisy} states the assumption that $\mathrm{sn}_S, \mathrm{sp}_S$ lie in a common half-interval.
By \Cref{lemma:domain}, the only permutation matrix consistent with the column domain of $P_{N_z}(\tilde{S} \mid Z, S^{P_0})$ is the identity matrix.
We conclude that the $3$-way decomposition of $L$, $[P_{N_z}(\tilde{S} \mid Z, S^{P_0}),P_{N_z}(A \mid S^{P_0}),P_{N_z}(S^{P_0} \mid R)^T]$, is unique.
It follows that two different decompositions $[P_{N_z}(\tilde{S} \mid Z, S^{P_0}),P_{N_z}(A \mid S^{P_0}),P_{N_z}(S^{P_0} \mid R)^T]$ and $[P_{N_z}(\tilde{S} \mid Z, S^{P_0})^\prime,P_{N_z}(A \mid S^{P_0})^\prime,(P_{N_z}(S^{P_0} \mid R)^T)^\prime]$ yield different $L$s.
By the fact that $L$ is a complete characterization of the data distribution $P(\tilde{S} = s, A = k \mid Z = z_{j}, R = r)$ and \Cref{defn:id} the parameter set $[P_{N_z}(\tilde{S} \mid Z, S^{P_0}),P_{N_z}(A \mid S^{P_0}),P_{N_z}(S^{P_0} \mid R)^T]$ is strictly identifiable.


Define the matrix $P(\tilde{Y} \mid Z, R, A = k)$ with dimensions $N_z \times N_r$ with elements $P(\tilde{Y} = y\mid R = r, Z = z, A = k)$
\[
P(\tilde{Y} \mid Z, R, A = k)_{i,r} = P(\tilde{Y} = 1 \mid Z = z_{i}, R = r, A = k).
\]
Let the matrix $P_{N_z}(\tilde{Y} \mid Z, S^{P_0}, A = k)$ be in $\R^{N_z \times 2^{N_z}}$ for all $k \in \{1,\dots,N_a\}$
with elements
\begin{align}
\begin{split}\label{eq:noisy-y-mat}
 P_{N_z}(\tilde{Y} \mid Z, S^{P_0}, A = k)_{i, j} &= \varpi_{N_z}(j-1)_{i} r_Y P(Y = 1 \mid Z = z_{i}, S^{P_0} = \varpi_{N_z}(j-1),  A = k) \\
& + (1 - \mathrm{sp}_Y)
\end{split}
\end{align}
where $r_Y = \mathrm{sp}_Y + \mathrm{sn}_Y - 1$.
Then
\begin{align*}
  P(\tilde{Y} = 1 \mid Z = z_{i}, A = k,  R = r) = \sum_{j=1}^{2^{N_z}} & P_{N_z}(\tilde{Y} \mid Z, S^{P_0}, A = k)_{i, j}P_{N_z}(S^{P_0} \mid R)_{j,r} \\
  & \times P_{N_z}(A = k \mid S^{P_0})_{j,j} / P(A = k \mid R = r).
\end{align*}
Recalling the definitions of diagonal matrices $P(A = k \mid R)$ and $P_{N_z}(A = k \mid S^{P_0})$, the expression can be rewritten as matrix multiplication:
\begin{align}
  P(\tilde{Y} \mid Z,  R, A = k) & = P_{N_z}(\tilde{Y} \mid Z, S^{P_0}, A = k)P_{N_z}(A = k \mid S^{P_0}) P_{N_z}(S^{P_0} \mid R) P(A = k \mid R)^{-1}.
\end{align}
Given our assumption that $P_{N_z}(S^{P_0} \mid R)$ is full row rank,$P_{N_z}(S^{P_0} \mid R)P_{N_z}(S^{P_0} \mid R)^+ = \mathbf{I}_{2^{N_z}}$.
We assume without loss of generality that $P(A = k \mid S^{P_0} = u) > 0\, \forall\, k,\, u$, so $P_{N_z}(A = k \mid S^{P_0})$ is invertible, and that $P(A = k \mid R = r) > 0 \forall \, k, \, r$ so $P(A = k \mid R)$ is invertible.
\begin{align}\label{eq:matrix-pyku}
  P(\tilde{Y} \mid Z,  R, A = k)P(A = k \mid R) P_{N_z}(S^{P_0} \mid R)^+P_{N_z}(A = k \mid S^{P_0})^{-1} & = P_{N_z}(\tilde{Y} \mid Z, A = k, S^{P_0})
\end{align}
If to the contrary that $P(A=k \mid S^{P_0} = u) = 0$ for some $k$ and $u$, we adopt the convention that $P(\tilde{Y} = y \mid Z = z, S^{P_0} = u, A = k)$ is undefined.
We can then reduce the set of principal strata included in the sum to include only those for which $P(A = k \mid S^{P_0} = u) > 0$.
It follows from the full-row-rank assumption on $P_{N_z}(S^{P_0} \mid R)$ that the matrix formed from any subset of rows of this matrix is still full row-rank.
Given the modified matrices, \Cref{eq:matrix-pyku} will hold with the reduced set of principal strata.
We can use a similar technique when $P(A = k \mid R = r) = 0$ for some combination of $k$ and $r$.
This condition is empirically testable.

It then follows the definition of $P_{N_z}(\tilde{Y} \mid Z, S^{P_0}, A = k)$ in \Cref{eq:noisy-y-mat} that $\mathrm{sp}_Y$ is identifiable, as are the parameters $r_Y P(Y = 1 \mid Z = z_j, S^{P_0} = \varpi_{N_z}(j-1), A = k)$ for all $z_j, j \in \{1, \dots, 2^{N_z}\}$ and $k$.

Let any allowable post-infection outcome vaccine efficacy estimand, necessarily where $u_{j} \,u_{l} = 1$, be defined as
$$
\mathrm{VE}_{I,jl}^{u}(k) = 1 - \frac{\Exp{Y(z_j) \mid S^{P_0} = u, A = k}}{\Exp{Y(z_l) \mid S^{P_0} = u, A = k}}.
$$

By \Crefrange{cond:SUTVA}{cond:pc-obs-unconfound-multi-z} $P(Y = 1 \mid Z = z, S^{P_0}=u, A = k) = P(Y(z) = 1 \mid S^{P_0} = u, A = k)$ for all $z \in \{z_1, \dots, z_{N_z}\}$ and $\Exp{Y(z) \mid S^{P_0} = u, A = k} = P(Y(z) = 1 \mid S^{P_0} = u, A = k)$.
Note that $\mathrm{sp}_Y = 1 - P_{N_z}(\tilde{Y} \mid Z, S^{P_0}, A = k)_{1, 1}$ by our definition of $P_{N_z}(\tilde{Y} \mid Z, S^{P_0}, A = k)$ in \Cref{eq:noisy-y-mat}.
This is because the the first column of $P_{N_z}(\tilde{Y} \mid Z, S^{P_0}, A = k)$ corresponds to the principal stratum that is always uninfected, or $S^{P_0} = \varpi(0)$, which results in $P(\tilde{Y}=1 \mid Z = z, S^{P_0} = \varpi(0), A = k) = 1 - \mathrm{sp}_Y$ for all $z, k$.
The first row of the matrix $P_{N_z}(\tilde{Y} \mid Z, S^{P_0}, A = k)$ corresponds to $\tilde{Y} = 1$.
Then 
\[
P(Y = 1 \mid Z = z, S^{P_0} = u, A = k) = \frac{P_{N_z}(\tilde{Y} \mid Z, S^{P_0}, A = k)_{z, j} - P_{N_z}(\tilde{Y} \mid Z, S^{P_0}, A = k)_{1, 1}}{r_Y}
\]
where $j = \varpi_{N_z}^{-1}(u)+1$, so $\mathrm{VE}^{u}_{I,jl}(k)$ is identifiable.
\end{proof}

\begin{proof}{Proof of \Cref{cor:id-noisy-A}} \label{proof:id-noisy-A}\\
By the conditions set forth in \Cref{cor:id-noisy-A} we have that 
\begin{align*}
    P(\tilde{S} = \ind{i \leq N_z},& \tilde{A} = k \mid Z = z_{i - N_z \ind{i > N_z}}, R = r) = \\
    &\sum_{j=1}^{2^{N_z}} P_{N_z}(\tilde{S} \mid Z, S^{P_0})_{i,j} P_{N_z}(S^{P_0} \mid R)^T_{r,j}P_{N_z}(\tilde{A} \mid S^{P_0})_{k,j} .
\end{align*}
This decomposition holds because of our nondifferential misclassification assumption, namely $\tilde{A} \indy S^{P_0}, \tilde{S}, R, Z \mid A$, which allows for the following complete characterization of $\tilde{A} \mid S^{P_0}$:
$$P(\tilde{A} = k \mid S^{P_0} = u) = \sum_{\ell = 1}^{N_z} P(\tilde{A} = k \mid A = \ell)P(A = \ell \mid S^{P_0} = u).$$

Recall that $\mathrm{sn}_S,\mathrm{sp}_S$ lie in the same half interval of $[0,1]$, so by the same logic as \Cref{proof:id-noisy}, the distributions  
$P(\tilde{S} = 1 \mid Z = z, S^{P_0} = u),P(\tilde{A} = k \mid S^{P_0} = u),P(S^{P_0} = u \mid R = r)$ are identifiable. 
Define the matrix $P(\tilde{Y} \mid Z, R, \tilde{A} = k)$ with dimensions $N_z \times N_r$ with elements $P(\tilde{Y} = y\mid R = r, Z = z, \tilde{A} = k)$
\[
P(\tilde{Y} \mid Z, R, \tilde{A} = k)_{i,r} = P(\tilde{Y} = 1 \mid Z = z_{i}, R = r, \tilde{A} = k).
\]
Let the matrix $P_{N_z}(\tilde{Y} \mid Z, S^{P_0}, \tilde{A} = k)$ be defined in the same way as \Cref{eq:noisy-y-mat}.
Then
\begin{align}
  P(\tilde{Y} = 1 \mid Z = z_{i}, \tilde{A} = k,  R = r) = \sum_{j=1}^{2^{N_z}} & P_{N_z}(\tilde{Y} \mid Z, S^{P_0}, \tilde{A} = k)_{i, j} \\
  & \times P_{N_z}(S^{P_0} \mid R)_{j,r} P_{N_z}(\tilde{A} = k \mid S^{P_0})_{j,j} / P(\tilde{A} = k \mid R = r)
\end{align}
which can be represented as matrix multiplication, recalling the definitions of matrices $P_{N_z}(\tilde{A} = k \mid S^{P_0})$, and $ P(\tilde{A} = k \mid R)$:
\begin{align}
  P(\tilde{Y} \mid Z,  R, \tilde{A} = k) & = P_{N_z}(\tilde{Y} \mid Z, S^{P_0}, \tilde{A} = k) P_{N_z}(\tilde{A} = k \mid S^{P_0})P_{N_z}(S^{P_0} \mid R) P(\tilde{A} = k \mid R)^{-1}
\end{align}
Given our assumption that $P_{N_z}(S^{P_0} \mid R)$ is full row rank,$P_{N_z}(S^{P_0} \mid R)P_{N_z}(S^{P_0} \mid R)^+ = \mathbf{I}_{2^{N_z}}$.
Without loss of generality we assume that $P(\tilde{A} = k \mid R = r) > 0 \, \forall \, k, r$, which is estimable from observed data and $P(\tilde{A} = k \mid S^{P_0}=u) > 0 \, \forall \, k, u$, which is identified via the .
Then:
\begin{align}
  P(\tilde{Y} \mid Z,  R, \tilde{A} = k)P(\tilde{A} = k \mid R) P_{N_z}(S^{P_0} \mid R)^+P_{N_z}(\tilde{A} = k \mid S^{P_0})^{-1} & = P_{N_z}(\tilde{Y} \mid Z, \tilde{A} = k, S^{P_0})
\end{align}
It then follows the definition of $P_{N_z}(\tilde{Y} \mid Z, S^{P_0}, \tilde{A} = k)$ that $\mathrm{sp}_Y$ is identifiable, as are the parameters $r_Y P(Y = 1 \mid Z = z_j, S^{P_0} = \varpi_{N_z}(j-1), \tilde{A} = k)$ for all $j \in \{1, \dots, 2^{N_z}\}$ and $k \in \{1, \dots, N_a\}$.

Let any allowable post-infection outcome vaccine efficacy estimand, necessarily where $u_{j} \,u_{l} = 1$, be defined as
$$
\mathrm{VE}_{I,jl}^{u} = 1 - \frac{\Exp{Y(z_j) \mid S^{P_0} = u}}{\Exp{Y(z_l) \mid S^{P_0} = u}}.
$$

By \Crefrange{cond:SUTVA}{cond:pc-obs-unconfound-multi-z} $P(Y = 1 \mid Z = z, S^{P_0}=u, \tilde{A} = k) = P(Y(z) = 1 \mid S^{P_0} = u, \tilde{A} = k)$ and $\Exp{Y(z) \mid S^{P_0} = u, \tilde{A} = k} = P(Y(z) = 1 \mid S^{P_0} = u, \tilde{A} = k)$ for all $z \in \{z_1, \dots, z_{N_z}\}$.
Note that $\mathrm{sp}_Y = 1 - P_{N_z}(\tilde{Y} \mid Z, S^{P_0}, \tilde{A} = k)_{1, 1}$ by our definition of $P_{N_z}(\tilde{Y} \mid Z, S^{P_0}, \tilde{A} = k)$ in \Cref{eq:noisy-y-mat}.
Then 
\[
P(Y = 1 \mid Z = z, S^{P_0} = u, \tilde{A} = k) = \frac{P_{N_z}(\tilde{Y} \mid Z, S^{P_0}, \tilde{A} = k)_{z, j} - P_{N_z}(\tilde{Y} \mid Z, S^{P_0}, \tilde{A} = k)_{1, 1}}{r_Y}
\]
where $j = \varpi_{N_z}^{-1}(u)+1$.
Then 
\[
\mathrm{VE}_{I,jl}^{u} = 1 - \frac{\sum_k \lp P_{N_z}(\tilde{Y} \mid Z, S^{P_0}, \tilde{A} = k)_{z, j} - P_{N_z}(\tilde{Y} \mid Z, S^{P_0}, \tilde{A} = k)_{1, 1}\rp P(\tilde{A}=k \mid S^{P_0} = u)}{\sum_k \lp P_{N_z}(\tilde{Y} \mid Z, S^{P_0}, \tilde{A} = k)_{z, l} - P_{N_z}(\tilde{Y} \mid Z, S^{P_0}, \tilde{A} = k)_{1, 1}\rp P(\tilde{A}=k \mid S^{P_0} = u)}.
\] 
\end{proof}

\begin{lemma}[Identifiability of counterfactual post-infection outcome expectations]
  Suppose $P(S_i^{P_0} = u \mid R = r)$ is known 
\end{lemma}

\section{Kruskal rank properties}\label{sec:krank}
In the section that follows, we use properties and several theorems and lemmas that are proven in \cite{kruskalThreewayArraysRank1977}.
Where appropriate we will indicate on which pages the proofs of the theorems and lemmas can be found.
\begin{lemma}[Rank lemma]\label{lemma:rank}
 Let $$H_{AB}(n) = \min_{\mathrm{card}(A^\prime)=n} \left\{\mathrm{rank}(A^\prime) + \mathrm{rank}(B^\prime)\right\} - n$$
 for an integer $n$ where $A^\prime$ is an $n$-column subset of the matrix $A$ and $B^\prime$ is the same column-index subset of a matrix $B$.
 For any diagonal matrix $D \in \R^{n \times n}$ with rank $\delta$, 
 $$
\mathrm{rank}(A D B^T) \geq H_{AB}(\delta).
 $$
 See proof on p. 121 in \cite{kruskalThreewayArraysRank1977}.
\end{lemma}
\subsection{Kruskal's triple-product decomposition uniqueness theorem}

\begin{theorem}[Kruskal triple product decomposition uniqueness]\label{thm:krank-unique}
Let matrices $A, B, C$ be defined as in \Cref{defn:triple-prod}, with respective ranks $r_A, r_B, r_C$, and let array $L$ also be defined as in \Cref{defn:triple-prod}. 
Suppose that $k_A \leq r_A, k_B \leq r_B$, and $k_C \leq r_C$. Then if
 $$
 r_A + r_B + r_C - (2 M + 2) \geq \begin{cases}
 \min(r_A - k_A, r_B - k_B)  \\
 \min(r_A - k_A, r_C - k_C)
 \end{cases},
 $$ 
 $\min(k_A, k_B) + r_C \geq M + 2$, and $\min(k_A, k_C) + r_B \geq M + 2$ the decomposition $L = [A, B, C]$ is unique up to column permutation matrix $P$ and column scaling $\Lambda, G, N$ such that $\Lambda G N$ is the identity matrix. In other words, $L$ can be represented as the triple product of any three matrices $[\tilde{A}, \tilde{B}, \tilde{C}]$ such that $[\tilde{A} = A P \Lambda, \tilde{B} = B P G, \tilde{C} = C P N]$. See proof in \cite{kruskalThreewayArraysRank1977} on page 126.
\end{theorem}

\subsection{Corollary to \Cref{thm:krank-unique}}\label{subsec:trip-lemma-pf}
\begin{corollary}[Uniqueness with column and row sum conditions]\label{lemma:triple-prod-constrain}
  Suppose $B$ has rows that sum to $1$ and $C$ has columns that sum to $1$, or $B\mathbf{1}_{R \times 1} = \mathbf{1}_{J \times 1}$, and $\mathbf{1}_{1 \times K} C = \mathbf{1}_{1 \times R}$.
  If the rank conditions in \Cref{thm:krank-unique} on $A, B, C$ also hold, and $C$ is full column rank then $[A, B, C]$ is the unique triple product decomposition of array $L$ up to a common column permutation.
\end{corollary}

\begin{proof}
  Suppose that $L = [A, B, C]$ and that $[\bar{A}, \bar{B}, \bar{C}]$ is another decomposition of $L$, where $\bar{B}, \bar{C}$ satisfy the respective row- and column-sum constraints. 
Let $r_{\bar{B}}, r_{\bar{C}}$ be the ranks of $\bar{B}$ and $\bar{C}$ respectively.
\Cref{defn:triple-prod} implies that $A \mathrm{diag}(xC) B^T =  \bar{A}\mathrm{diag}(x\bar{C}) \bar{B}^T$ for all $x \in \R^{1 \times I}$. 
 If for any $y \in \R^{1 \times K}$ such that $y \bar{C} = 0 \implies y C = 0$ then $\mathrm{col}(C) \subset \mathrm{col}(\bar{C})$, $\mathrm{null}(C) \supset \mathrm{null}(\bar{C})$, and $r_C \leq r_{\bar{C}}$. If $y \bar{C} = 0$ then 
 $$
 \bar{A} \mathrm{diag}(y\bar{C}) \bar{B}^T = 0 \implies A \mathrm{diag}(yC) B^T = 0
 $$
 Recall the definition of $H_{AB}(n)$ from \Cref{lemma:rank}.
\cite{kruskalThreewayArraysRank1977} shows that the condition on the ranks and Kruskal ranks above imply the following inequalities (proof omitted):
\begin{align}
    k_A & \geq \max(R - r_B + 2, R - r_C + 2), \label{eq-line-1:in-eq-ka}\\
    k_B & \geq R - r_C + 2, \label{eq-line-2:in-eq-kb}\\
    k_C & \geq R - r_B + 2, \label{eq-line-2:in-eq-kc}\\
    H_{AB}(n) & \geq R - r_C + 2\, \mathrm{ if }\, n \geq R - r_C + 2 \label{eq-line:in-eq-hab}\\
    H_{AC}(n) & \geq R - r_B + 2\, \mathrm{ if }\, n \geq R - r_B + 2 \label{eq-line:in-eq-hac}\\
    H_{BC}(n) & \geq 1\, \mathrm{ if }\, n \geq 1 \label{eq-line:in-eq-hbc}
\end{align}
The inequality \cref{eq-line:in-eq-hab} implies that when $H_{AB}(n) < R - r_C + 2$ then $n < R - r_C + 2$.
When $n < R - r_C + 2$, the inequalities \crefrange{eq-line-1:in-eq-ka}{eq-line-2:in-eq-kc} and the definition of $H_{AB}(n)$ imply that $H_{AB}(n) = n$.
Then 
 \begin{align*}
     0 & = \mathrm{rank}(A \mathrm{diag}(yC) B^T) \\
     & \geq H_{AB}(\mathrm{rank}(\mathrm{diag}(yC))) \geq 0,
 \end{align*} 
 where the second to last inequality comes from \Cref{lemma:rank} and 
 the last inequality comes from the definition of $H_{AB}(n)$. This implies $yC = 0$.
 Let the function $w(y)$ for a generic vector $y$ return the number of nonzero entries in the vector $y$. 
 Let $v$ be any vector such that $w(v\bar{C}) \leq R - \bar{K}_0 + 1$. Then we'll show that $w(vC) \leq w(v\bar{C})$.
 \begin{align}
     R - r_C + 1 & \geq R - \bar{K}_0 + 1 \geq w(v\bar{C}) = \mathrm{rank}(\mathrm{diag}(v\bar{C}) \\
     & \geq \mathrm{rank}(A \mathrm{diag}(y\bar{C}) \bar{B}^T) = \mathrm{rank}(A \mathrm{diag}(yC) B^T) \\
     & \geq H_{AB}(\mathrm{rank}(\mathrm{diag}(vC)) = H_{AB}(w(vC)).
 \end{align}
The final line implies that $H_{AB}(w(vC)) = w(vC)$, which shows that $w(v\bar{C}) \geq w(vC)$ when $R - \bar{K}_0 + 1 \geq w(v\bar{C})$.

Given this condition, Kruskal's permutation lemma (proved on page 134 of \cite{kruskalThreewayArraysRank1977})
shows that for any matrices $C$ and $\bar{C}$ that satisfy the inequality, $\bar{C} = C P_C N$ where $P_C$ is a permutation matrix and $N$ is a diagonal nonsingular scaling matrix. 
If we have the stronger condition that every two columns of $C$ are linearly independent then $P_C$ and $N$ are unique.
Our matrices satisfy these conditions, so we have that $\bar{C} = C P_C N$, and a similar argument can be used to show $\bar{B} = B P_B M$

Given that we also have the condition that $\mathbf{1}_{1 \times K} C = \mathbf{1}_{1 \times R}$ and $\mathbf{1}_{1 \times K} \bar{C} = \mathbf{1}_{1 \times R}$, then this implies that $\bar{C} = C P_C$ because $\mathbf{1}_{1 \times K} \bar{C} = \mathbf{1}_{1 \times K} C P_C N = \mathbf{1}_{1 \times R} N$ which only equals $\mathbf{1}_{1 \times R}$ if $N = \mathbf{I}_{R \times R}$.

Furthermore, if $r_B = R$, the equation $B \nu = \mathbf{1}_{J \times 1}$ has a unique solution in $\nu \in \R^{R \times 1}$, namely $\nu = \mathbf{1}_{R \times 1}$.
This implies that $M$ is the identity matrix, as the condition $\bar{B} \mathbf{1}_{R\times 1} = \mathbf{1}_{J \times 1}$ results in:
\begin{align}
 \mathbf{1}_{J \times 1} & = \bar{B} \mathbf{1}_{R\times 1}\\
 & = B P_B M \mathbf{1}_{R\times 1} \\
 & \implies P_B M \mathbf{1}_{R\times 1} = \mathbf{1}_{R \times 1}.
\end{align}
Given that $M$ is a nonsingular diagonal matrix and $P_B$ is a permutation matrix, $M$ must be the identity to solve the equation $P_B M \mathbf{1}_{R \times 1} = \mathbf{1}_{R \times 1}$.

We now have $\bar{C} = C P_C$ and $\bar{B} = B P_B$. 
We can apply Kruskal's permutation matrix proof from pages 129-130 in \cite{kruskalThreewayArraysRank1977} to show that $P_C = P_B = P$.
The following two identities hold for any diagonal scaling matrices $M,N$, any permutation matrix $P$, and any vector $v$:
\begin{align}
  M \mathrm{diag}(v) N & = \mathrm{diag}(v M N) \label{eq:scaling} \\
  P \mathrm{diag}(v) P^T & = \mathrm{diag}(v P^T)\label{eq:permuting}.
\end{align}
Given \Crefrange{eq:scaling}{eq:permuting}
and the condition that $L = [A, B, C] = [\bar{A}, \bar{B}, \bar{C}]$, then, for all vectors $v \in \R^{1 \times J}$,
\begin{align}
B \mathrm{diag}(vA) C^T & = \bar{B} \mathrm{diag}(v\bar{A}) \bar{C}^T \\
& = B P \mathrm{diag}(v\bar{A}) P^T C^T \\
& = B \mathrm{diag}(v\bar{A}P^T) C^T.
\end{align}
The equality $B \mathrm{diag}(vA) C^T =B \mathrm{diag}(v\bar{A}P^T) C^T$ implies
\begin{align}
B \mathrm{diag}(v(A - \bar{A}P^T)) C^T = 0
\end{align}
for all $v$.
Furthermore, 
\begin{align}
0 & = \mathrm{rank}(B \mathrm{diag}(v(A - \bar{A}P^T)) C^T) \\
& \geq H_{BC}(\mathrm{rank}(\mathrm{diag}(v(A - \bar{A}P^T))) \geq 0.
\end{align}
The last line follows from \Cref{lemma:rank}.
Then using the implication from \cref{eq-line:in-eq-hbc} that if $H_{BC}(n) < 1 \implies n = 0$, $\mathrm{rank}(\mathrm{diag}(v(A - \bar{A}P^T))=0$ or 
$v(A - \bar{A}P^T) = 0$ for all $v$. 
This further implies that
$$
A = \bar{A} P^T
$$
or 
$$
\bar{A} = A P.
$$
\end{proof}

\section{Supporting lemmas and definitions from other work}

\begin{lemma}[Block rank lemmas \cite{tian2004rank}]\label{lemma:block-rank-alt}
Let $A \in \R^{m\times n}, B \in \R^{m \times k}, C \in \R^{l \times n}$.
$$
\mathrm{rank} \lp \begin{bmatrix}
A & B \\
C & 0
\end{bmatrix} \rp = \mathrm{rank}(B) + \mathrm{rank}(C) + \mathrm{rank}((I - BB^+)A(I - C^+C))
$$
If $\mathrm{range}(B) \subseteq \mathrm{range}(A)$ and $\mathrm{range}(C^T) \subseteq \mathrm{range}(A^T)$
$$
\mathrm{rank} \lp \begin{bmatrix}
A & B \\
C & D
\end{bmatrix} \rp = \mathrm{rank}(A) + \mathrm{rank}(D - C A^+ B)
$$

\end{lemma}
\begin{lemma}[Block rank lemma extension]\label{lemma:block-rank-ext}
Let $A \in \R^{m\times n}, B \in \R^{m \times k}, C \in \R^{l \times n}$.
If $\mathrm{range}(C^T) \subseteq \mathrm{range}(A^T)$
$$
\mathrm{rank} \lp \begin{bmatrix}
A & A \\
C & D
\end{bmatrix} \rp = \mathrm{rank}(A) + \mathrm{rank}(D - C)
$$
\end{lemma}
\begin{proof}
Given that $\mathrm{range}(A) \subseteq \mathrm{range}(A)$, we can apply the second block rank lemma from \Cref{lemma:block-rank-alt} with $B = A$.
$$
\mathrm{rank} \lp \begin{bmatrix}
A & A \\
C & D
\end{bmatrix} \rp = \mathrm{rank}(A) + \mathrm{rank}(D - C A^+ A).
$$
By supposition, $\mathrm{range}(C^T) \subseteq \mathrm{range}(A^T)$ and $A^+ A$ is the projection matrix onto the column space of $A^T$.
Then $C A^+ A = C$, and the statement follows.
\end{proof}

\section{Details behind numerical examples}\label{sec:numerical-deets}

We have three simulation scenarios where we vary the sample size to determine the power: a two-arm trial to determine vaccine efficacy against severe symptoms, a three-arm trial to determine relative vaccine efficacy against severe symptoms, and a two-arm trial to determine vaccine efficacy against transmission. 
All trials are designed such that the assumptions of \Cref{thm:id-noisy} are satisfied, so the three-arm trial includes $16$ study sites, and a categorical covariate with $7$ levels, and both two-arm trials include $8$ study sites, and a categorical covariate with $3$ levels. Within each scenario we suppose that $A$, the categorical covariate, to be measured perfectly.
In addition, we assume a $3$-level, pretreatment categorical covariate has been measured for each participant.
We simulate from the parametric model defined in \Cref{subsec:sens}, which requires that we specify $\mu_u^r$, or the log-odds of belonging to stratum $u$ relative to base stratum $u_0$ for each study site $r$. 
Let the ordered collection of log-odds of being in stratum $u$ relative to stratum $u_{2^{N_z}}$ for the reference covariate level $x=1$ be $\mu^r = \lp\mu_{u_1}^r, \mu_{u_2}^r,\dots, \mu_{u_{2^{N_z}-1}}^r, 0\rp$.

Let $\mathrm{softmax}$ be the function from $v \in \R^L$ to the $L+1$-dimensional probability simplex, defined elementwise for the $i^\mathrm{th}$ element as:
$$
\mathrm{softmax}(v)_i = \frac{e^{v_i}}{\sum_{l=1}^{L} e^{v_l^r}}
$$
and let $\mathrm{softmax}^{-1}$ be the inverse function from $\theta \in $ the $L+1$-dimensional simplex to $\R^L$, where the $i^\mathrm{th}$ element, $i < L+1$ is defined as
$$
\mathrm{softmax}(\theta)^{-1}_i = \log(\theta_i) - \log(\theta_{L+1})
$$

Let $\theta_u^{r,x} = P(S^{P_0}=u \mid R = r, X = x)$, and let $\theta^{r,x}$ be the ordered vector $(\theta_{u_1}^{r,x}, \theta_{u_1}^{r,x}, \dots, \theta_{u_{2^{N_z}}}^{r,x})$ .   
For the $2$-arm trials, the population principal strata proportions are as follows:

$$
\theta^{r,1} \overset{\mathrm{iid}}{\sim} \mathrm{Dirichlet}((78, 10, 2, 10)) \forall r
$$
while for the $3$-arm trials, the proportions are
{
\begin{align*}
\theta^{r,1} \overset{\mathrm{iid}}{\sim} \mathrm{Dirichlet}((70, 13,1, 1, 1, 1, 1, 12)) \forall r
\end{align*}
}
Assuming equal proportions of participants in each treatment group, these parameter settings equate to a cumulative true incidence of roughly $0.16$ and $0.19$, respectively.
Recall from \Cref{subsec:sens} that $\eta^x \in \R^{2^{N_z}}$, so $\eta^x_{u}$ is the change in log-odds of belonging to principal stratum $u$ vs. $u_0$ relative to $x=1$. We set $\eta^x_{2^{N_z}} = 0$ for identifiability. 
Then let $\mu_u^r = \mathrm{softmax}^{-1}(\theta^{r,1})$, 
\[
  \theta^{r,x} = \mathrm{softmax}\lp\mu_u^r + \eta^x\rp, 
\]
where for all $x > 1$
$$
\eta^{x}_{i} \overset{\mathrm{iid}}{\sim} \mathrm{Normal}(0, 1),\, i < 2^{N_z}, \, \eta^{x}_{2^{N_z}} = 0.
$$
Let the $N_a$-vector $a^{u,x}$ be defined elementwise as $a_k^{u,x}$ where $a_k^{u,x} = P(A = k \mid S^{P_0} = u, X = x)$.
{
\begin{align*}
a^{u,1} \overset{\mathrm{iid}}{\sim} \mathrm{Dirichlet}(4 \mathbf{1}_{N_a}) \forall u,
\end{align*}
}
and $\nu^{u} = \mathrm{softmax}^{-1}(a^{u,1})$. Then recall that $\gamma^x \in \R^{N_a}$ such that $\gamma^{x}_k$ is the change in log-odds of $A = k$ relative to $A = k_0$, and that $\gamma^x_{N_a} = 0$ for identifiability. 
Then 
{
\begin{align*}
a^{u,x} = \mathrm{softmax}\lp \nu^u + \gamma^x \rp,
\end{align*}
}
and for all $x > 1$
$$
\gamma^{x}_{i} \overset{\mathrm{iid}}{\sim} \mathrm{Normal}(0, 1),\, i < N_a, \, \gamma^{x}_{N_a} = 0.
$$
For data generating process that does not adhere to the \nameref{cond:a-indy} assumption, $a^{u,x,r}$ is generated as follows
{
  \begin{align}
    \begin{split}\label{eq:dgp-non-indy}
    a^{u,x,r} & = \mathrm{softmax}\lp \nu^u + \gamma^x + \varepsilon^{r} \rp \\
      \varepsilon_k^{r} & \sim \text{Normal}(0, \sigma^2_{\varepsilon}) \,\forall r > 1, k < N_a  \\
      \varepsilon_{N_a}^{r} & = 0 \forall r \\
      \varepsilon_k^{1} & = 0 \forall k
  \end{split}
  \end{align}
}
We set $\sigma_{\varepsilon} = \{0.5, 1, 2\}$ for our simulation study scenarios.

Finally, recall that 
\begin{equation*}
    \log\frac{P(Y(z_j) = 1 \mid S^{P_0} = u, A = k, X = x)}{P(Y(z_j) = 0 \mid S^{P_0} = u, A = k, X = x)} = \alpha_j^{u} + \delta^{u}_{j,k} + \omega^x_j,
\end{equation*}

where $\omega_j^x$ is the change in log-odds of $Y(z_j) = 1$, all else being equal, compared to $x=1$.
In all of our simulations, $\omega_j^x = (x - 1) \log(1.1)$ for all $j$.
For the $2$-arm trial example, we let $\alpha_1^{(1,1)} = \log(0.3 / 0.7), \alpha_2^{(1,1)} = \log(0.3 / 0.7) + \log(0.4)$, and 
$\delta^{(1,1)}_{1,k} = (k - 1) \log(0.925), \delta^{(1,1)}_{2,k} = (k - 1) \log(0.825)$.
Further, we let $\alpha_1^{(1,0)} = \log(0.15/0.85)$, $\alpha_2^{(0,1)} = \log(0.2/0.8)$, and $\delta^{(1,0)}_{1,k} = (k - 1) \log(0.925)$, and $\delta^{(0,1)}_{2,k} = 0$

For the $3$-arm trial example, we let $\alpha_1^{(1,1,1)} = \log(0.3 / 0.7), \alpha_2^{(1,1,1)} = \log(0.3 / 0.7), \alpha_3^{(1,1,1)} = \log(0.3 / 0.7) + \log(0.4)$, and 
$\delta^{(1,1,1)}_{j,k} = (k - 1) \log(0.925)$ for $j=1,2,3$.
Further, we let $\alpha_1^{(1,0,1)} = \log(0.2/0.8)$, $\alpha_3^{(1,0,1)} = \log(0.1/0.9)$,$\alpha_1^{(1,1,0)} = \log(0.3/0.7)$, $\alpha_2^{(1,1,0)} = \log(0.15/0.85)$,$\alpha_2^{(0,1,1)} = \log(0.25/0.75)$, $\alpha_3^{(0,1,1)} = \log(0.08/0.92)$, $\alpha_3^{(0,0,1)} = \log(0.25/0.75)$, $\alpha_2^{(0,1,0)} = \log(0.25/0.75)$, $\alpha_1^{(1,0,0)} = \log(0.1/0.9)$ and $\delta^{u}_{j,k} = 0$ for all $k$, $u\in \{(1,0,0),(0,1,0),(1,1,0),(0,0,1),(1,0,1),(0,1,1)\}$, and all allowable $j$

In the $2$- and $3$-arm trial examples that pertain to inferring vaccine efficacy against severe symptoms, we set $\mathrm{sn}_S = 0.8, \mathrm{sp}_S = 0.99$ which reflects the sensitivity and specificity of a typical PCR collected via nasopharyngeal swab \citep{kisslerViralDynamicsAcute2021}, and $\mathrm{sn}_Y = 0.99, \mathrm{sp}_S = 0.9$ to reflect the fact that most severe illness caused by the pathogen of interest will be reported, but that there are many severe illness episodes that are reported that may be caused by other pathogens. 
These lead to a true rate of severe illness of $0.04$ but a rate of reported severe illness of $0.14$. For comparison \cite{montoComparativeEfficacyInactivated2009} symptom reporting data shows that $10\%$ of participants reported at least one severe symptom, but the cumulative incidence was $0.07$.


For each hypothetical participant in a study site $R = r$ in our study we draw data in the following manner
\begin{align}
\begin{split}
    Z_i & \overset{\mathrm{iid}}{\sim} \mathrm{Categorical}(\frac{1}{N_z} \mathbf{1}_{N_z}) \\
    X_i & \overset{\mathrm{iid}}{\sim} \mathrm{Categorical}(\frac{1}{3}\mathbf{1}_{3}) \\
   S^{P_0}_i \mid R = r, X = x & \overset{\mathrm{iid}}{\sim} \mathrm{Categorical}(\theta^{r,x}) \\
   A_i \mid S^{P_0} = u, X = x & \overset{\mathrm{iid}}{\sim}  \mathrm{Categorical}(a^{u,x}) \\
   Y_i \mid S^{P_0} = u, A = k, X = x, Z = j & \overset{\mathrm{iid}}{\sim}\mathrm{Bernoulli}(\mathrm{inv\_logit}(\alpha_j^u + \delta_{j,k}^u + \omega_x))\\
   \tilde{Y}_i \mid Y = y & \overset{\mathrm{iid}}{\sim}\mathrm{Bernoulli}(y \mathrm{sn}_Y + (1 - y) (1 - \mathrm{sp}_Y)) \\
   \tilde{S}_i \mid S^{P_0} = u, Z = j & \overset{\mathrm{iid}}{\sim}\mathrm{Bernoulli}(u_j \mathrm{sn}_S + (1 - u_j) (1 - \mathrm{sp}_S)) \\
   \tilde{A}_i \mid A = a & \overset{\mathrm{iid}}{\sim}\mathrm{Categorical}(p_{N_a}^a)
   \end{split}
\end{align}
and we do this for all sites $R \in \{1, \dots, N_r\}$.

We fit the model defined in \Cref{eq:misclass-obs-model-multi-site-multi-z}.
Recall
\begin{align}
  A_i  \mid S_i^{P_0} = u, X_i = x, R_i = r & \sim \text{Categorical}(\boldsymbol{\pi}_{u,x,r})\\
  S_i^{P_0}  \mid R_i = r, X_i = x & \sim \text{Categorical}(\boldsymbol{\rho}_{r,x}) \\
  Y_i(z_j) \mid S^{P_0}_i = u, A_i = k, X_i = x & \sim \text{Bernoulli}(\beta^{u,x}_{j,k}) 
  \end{align}
  where we define $\boldsymbol{\pi}_{u,x,r}$ as in \Cref{eq:a-interact-mod}:
  \begin{align}
    \begin{split}
    \boldsymbol{\pi}_{u,x,r} & = \mathrm{softmax}\lp \boldsymbol{\nu}^u + \boldsymbol{\gamma}^{x} + \boldsymbol{\gamma}^{u,x} + \boldsymbol{\epsilon}^{r} \rp \\
    \epsilon_k^{r} & \sim \text{Normal}(0, \sigma^2_{\epsilon}) \,\forall r > 1, k < N_a  \\
      \gamma_k^{u,x} & \sim \text{Normal}(0, \sigma^2_{\gamma}) \,\forall u, x > 1, k < N_a  \\
      \sigma_{\epsilon} & \sim \text{StudentT}^+(3, 0, 0.1), \, \sigma_{\gamma} \sim \text{StudentT}^+(3, 0, 0.1) \\
      \epsilon_{N_a}^{r}  = 0 \forall r,\, \epsilon_k^{1} = 0 \forall k,\,\gamma_{N_a}^{x} &= 0 \forall x,\, \gamma_k^{1} = 0 \forall k,\, \gamma_{N_a}^{u,x} = 0 \forall u,x,\, \gamma_k^{u,1} = 0, \gamma_k^{1,x} = 0 \forall k \\
  \end{split}
  \end{align}
  We define a hierarchical model for $\beta_{j,k}^{u,x}$ as:
  \begin{align}
    \begin{split}
      \text{logit}  \beta_{j,k}^{u,x} & = \alpha_{j}^u + \epsilon_{j,k}^{u,x} \\
      \epsilon_{j,k}^{u,x} & \sim \text{Normal}(0, \tau^2_{\epsilon}) \\
      \tau_{\epsilon} &\sim \text{Normal}^+(0, 5^2) \\
      \alpha_j^u & \sim \text{Normal}(0, 1.7^2)
      \end{split}
    \end{align}
  We use the following priors:
\begin{align}
    \begin{split}
    \mathrm{sn}_S & \sim \mathrm{Beta}(0.5, 1, 4, 2) \\
    \mathrm{sp}_S & \sim \mathrm{Beta}(0.5, 1, 10, 2) \\
    \mathrm{sn}_Y & \sim \mathrm{Beta}(0.5, 1, 1, 1) \\
    \mathrm{sp}_Y & \sim \mathrm{Beta}(0.5, 1, 4, 2) \\
      \boldsymbol{\rho}_{r,x} & \overset{\text{iid}}{\sim} \text{Dirichlet}((7,1,1, 1)^T), \,\forall r,x \\
      \nu_k^{u} & \sim \mathrm{Normal}(0, 1.3^2), \forall u, 1 \leq k < N_a  \\
      \gamma_k^{x} & \sim \mathrm{Normal}(0, 1.3^2), \forall x > 1, 1 \leq k < N_a 
    \end{split}
\end{align}
where $\mathrm{Beta}(0.5, 1, 4, 2)$ is the shifted, scaled Beta distribution in which the first two arguments define the support of the distribution, and the second two parameters are shape parameters. For example, for $\mathrm{sn}_S$ this corresponds to $\chi \sim \mathrm{Beta}(4, 2)$ and $\mathrm{sn}_S = (1 - 0.5) \chi + 0.5$.

When $N_z = 3$, the prior for $\boldsymbol{\rho}_{r,x}$ changes to:
$$
\boldsymbol{\rho}_{r,x} \overset{\text{iid}}{\sim} \text{Dirichlet}((2,1/3\,\mathbf{1}_{2^{N_z}-1}^T)^T), \,\forall r,x 
$$

We use Stan for inference \citep{stanmanual} using the \texttt{cmdstanr} package \citep{cmdstan} in \texttt{R} \citep{R}.
Four MCMC chains were run for each model on each scenario.
All models for the two-arm trials were run for $3{,}000$ warmup and $3{,}000$ post-warmup iterations; nearly all $\hat{R}$ statistics \citep{gelman2013bayesian} were below $1.01$, as recommended by \cite{vehtarirhat}. The maximum $\hat{R}$ recorded across all two-arm trials was $1.04$.
For the three-arm trials, all models were run for $3{,}000$ warmup and $12{,}000$ post-warmup iterations.
Nearly all $\hat{R}$ statistics were below $1.01$ with a maximum $\hat{R}$ recorded across all three-arm trials was $1.02$.

\section{Asymptotic variance derivations}

\begin{lemma}[Affine transformation of vectorized normal]\label{lemma:vec-normal}
  Suppose $X \in \R^{n \times p}$, $\text{vec}(X) \sim N(0, \Sigma)$ for a positive definite matrix $\Sigma$.
  Given a matrix $D \in \R^{k \times n}$ with $\rank D = k$ and a matrix $C \in \R^{p \times m}$ with $\rank C = m$, $\text{vec}(D X C)  \sim N(0, (C^T \otimes D)\Sigma (C \otimes D^T))$ where $(C^T \otimes D)\Sigma (C \otimes D^T)$ is positive definite.
\end{lemma}
\begin{proof}
  $\text{vec}(DXC) = (C^T \otimes D) \text{vec}(X)$.
  Affine transformations of multivariate normal vectors, $Y=B Z$, where $Z \sim N(0, \Sigma)$ and $B$ is a full-row rank matrix leads to $Y \sim N(0, B \Sigma B^T)$, $B \Sigma B^T$ positive definite.
  Then $$\text{vec}(DXC) \sim \text{N}(0,(C^T \otimes D) \Sigma (C \otimes D^T)).$$
  The rank of $C^T \otimes D = \rank C^T \rank D = m k$.
  The dimensions of $C^T \otimes D$ are $m k \times n p$, so $C^T \otimes D$ is full-row rank.
\end{proof}

\begin{lemma}[Limiting distribution of least squares estimator]\label{lemma:ls-dist}
  Let $b = X \omega$ where $X$ is full column rank with elements in $\R$ and $\omega \not = 0$.
  Given estimators $\hat{X}, \hat{b}$ the least-squares estimator for $\omega$, $\hat{\omega}$, is $\lp \hat{X}^T \hat{X} \rp^{-1}\hat{X}^T \hat{b}$.
  Let $X^+ = (X^T X)^{-1} X^T, \vc{\sqrt{N}D_X(\hat{X} - X)} \overset{d}{\to} \mathrm{Normal}(0, \Sigma),$ and $\sqrt{N}D_b(\hat{b} - b) \overset{d}{\to} \mathrm{Normal}(0, C)$ where $D_X, D_b$ are positive definite diagonal scaling matrices.
  Then $$\sqrt{N}(\hat{\omega} - \omega) \overset{d}{\to} \text{Normal}(0, (\omega^T \otimes X^+ D_X^{-1})\Sigma(\omega \otimes (X^+D_X^{-1})^T) + X^+D_b^{-1} C (X^+D_b^{-1})^T).$$
\end{lemma}
\begin{proof}
  $\omega(X,b) = \lp X^T X \rp^{-1}X^T b$, so we can linearize this function in $X$ and $b$, following the logic in \cite{bonhomme_non-parametric_2016}:
  \begin{align*}
    \omega(X + dX, b + db) & = \omega(X,b) + d \lp X^T X \rp^{-1} X^T b + \lp X^T X \rp^{-1} (d X^T) b + \lp X^T X \rp^{-1}  X^T db 
  \end{align*}
  For a square invertible matrix $\Omega$, $d \Omega^{-1} = -\Omega^{-1}\, d \Omega\, \Omega^{-1}$ so 
  \begin{align}
    d \lp X^T X \rp^{-1} & = - \lp X^T X \rp^{-1} d (X^T X) \lp X^T X \rp^{-1} \\
                         & = - \lp X^T X \rp^{-1} (d X^T X + X^T dX)\lp X^T X \rp^{-1}
  \end{align}
  Then, 
  \begin{align*}
    \omega(X + dX, b + db)  = & \omega(X,b) - \lp X^T X \rp^{-1} (d X^T X + X^T dX)\lp X^T X \rp^{-1} X^T b  \\
                              & + \lp X^T X \rp^{-1} (d X^T) b + \lp X^T X \rp^{-1}  X^T db  + o(1)\\
                              & =\omega(X,b) + \lp X^T X \rp^{-1} d X^T(I - X X^+) b + X^+ dX  X^+ b  \\
                              & + X^+ db   + o(1)
  \end{align*}
  Note that $(I - X X^+) b = (I - X X^+) X \omega = 0$ and $D_X^{-1}D_X = D_b^{-1}D_b = I$
  So
  \begin{align}
    \sqrt{N}(\hat{\omega} - \omega) & = X^+ D_X^{-1} \sqrt{N} D_X(\hat{X} - X)\,  \omega + X^+ D_b^{-1}\sqrt{N}D_b(\hat{b} - b)   + o_p(1)
  \end{align}
  By assumption,
  $$\sqrt{N}D_b(\hat{b} - b) \overset{d}{\to} N(0, C), \, \vc{\sqrt{N}D_X(\hat{X} - X)}  \overset{d}{\to} \text{Normal}(0,\Sigma),$$ so
  \begin{align}
    \vc{X^+ D_X^{-1} \sqrt{N}D(\hat{X} - X)  \,\omega} & \overset{d}{\to} \text{Normal}(0, (\omega^T \otimes X^+ D_X^{-1})\Sigma(\omega \otimes (X^+D_X^{-1})^T)) \\
    X^+ D_b^{-1}\sqrt{N}D(\hat{b} - b) & \overset{d}{\to} N(0, X^+D_b^{-1} C (X^+D_b^{-1})^T)
\end{align}
given that $\vc{X^+ D_X^{-1} \sqrt{N}D_X(\hat{X} - X)  \,\omega} = X^+ D_X^{-1}\sqrt{N}D_X(\hat{X} - X)  \,\omega$, we get
\begin{align}
  \sqrt{N}(\hat{\omega} - \omega)\overset{d}{\to} \text{Normal}(0, (\omega^T \otimes X^+ D_X^{-1})\Sigma(\omega \otimes (X^+D_X^{-1})^T) + X^+D_b^{-1} C (X^+D_b^{-1})^T)
\end{align}
\end{proof}

\subsubsection{Interpretation for VE modeling}\label{sssec:ls-dist}
When considering \Cref{lemma:ls-dist} in the context of the two-arm model presented in the paper, the vector $\omega$ is the vector
$$(\mathrm{sp}_Y, \beta_z^{(0,1)}, \beta_z^{(1,0)}, \beta_z^{(1,1)})$$, the matrix $X$ is $P_{N_z}^x(S^{P_0} \mid R)$, and the matrix $D_X$ in \Cref{lemma:ls-dist} is the square-root of the proportion of individuals assigned to a study site $R_i = r$.
$D_b$ is the proportion of people assigned to a study site $R_i = r$ and $Z_i = z$.
If the design of the study calls for equal proportions of individuals spread out between study sites and equal numbers of people assigned to treatment and control for each study site, then $D_b = (2 N_r)^{-\frac{1}{2}}\mathbf{I}_{N_r}$ and $D_X = (N_r)^{-\frac{1}{2}}\mathbf{I}_{N_r}$.
This simplifies the expression above to:
\begin{align}\label{eq:ls-dist-simple}
  \sqrt{N}(\hat{\omega} - \omega)\overset{d}{\to} \text{Normal}(0, N_r \lp (\omega^T \otimes X^+ )\Sigma(\omega \otimes (X^+)^T) + 2 X^+C (X^+)^T\rp)
\end{align}
This highlights that variance increases as a function of $N_r$ due to distributing people over more sites.
This effect is in tension with the Moore-Penrose inverse $X^+$ because if we take the view that study sites are drawn from a superpopulation of study sites with a generating distribution for proportions of principal strata, the more study sites we have, the better estimate of the inverse of the second moment matrix $X^T X$.

\section{Selection bias in $\mathrm{VE}_I^{\mathrm{net}}$}

Recall $\theta_{(i,j)}= P(S_i^{P_0} = (i,j))$ and $\beta_z^{(i,j)} = P(Y_i(z) = 1 \mid S_i^{P_0} = (i,j))$, $\theta_{(+,j)} = \theta_{(0,j)} +\theta_{(1,j)}$, and $\theta_{(i,+)} = \theta_{(i,0)} +\theta_{(i,1)}$.
For notational ease in the derivation, let the target estimand $\Exp{Y_i(0) - Y_i(1) \mid S^{P_0} = (1,1)}$ be denoted as $\Delta^{(1,1)}$, and note that $\beta_0^{(1,1)} = \beta_1^{(1,1)} + \Delta^{(1,1)}$.
\begin{align*}
  &  \Exp{Y_i(0) \mid S_i(0) = 1} - \Exp{Y_i(1) \mid S_i(1) = 1} \\
   &= \frac{P(Y_i(0) = 1, S_i(0) = 1)}{P(S_i(0) = 1)}  - \frac{P(Y_i(1) = 1, S_i(1) = 1)}{P(S_i(1) = 1)} \\
                               & = \frac{\beta_0^{(1,1)} \theta_{(1,1)} + \beta_0^{(0,1)} \theta_{(0,1)}}{\theta_{(+,1)}}  - \frac{\beta_1^{(1,1)} \theta_{(1,1)} + \beta_1^{(1,0)} \theta_{(1,0)}}{\theta_{(1,+)}} \\
                               & = \frac{1}{\theta_{(1,+)}\theta_{(+,1)}} \theta_{(1,+)}\lp \beta_0^{(1,1)} \theta_{(1,1)} + \beta_0^{(0,1)} \theta_{(0,1)}\rp - \theta_{(+,1)}\lp \beta_1^{(1,1)} \theta_{(1,1)} + \beta_1^{(1,0)} \theta_{(1,0)}\rp \\
                               & =  \frac{1}{\theta_{(1,+)}\theta_{(+,1)}}\Bigg(\theta_{(1,1)}^2 (\beta_0^{(1,1)} - \beta_1^{(1,1)})+\theta_{(1,0)}\theta_{(1,1)} \beta_0^{(1,1)} -\beta_1^{(1,1)} \theta_{(1,1)}\theta_{(0,1)}\\
                               & + \beta_0^{(0,1)} \theta_{(1,0)}\theta_{(0,1)} + \beta_0^{(0,1)} \theta_{(1,1)}\theta_{(0,1)}-\beta_1^{(1,0)} \theta_{(1,0)}\theta_{(1,1)}-\beta_1^{(1,0)} \theta_{(1,0)}\theta_{(0,1)}\Bigg) \\
                               & =  \Delta^{(1,1)} + \frac{1}{\theta_{(1,+)}\theta_{(+,1)}}\Bigg(-\theta_{(0,1)}\theta_{(1,0)} (\beta_0^{(1,1)} - \beta_1^{(1,1)})+\theta_{(1,0)}\theta_{(1,1)} \beta_1^{(1,1)} -\beta_0^{(1,1)} \theta_{(1,1)}\theta_{(0,1)}\\
  & + \beta_0^{(0,1)} \theta_{(1,0)}\theta_{(0,1)} + \beta_0^{(0,1)} \theta_{(1,1)}\theta_{(0,1)}-\beta_1^{(1,0)} \theta_{(1,0)}\theta_{(1,1)}-\beta_1^{(1,0)} \theta_{(1,0)}\theta_{(0,1)}\Bigg) \\
  & =  \Delta^{(1,1)} + \frac{1}{\theta_{(1,+)}\theta_{(+,1)}}\Bigg(\theta_{(0,1)}\theta_{(1,0)} (\beta_1^{(1,1)} - \beta_1^{(1,0)} +\beta_0^{(0,1)} - \beta_0^{(1,1)})\\
  &  +\theta_{(1,0)}\theta_{(1,1)} (\beta_1^{(1,1)}-\beta_1^{(1,0)}) + \theta_{(1,1)}\theta_{(0,1)}(\beta_0^{(0,1)}-\beta_0^{(1,1)})\Bigg) 
\end{align*}
This leads to the following expression in terms of expectations:
\begin{align*}
    \begin{split}
    \Exp{Y_i(0) \mid S_i(0) = 1} &  - \Exp{Y_i(1) \mid S_i(1) = 1}  =\\
    & \Exp{Y_i(0) - Y_i(1) \mid S^{P_0} = (1,1)}  \\
    & + (\Exp{Y_i(1) \mid S^{P_0} = (1,1)} - \Exp{Y_i(1) \mid S^{P_0} = (1,0)})\frac{\theta_{(1,1)}\theta_{(1,0)}}{\theta_{(1,+)}\theta_{(+,1)}}  \\
    & + (\Exp{Y_i(0) \mid S^{P_0} = (0,1)} - \Exp{Y_i(0) \mid S^{P_0} = (1,1)})\frac{\theta_{(1,1)}\theta_{(0,1)}}{\theta_{(1,+)}\theta_{(+,1)}} \\
    & + (\Exp{Y_i(0) \mid S^{P_0} = (0,1)} - \Exp{Y_i(0) \mid S^{P_0} = (1,1)})\frac{\theta_{(1,0)}\theta_{(0,1)}}{\theta_{(1,+)}\theta_{(+,1)}} \\
    & + (\Exp{Y_i(1) \mid S^{P_0} = (1,1)} - \Exp{Y_i(1) \mid S^{P_0} = (1,0)})\frac{\theta_{(1,0)}\theta_{(0,1)}}{\theta_{(1,+)}\theta_{(+,1)}} 
\end{split}
\end{align*}
Even under monotonicity where $\theta_{(1,0)}=0$ the selection bias remains:
\begin{align*}
  \begin{split}
    \Exp{Y_i(0) \mid S_i(0) = 1} &  - \Exp{Y_i(1) \mid S_i(1) = 1}  =\\
                                 & \Exp{Y_i(0) - Y_i(1) \mid S^{P_0} = (1,1)}  \\
                                 & + (\Exp{Y_i(0) \mid S^{P_0} = (0,1)} - \Exp{Y_i(0) \mid S^{P_0} = (1,1)})\frac{\theta_{(1,1)}\theta_{(0,1)}}{\theta_{(1,+)}\theta_{(+,1)}}.
  \end{split}
\end{align*}

\section{Asymptotic bias derivations}

Let asymptotic bias for an estimator $\hat{\phi}$ of a parameter $\phi$ be defined as $\text{AsympBias}(\hat{\phi}) = \hat{\phi} - \phi$.

\subsection{Asymptotic bias of $\mathrm{VE}_S$ under measurement error}

Let the latent binary infection state be denoted $S_i$, the noisy measurement denoted as $\tilde{S}_i$ with sensitivity $\mathrm{sn}_S$ and specificity $\mathrm{sp}_S$.
Treatment assignment is $Z_i$ with $1$ indicating the treated group while $0$ indicates the placebo group.
The target estimand is $\mathrm{VE}_S = 1 - \frac{\Exp{S_i \mid Z_i = 1}}{\Exp{S_i \mid Z_i = 0}}$ which equals the causal estimand under randomization and consistency.
We investigate the asymptotic bias of the naive estimator:
$$
\hat{\mathrm{VE}}_S = 1 - \frac{\sum_i \ind{\tilde{S}_i = 1} \ind{Z_i = 1} / \sum_i \ind{Z_i = 1}}{\sum_i \ind{\tilde{S}_i = 1} \ind{Z_i = 0} / \sum_i \ind{Z_i = 0}}.
$$
$\hat{\mathrm{VE}}_S$ converges in probability to the following estimand:
$$
\tilde{\mathrm{VE}}_S = 1 - \frac{\Exp{\tilde{S}_i \mid Z_i = 1}}{\Exp{\tilde{S}_i \mid Z_i = 0}}.
$$
Recall the identity for binary variables, and $z \in \{0,1\}$:
\begin{align}\label{eq:ident}
  \Exp{S_i \mid Z_i = z} = \frac{\Exp{\tilde{S}_i \mid Z_i = z} - (1 - \mathrm{sp}_S)}{\mathrm{sn}_S + \mathrm{sp}_Y - 1}
\end{align}
to express the target estimand in terms of the noisy infection measurements:
$$
\mathrm{VE}_S = \frac{\Exp{\tilde{S}_i \mid Z_i = 1} - (1 - \mathrm{sp}_S)}{\Exp{\tilde{S}_i \mid Z_i = 0} - (1 - \mathrm{sp}_S)}
$$
Then the bias is:
\begin{align}
  \text{AsympBias}(\hat{\mathrm{VE}}_I) & = \frac{\Exp{\tilde{S}_i \mid Z_i = 1} - (1 - \mathrm{sp}_S)}{\Exp{\tilde{S}_i \mid Z_i = 0} - (1 - \mathrm{sp}_S)} -
                                          \frac{\Exp{\tilde{S}_i \mid Z_i = 1}}{\Exp{\tilde{S}_i \mid Z_i = 0}} \\
                                        & = -\frac{1 - \mathrm{sp}_S}{\Exp{\tilde{S}_i \mid Z_i = 0}}\frac{(\Exp{\tilde{S}_i \mid Z_i = 0} - \Exp{\tilde{S}_i \mid Z_i = 1})}{\Exp{\tilde{S}_i \mid Z_i = 0} - (1 - \mathrm{sp}_S)} \label{eq:fin-line}
\end{align}
Again using the identity \cref{eq:ident} in the numerator and denominator of the second fraction in line \cref{eq:fin-line}, the bias can be expressed as:
$$
\text{AsympBias}(\hat{\mathrm{VE}}_I) = -\frac{1-\mathrm{sp}_S}{\Exp{\tilde{S}_i \mid Z_i = 0}}\mathrm{VE}_I.
$$

\subsection{Bias of $\mathrm{VE}_I$ under erroneous monotonocity assumption}

Recall the definition of $\mathrm{VE}_I$:
$$
\mathrm{VE}_I = 1 - \frac{P(Y_i(1) = 1 \mid S^{P_0} = (1, 1))}{P(Y_i(0) = 1 \mid S^{P_0} = (1, 1))}.
$$
The sample estimator, $\hat{p}_{1+zr} = \frac{\sum_{i}\ind{R_i = r}\ind{S_i = 1}\ind{Z_i = z}}{\sum_{i}\ind{R_i = r}\ind{Z_i = z}}$ converges in probability to $P(S_i = 1 \mid Z_i = z, R_i = r)$ for $z\in\{0,1\}$.
Under \crefrange{cond:SUTVA}{cond:pc-obs-unconfound-multi-z} and , $P(S_i = 1 \mid Z_i = z, R_i = r) = P(S_i(z) = 1 \mid R_i = r)$.
Under monotonicity, $P(S_i(1) = 1, S_i(0)=0 \mid R_i = r) = 0$ so $P(S_i(1) = 1 \mid R_i = r) = P(S_i(1) = 1, S_i(0)=1 \mid R_i = r)$ and $\hat{\theta}_{11}^r = \hat{p}_{1+zr} \overset{\mathrm{P}}{\to} \theta_{11}^r$.
Futhermore, $\hat{p}_{11zr} = \frac{\sum_{i}\ind{R_i = r}\ind{S_i = 1}\ind{Y_i = z}\ind{Z_i = 1}}{\sum_{i}\ind{R_i = r}\ind{Z_i = z}}$ converges in probability to $P(Y_i = 1, S_i = 1 \mid Z_i = z, R_i = r)$.
For $z = 1$, this is equivalent to:
\begin{align}
  P(Y_i = 1,  S_i = 1 \mid Z_i = 1, R_i = r) & = P(Y_i(1) = 1, S_i(1) = 1 \mid R_i = r) \\
                                            & = P(Y_i(1) = 1 \mid S_i(1) = 1, R_i = r)P(S_i(1) = 1 \mid R_i = r) \\
                                            & = P(Y_i(1) = 1 \mid S_i^{P_0} = (1, 1), R_i = r)P(S_i^{P_0} = (1, 1) \mid R_i = r)  \\
                                             & = P(Y_i(1) = 1 \mid S_i^{P_0} = (1, 1))P(S_i^{P_0} = (1, 1) \mid R_i = r) 
\end{align}
Where the first line follows from \crefrange{cond:SUTVA}{cond:pc-obs-unconfound-multi-z}, the third line follows from the equivalence of $\{S_i(1) = 1\}$ and $\{S_i^{P_0} = (1, 1)\}$, and the fourth line follows from \cref{cond:homogeneity}.
Then it follows that a consistent estimator for $P(Y_i(1) = 1 \mid S_i^{P_0} = (1, 1))$ is .
$$
\frac{\hat{p}_{111r}}{\hat{p}_{1+1r}} \overset{\mathrm{P}}{\to}P(Y_i(1) = 1 \mid S_i^{P_0} = (1, 1))
$$
provided $\hat{p}_{1+1r}$ is bounded away from zero.
Let $P(S_i(1) = j, S_i(0)=k \mid R_i = r) = \theta_{jk}^r$ for $j, k \in \{0,1\}$, and let $\beta_z^u = P(Y_i(z) = 1 \mid S_i^{P_0} = u)$.
When monotonicity does not hold, 
$$
P(S_i(1) = 1 \mid R_i = r) = \theta_{11}^r + \theta_{10}^r,
$$
and 
\begin{align*}
  P(Y_i(1) = 1, S_i(1) = 1 \mid R_i = r)  = \theta_{11}^r\beta_1^{(1,1)} + \theta_{10}^r\beta_1^{(1,0)}
\end{align*}
Following the notation introduced at the beginning of the Appendix, the asymptotic limit of our simple moment estimator for the proportion is:
\begin{align}
  \frac{\hat{p}_{111r}}{\hat{p}_{1+1r}} \overset{\mathrm{P}}{\to} \frac{\theta_{11}^r \beta_1^{(1,1)} + \theta_{10}^r \beta_1^{(1,0)}}{\theta_{11}^r + \theta_{10}^r}\label{eq:asymp-treat}
  \end{align}

Under monotonicity, a consistent estimator for $\theta_{01}^r$ is $\hat{\theta}_{01}^r = \hat{p}_{1+0r} - \hat{p}_{1+1r}$.
One may construct a consistent estimator for the vector $(\beta_0^{(0,1)}, \beta_0^{(1,1)})^T$ by recognizing that the following system of equations can be solved:
\begin{align}
  p_{110{r_1}} & = \theta_{01}^{r_1} \beta_0^{(0,1)} + \theta_{11}^{r_1} \beta_0^{(1,1)} \\
  p_{110{r_2}} & = \theta_{01}^{r_2} \beta_0^{(0,1)} + \theta_{11}^{r_2} \beta_0^{(1,1)} 
\end{align}
provided the matrix
$$
\begin{bmatrix}
  \theta_{01}^{r_1}  &\theta_{11}^{r_1} \\
  \theta_{01}^{r_2}  &\theta_{11}^{r_2} 
  \end{bmatrix}
$$
is invertible.
Then the following holds:
$$
\begin{bmatrix}
  \hat{\theta}_{01}^{r_1}  &\hat{\theta}_{11}^{r_1} \\
  \hat{\theta}_{01}^{r_2}  &\hat{\theta}_{11}^{r_2} 
\end{bmatrix}^{-1} \begin{bmatrix}
                     \hat{p}_{110{r_1}} \\
                     \hat{p}_{110{r_2}}
                   \end{bmatrix} \overset{\mathrm{P}}{\to}
                   \begin{bmatrix}
                     \beta_0^{(0,1)}\\
                     \beta_0^{(1,1)}
                     \end{bmatrix}.
$$

When monotonicity does not hold, however,
$$
\hat{\theta}_{01}^r = \hat{p}_{1+0r} - \hat{p}_{1+1r} \overset{\mathrm{P}}{\to}\theta_{01}^r - \theta_{10}^r.
$$
Furthermore, 
$$
\begin{bmatrix}
  \hat{\theta}_{01}^{r_1}  &\hat{\theta}_{11}^{r_1} \\
  \hat{\theta}_{01}^{r_2}  &\hat{\theta}_{11}^{r_2} 
\end{bmatrix}^{-1}
\begin{bmatrix}
  \hat{p}_{110{r_1}} \\
  \hat{p}_{110{r_2}}
\end{bmatrix} \overset{\mathrm{P}}{\to}
\begin{bmatrix}
  \theta_{01}^{r_1} - \theta_{10}^{r_1}  & \theta_{11}^{r_1} + \theta_{10}^{r_1}\\
  \theta_{01}^{r_2} - \theta_{10}^{r_2}  & \theta_{11}^{r_2} + \theta_{10}^{r_2}
\end{bmatrix}^{-1}
\begin{bmatrix}
  \theta_{01}^{r_1} \beta_0^{(0,1)} + \theta_{11}^{r_1} \beta_0^{(1,1)} \\
  \theta_{01}^{r_2} \beta_0^{(0,1)} + \theta_{11}^{r_2} \beta_0^{(1,1)} 
\end{bmatrix}.
$$
Let
$$\mathbf{U}_{{r_1},{r_2}} = \begin{bmatrix} \theta_{01}^{r_1} & \theta_{01}^{r_2} \\ \theta_{10}^{r_1} & \theta_{10}^{r_2}\end{bmatrix},\,
\mathbf{V}_{{r_1},{r_2}} = \begin{bmatrix} \theta_{11}^{r_2} & \theta_{01}^{r_2} - \theta_{10}^{r_2} \\ \theta_{11}^{r_1} & \theta_{01}^{r_1} - \theta_{10}^{r_1}\end{bmatrix}.
$$
Then the moment estimator for $\beta_0^{(1,1)}$ converges in probability to:
\begin{align}\label{eq:asymp-plac}
  \frac{\deter{\mathbf{U}_{{r_1},{r_2}}}\beta_0^{(0,1)} + \deter{\mathbf{V}_{{r_1},{r_2}}}\beta_0^{(1,1)}}{\deter{\mathbf{U}_{{r_1},{r_2}}} + \deter{\mathbf{V}_{{r_1},{r_2}}}}.
\end{align}
The asymptotic bias for the moment estimator of $\beta_0^{(1,1)}$ is
\begin{align}
  \text{AsympBias}(\hat{\beta}_0^{(1,1)}) = \frac{\deter{\mathbf{U}_{{r_1},{r_2}}}}{\deter{\mathbf{U}_{{r_1},{r_2}}} + \deter{\mathbf{V}_{{r_1},{r_2}}}}(\beta_0^{(0,1)} - \beta_0^{(1,1)})
\end{align}
while the asymptotic bias for the moment estimator of $\beta_1^{(1,1)}$ is
\begin{align}
  \text{AsympBias}(\hat{\beta}_1^{(1,1)}) = \frac{\theta_{10}^{r_3}}{\theta_{11}^{r_3} + \theta_{10}^{r_3}}(\beta_1^{(1,0)} - \beta_1^{(1,1)})
\end{align}
The superscript $r_3 \in \{1, \dots, N_r\}$, while, of course, $r_1 \neq r_2$. 
We can express our asymptotic bias for our plug-in estimator for $\mathrm{VE}_I$ as
$$
\frac{\beta_1^{(1,1)}}{\beta_0^{(1,1)}} - \frac{\beta_1^{(1,1)} +\text{AsympBias}(\hat{\beta}_1^{(1,1)}) }{\beta_0^{(1,1)} + \text{AsympBias}(\hat{\beta}_0^{(1,1)})}.
$$
This leads to a tidy expression for the asymptotic bias in our ratio estimator:
$$
\frac{\beta_1^{(1,1)}\frac{\deter{\mathbf{U}_{{r_1},{r_2}}}}{\deter{\mathbf{U}_{{r_1},{r_2}}} + \deter{\mathbf{V}_{{r_1},{r_2}}}}(\beta_0^{(0,1)} - \beta_0^{(1,1)}) - \beta_0^{(1,1)}\frac{\theta_{10}^{r_3}}{\theta_{11}^{r_3} + \theta_{10}^{r_3}}(\beta_1^{(1,0)} - \beta_1^{(1,1)})}{\beta_0^{(1,1)}\hat{\beta}_0^{(1,1)}}
$$
The expression shows that the bias can be quite large in magnitude if $\beta_0^{(1,1)}$ is small and there is selection bias in the sense that $\beta_1^{(1,0)} > \beta_1^{(1,1)}$.
This bias can be mitigated by a countervailing selection bias in the placebo group, if $\beta_0^{(0,1)} < \beta_0^{(1,1)}$, provided the ratio of determinants is positive, which is not guaranteed.

\end{appendices}

\bibliography{ve.bib}
\end{document}